\documentclass[12pt,reqno]{amsart}
\usepackage[dvips]{color}
\usepackage{amsmath}
\usepackage{amsxtra}
\usepackage{amscd}
\usepackage{amsthm}
\usepackage{amsfonts}
\usepackage{amssymb}
\usepackage{eucal}
\usepackage{epsfig}
\usepackage{graphics}
\usepackage{dirtytalk}
\usepackage{ulem}
\usepackage[breakable]{tcolorbox}

\usepackage{tikz}
\usetikzlibrary{arrows}
\newcommand{\midarrow}{\tikz \draw[-triangle 90] (0,0) -- +(.1,0);}

\usepackage{graphicx} 
\graphicspath{ {./photoofqt/} }

\numberwithin{equation}{section}
\newtheorem{thm}{Theorem}[section]

\newtheorem{lem}[thm]{Lemma}
\newtheorem{rem}[thm]{Remark}

\newtheorem{dfn}[thm]{Definition}

\newcommand{\nn}{\nonumber}

\newcommand{\End}{\mathop{\rm End}}

\newcommand{\gq}{\mathfrak{q}}
\newcommand{\gd}{\mathfrak{d}}
\newcommand{\ind}{\varinjlim}

\newcommand{\tens}{\otimes}

\newcommand{\inlim}[1]{\mathop{\lim_{\longrightarrow}}_{#1}}

\newcommand{\cals}[1]{\mathcal{#1}}

\newcommand{\bb}[1]{\mathbb{#1}}

\newcommand{\comp}{\circ}

\begin{document}

	\begin{titlepage}
		
		\bigskip
		\begin{center}
			\Large{ \bf
				MacMahon KZ equation for Ding-Iohara-Miki algebra
			}
		\end{center}
		\bigskip
		\bigskip
		
		\begin{center}
			\large 
			Panupong Cheewaphutthisakun$^{a}$\footnote{panupong.cheewaphutthisakun@gmail.com}
			and
			Hiroaki Kanno$^{a,b,}$\footnote{kanno@math.nagoya-u.ac.jp} \\
			\bigskip
			\bigskip
			$^a${\small {\it Graduate School of Mathematics, Nagoya University,
					Nagoya, 464-8602, Japan}}\\
			$^b${\small {\it KMI, Nagoya University,
					Nagoya, 464-8602, Japan}} \\
		\end{center}
		{\vskip 6cm}
		{\small
			\begin{quote}
				\noindent {\textbf{\textit{Abstract}}.}
				We derive a generalized Knizhnik-Zamolodchikov equation for the correlation function of the intertwiners of
				the vector and the MacMahon representations of Ding-Iohara-Miki algebra. These intertwiners are cousins of 
				the refined topological vertex which is regarded as the intertwining operator of the Fock representation. 
				The shift of the spectral parameter of the intertwiners is generated by the operator
				which is constructed from the universal $R$ matrix. The solutions to the generalized KZ equation are
				factorized into the ratio of two point functions which are identified with generalizations of the Nekrasov
				factor for supersymmetric quiver gauge theories. 
		\end{quote}}

	\end{titlepage}
	
	%
	%
	%
	
	\setcounter{footnote}{0}
	\baselineskip=18pt
	
	
	\section{Introduction}

	Quantum Knizhnik-Zamolodchikov (q-KZ) equation is originally introduced for 
	the correlation function of the intertwiners of the quantum affine algebra \cite{FR} \cite{S} \cite{JM} \cite{EFK}.
	In \cite{qtKZdim} we have shown that one can develop a parallel story for
	Ding-Iohara-Miki (DIM) algebra \cite{DI} \cite{M}, which is regarded as 
	the quantum toroidal algebra of $\mathfrak{gl}_1$.
	The trivalent intertwining operator of the Fock representations of DIM algebra
	agrees with what is called the refined topological vertex \cite{AK2005} \cite{AK2008} \cite{IKV} \cite{AFS}. 
	Consequently, the correlation function of the intertwiners gives a building block of 
	the Nekrasov partition functions of the five dimensional supersymmetric 
	quiver gauge theories in accord with the AGT correspondence.

	Our main interest in the generalized q-KZ equation comes from
	the AGT correspondence and its generalizations \cite{AGT} \cite{MM} \cite{W} \cite{AY}, 
	which tell us that the conformal blocks of two dimensional conformal field theories 
	and the instanton partition functions, and hence the low-energy effective actions 
	of the supersymmetric gauge theories are related. This correspondence allows us 
	to study problems on one side from the perspective of the other side. 
	For example, the modular properties of the six dimensional  Seiberg-Witten theory with adjoint hypermultiplet
	can be explained by using the elliptic Knizhnik-Zamolodchikov equation of DIM algebra \cite{AKMMSZ2018}. 
	
	In the paper \cite{qtKZdim}, the authors gave a general method of deriving the generalized KZ equation 
	for the Fock intertwiners, and wrote out explicit solutions to the equation. 
	The general solutions are written in terms of a product of the propagators of
	the Fock intertwiners (the refined topological vertex). 
	Moreover, as expected from the AGT correspondence, the solution relates to the Nekrasov function
	which gives the instanton partition function. 
	We know from the papers \cite{FFJMM1} \cite{FJMM11}
	that the MacMahon modules are, in a sense, a generalization of the Fock modules. 
	More precisely, the MacMahon modules can be constructed as the inductive limit of an inductive system whose objects are Fock modules. 
	Moreover, we are able to express the MacMahon intertwiner in terms of the Fock intertwiners  \cite{MacMahon}. 
	Thus, it is natural to ask whether we can generalize the story in \cite{qtKZdim} to derive the generalized KZ 
	equation for the MacMahon representation of DIM algebra. In this paper we will show this is indeed the case.

	Let $\Psi_\lambda(u)$ be the intertwiner of the Fock representation of DIM algebra, where $\lambda$ is a Young diagram
	which labels a basis of the Fock representation and $u$ is the spectral parameter. 
	The $(q,t)$-KZ equation derived in \cite{qtKZdim} implies that the correlation function is expressed 
	as a product of two point functions like the Wick theorem for free fields. 
	One can check that the inverse of the two point function agrees with
	the Nekrasov factor which is a building block of five dimensional quiver gauge
	theories with the equivariant parameters $(q_1, q_2)= (e^{\epsilon_1}, e^{\epsilon_2})$, \cite{N} \cite{Nakajima} \cite{NY} \cite{FP}.
	Namely, if we define
	\begin{equation}
		N_{\lambda \mu}(u,v) 
		= \langle \Psi_\lambda(u) \Psi_\mu(v)  \rangle^{-1}, 
	\end{equation}
	then we can show that
	\begin{align}
		N_{\lambda \mu}(u,v) 
		&= \prod_{\genfrac{}{}{0pt}{}{\square \in \lambda}{\blacksquare \in \mu}} 
		\frac {\left(1-q_1 \frac{\chi_\square(u)}{\chi_\blacksquare(v)}\right)
			\left(1- q_2 \frac{\chi_\square(u)}{\chi_\blacksquare(v)}\right)}
		{\left(1-\frac{\chi_\square(u)}{\chi_\blacksquare(v)}\right)\left(1- q_1 q_2 \frac{\chi_\square(u)}{\chi_\blacksquare(v)}\right)}
		\cdot \prod_{\square \in \lambda}\left( 1- q_1 q_2 \frac{\chi_\square(u)}{v} \right) 
		\cdot \prod_{\blacksquare \in \mu}\left( 1- \frac{u}{\chi_\blacksquare(v)} \right) \nonumber \\
		&= \prod_{\square \in \lambda} \left(1- q_1^{-\ell_\mu(\square)} q_2^{a_\lambda(\square)+1}\frac{u}{v}\right)
		\prod_{\blacksquare \in \mu} \left( 1- q_1^{\ell_\lambda(\blacksquare)+1} q_2^{-a_\mu(\blacksquare)}\frac{u}{v} \right),
		\label{1.2}
	\end{align}
	where for $\square=(i,j) \in \lambda$
	\begin{equation}
		\chi_\square(u) := u \cdot q_1^{i-1} q_2^{j-1},
	\end{equation}
	and 
	\begin{equation}
		a_\mu(\square) = \mu_i - j, \qquad \ell_\mu(\square) = \mu^\vee_j -i.
	\end{equation}

	\subsection{Strategy of deriving q-KZ equation}
	
	Let $\Psi_\alpha(z)$ be the component of the DIM intertwiner in general. Namely $\alpha$ stands for 1d, 2d and 3d Young diagrams 
	for the vector, Fock and MacMahon representations, respectively. Our first task is to construct the shift operator 
	which generates the shift of the spectral parameter;
	\begin{equation}
		p^{z \partial_z} \Psi_\alpha(z) = \mathcal{T}_\alpha^{-}(\gq^{-3/2} z)  \cdot \Psi_\alpha(z) \cdot \mathcal{T}_\alpha^{+}(\gq^{-1/2} z).
	\end{equation}
	This is a quantum version of the Sugawara construction $L_{-1} \sim \sum_{n} :J_{n} J_{-n-1}:$ for the classical KZ 
	equation. One of the technical problems in deriving the generalized KZ equation is the construction of such a shift operator.
	In \cite{qtKZdim} and \cite{AKMMSZ2018}, the shift operator is identified with the composition of 
	the intertwiner and the dual intertwiner  by tuning their spectral parameters appropriately. 
	Remarkably the shift parameter is fixed to be $p=\gq^{-2}= q_1 q_2$. It turns out that up to the normalization,
	the same shift operator is obtained in terms of the universal $R$ matrix $\cals{R}_0$ of DIM algebra (see section \ref{universal});
	\begin{equation}
		\mathcal{T}^{+} = [\rho^V \otimes \rho^H] (\cals{R}_0), \qquad
		\mathcal{T}^{-} = [\rho^H \otimes \rho^V] (\cals{R}_0),
	\end{equation}
	where $\rho^V$ and $\rho^H$ are the vertical and the horizontal representations, respectively (see section~\ref{section3}). 
	
	The second ingredient is the commutation relations of the intertwiners and the shift operator;
	\begin{equation}\label{comR}
		\mathcal{T}_\alpha^{\pm}(z) \Psi_{\beta}(w) = \mathcal{R}_{\beta \alpha}\left( \gq^{n} \frac{z}{w} \right) 
		\Psi_{\beta}(w)  \mathcal{T}_\alpha^{\pm}(z),
	\end{equation}
	where $ \mathcal{R}_{\beta \alpha}$ is the (diagonal) $R$ matrix of DIM algebra and 
	the power $n$ of $\gq$ depends on $\mathcal{T}_\alpha^{\pm}$. We also have a similar commutation relation
	for the dual intertwiner $\Psi_{\beta}^{*}(w)$, where we should change the power $n$ appropriately. 
	The relation \eqref{comR} replaces  the commutation (OPE) relations of the vertex operators with the current algebra
	in the classical case. The $R$ matrix appearing in \eqref{comR} can be identified with
	\begin{equation}
		\mathcal{R} = [\rho^{V_1} \otimes \rho^{V_2}] (\cals{R}_0),
	\end{equation}
	which implies that $\mathcal{R}_{\beta \alpha}$ appears in the commutation relation 
	of the intertwiners \cite{anomaly} \cite{MacMahon}. 
	If we define the shift operator by the composition of $\Psi$ and $\Psi^*$,
	\eqref{comR} follows from this property.
	On the other hand, the commutation relation \eqref{comR} also comes from the fact 
	that both the shift operator and the $R$ matrix are derived from the same object; the universal $R$ matrix.

	After all these are prepared, the shift of the spectral parameter of the intertwiners is 
	achieved by an insertion of a pair of the shift operators $\mathcal{T}_\alpha^{\pm}(z)$. Then we can use the commutation 
	relations to move them to the left or the right most position of the strip of the intertwiners, where they act on the vacuum state.
	In this precess a product of $R$ matrices is produced. 
	Since the shift operators $\mathcal{T}_\alpha^{\pm}(z)$ have only the positive or the negative modes of free bosons,
	the vacuum is their eigenstate and consequently we obtain a difference equation for the correlation function of the intertwiners,
	which we identify with a generalized KZ equation for DIM algebra.

	\subsection{Solutions to the generalized KZ equation}
	
	For simplicity, let us assume the correlation function does not involve the dual intetwiners\footnote{For the case with 
		the dual intertwiners, see the main text for detail.}. In this case 
	our generalized  KZ equation for DIM algebra canonically takes the following form
	\begin{equation}\label{DIMKZ}
		\langle \Psi_{\alpha_1}(z_1) \cdots \Psi_{\alpha_k}(pz_k) \cdots \Psi_{\alpha_n}(z_n) \rangle
		= A_k \cdot \langle \Psi_{\alpha_1}(z_1) \cdots \Psi_{\alpha_k}(z_k) \cdots \Psi_{\alpha_n}(z_n) \rangle,
	\end{equation}
	where
	\begin{equation}
		A_k = \prod_{i < k}R_{\alpha_i \alpha_k} (p^{-1} z_i/z_k)^{-1} \cdot \prod_{k<j} R_{\alpha_k \alpha_j} (z_k/z_j).
	\end{equation}
	As mentioned before the shift parameter is $p = \gq^{-2} = q_1q_2$, which is fixed by the choice of 
	the (horizontal) Fock space. This is schematically the same form as the q-KZ equation for the quantum affine algebra
	\cite{FR} \cite{S} \cite{JM} \cite{EFK}.
	Note that, contrary to the case of the quantum affine algebra,  the $R$ matrices appearing in Eq.\eqref{DIMKZ} 
	is diagonal with respect to the label $\alpha$ of basis of the vertical (or evaluation) representation
	and their ordering does not matter. Due to such an abelian nature of the $R$ matrix, the solutions to our KZ equation
	are factorized into the ratio of fundamental building blocks which are given by two point functions.
	We find the two point functions of MacMahon KZ equation can be regarded as a generalization of the Nakrasov factor
	\eqref{1.2}.

	\subsection{Organization of material}
	
	The paper is organized as follows; Sections 2 -- 4 are preliminaries. 
	We start by reviewing the definition of DIM algebra and some of its properties in section \ref{section2}. 
	We introduce the universal $R$ matrix of DIM algebra. 
	Then, in section \ref{section3}, we concern with the representations of DIM algebra. 
	The coproduct of DIM algebra is crucial for defining the representation by the tensor product. 
	We show that we can construct the vertical Fock representation from the vector representations 
	which are the simplest vertical representation of DIM algebra. By applying a similar method, we can construct 
	the so-called MacMahon representation from the Fock representation. We also introduce the horizontal Fock representation 
	by using the deformed Heisenberg algebra. 
	In section \ref{intertwiners} we give explicit expressions of the trivalent intertwiners,
	where vertical representations are vector, Fock, and MacMahon representations, respectively. 
	We also discuss the dual intertwiners. 
	
	Before embarking on the task of derivation of the MacMahon KZ equation, 
	we discuss the vector KZ equation in section \ref{section6},
	since basic ideas are well-illustrated in a simplified setting. 
	Following the method of the paper \cite{qtKZdim}, we construct the shift operator $ \cals{T} $ 
	which plays an important role in the derivation. We also show 
	an alternative way to construct the shift operator $ \cals{T} $ from the universal $ R $-matrix,
	which is useful in our construction of the shift operator in MacMahon case. 
	
	Sections \ref{section8} and \ref{section9} are the main part of the paper. 
	In section \ref{section8} we use the method discussed in section \ref{section6} to construct the shift operator 
	and provide a derivation of the MacMahon KZ equation. 
	Finally, in section \ref{section9}, we solve the MacMahon KZ equation and show that the solution can be regarded as a generalized Nekrasov function. 
	Some of the techinical computation are presented in Appendix.

	\subsection{Definitions and useful formulas}
	
	Before ending the introduction, let us introduce the theta function and also a well-known lemma. 
	
	The theta function $ \theta_p(z) $ is defined by 
	\begin{align}\label{theta}
		\theta_p(z) :=  (p;p)_{\infty}(z;p)_{\infty} (pz^{-1};p)_{\infty}
		= (1-z) \prod_{k=1}^{\infty}(1-p^{k})(1-p^k z)(1-p^k z^{-1}),  
	\end{align}
	where we use the infinite product;
	\begin{equation}\label{q-factorial}
		\big(	x ; p	\big)_{\infty} := \prod\limits_{j\geq 0}(1-p^jx)  
		= \exp \left(- \sum_{k=1}^{\infty} \frac{x^k}{k(1-p^k)} \right).
	\end{equation}
	It is easy to check
	\begin{equation}\label{p-shift}
		\theta_p(p^n z) = (-z)^n p^{-\frac{1}{2} n(n-1)} \theta_p(z).
	\end{equation}
	
	Finally, we state a particular case of the Campbell-Baker-Hausdorff (CBH) formula.  
	\begin{lem}\label{L1}
		If $ [A,B] $ is central, i.e. $ [[A,B],A] = [[A,B],B] = 0 $, then 
		\begin{equation}\label{CBH}
			e^Ae^B = e^Be^Ae^{[A,B]}.
		\end{equation}
	\end{lem}
	This formula is repeatedly used for computing the commutation relations of intertwining operators or the vertex operators.
	
	
	\section{Ding-Iohara-Miki algebra}
	\label{section2}
	In this section we provide a quick review of the definition of Ding-Iohara-Miki (DIM) algebra,
	which is the quantum toroidal algebra of $ \mathfrak{gl}_1$.
	The material in this and the next section is based on the papers \cite{FFJMM1}  \cite{MacMahon}. 
	
	\subsection{Definition of DIM algebra}
	The intriguing triality of DIM algebra becomes manifest by using the parameters $(q_1, q_2, q_3)$ with $q_1q_2q_3=1$.
	We assume they are generic in the sense that for any $ a, b, c \in \mathbb{Z} $,
	\begin{gather}
		q^a_1q^b_2q^c_3 = 1 \implies a = b =c .
	\end{gather}
	We may parametrize $(q_1, q_2, q_3)$ by $ \gq, \gd \in \mathbb{C} $ as follows;
	\begin{align}
		q_1 = \frac{\gd}{\gq},
		\quad
		q_2 = \frac{1}{\gd\gq},
		\quad 
		q_3 = \gq^2.
	\end{align}
	But note that this parametrization breaks the triality of DIM algebra. 
	We define the DIM algebra $U_{\gq,\gd}(\widehat{\widehat{\mathfrak{gl}}}_1)$ 
	to be the associative algebra with the generators 
	$ E_k, F_k, K_0^{\pm}, H_r \,\, ( k \in \mathbb{Z}, r \in \mathbb{Z}\backslash\{	0	\})$ 
	and $ C $. Introducing the generating functions (currents)\footnote{
		Compared with \cite{MacMahon}, the normalization of $H_{\pm r}$ is changed by the factor $(\gq - \gq^{-1})$.}; 
	\begin{align}
		E(z) = \sum_{k\in\mathbb{Z}} E_k z^{-k}, \quad
		F(z) = \sum_{k\in\mathbb{Z}} F_k z^{-k}, \quad
		K^{\pm}(z) = K_0^{\pm}\exp\left(
		\pm \sum_{r=1}^{\infty} H_{\pm r} z^{\mp r}
		\right), 
		\label{2.12}
	\end{align}
	and the structure function
	\begin{align}
		g(z,w) &:= (z-q_1w)(z-q_2w)(z-q_3w),
	\end{align}
	we can write the following defining relations : 
	\begin{align}
		C &\text{ is central}, \\
		K_0^{+}K_0^{-} &= 1 = K_0^{-}K_0^{+}, \\
		K^{\pm}(z) K^{\pm}(w) &= K^{\pm}(w) K^{\pm}(z), \label{KK1} \\
		\frac{g(C^{-1}z,w)}{g(Cz,w)} ~K^{-}(z) K^{+}(w) &= \frac{g(w,C^{-1}z)}{g(w,Cz)} ~K^{+}(w) K^{-}(z) \label{KK2}, \\
		g(z,w) ~K^{\pm} (C^{(1\mp 1)/2} z) E(w) &+ g(w,z) ~E(w) K^{\pm} (C^{(1\mp 1)/2} z) = 0, \\
		g(w,z) ~K^{\pm} (C^{(1\pm 1)/2} z) F(w) &+ g(z,w) ~F(w) K^{\pm} (C^{(1\pm 1)/2} z) = 0, \\
		\left[ E(z), F(w) \right] = \tilde{g} &\left( \delta (C \frac{w}{z}) ~K^{+}(z)  -  \delta (C \frac{z}{w}) ~K^{-}(w) \right), \label{EFcoef} \\
		g(z,w) ~E(z) E(w) &+ g(w,z) ~E(w) E(z) = 0, \\
		g(w,z) ~F(z) F(w) &+ g(z,w) ~F(w) F(z) =0,
	\end{align}
	where the multiplicative delta function is defined by 
	\begin{equation}
		\delta(z) = \sum_{n \in \mathbb{Z}} z^{n}.
	\end{equation}
	It is convenient to employ the notation
	\begin{align}
		\kappa_n &:= \prod_{i=1}^3 (q_i^{\frac{n}{2}} - q_i^{-\frac{n}{2}}) = 
		\prod_{i=1}^3 (q_i^n -1) = \prod_{i=1}^3 (1- q_i^{-n}) = \sum_{i=1}^3 (q_i^n - q_i^{-n}),
	\end{align}
	which satisfies $\kappa_{-n} = - \kappa_n$. 
	We choose the normalization of \eqref{EFcoef} as $\tilde{g} = \kappa_1^{-1}$.

	\begin{rem}
		One should bear in mind that actually there are the Serre's relations in the defining relations. 
		Since we do not use the Serre's relations in this paper, we do not write it out. 
	\end{rem}
	
	By investigating the defining relations we see that $ K_0^{\pm} $ is central\footnote{
		In \cite{FJMM} $ K_0^{-} $ is denoted by $C^\perp$.}. 
	Thus we conclude that the DIM algebra has two-dimensional center $(C,  K_0^{\pm})$. 
	Note that $K_0^{+}$ is the inverse of $ K_0^{-} $. We also see that \eqref{KK2} implies 
	\begin{align}
		[H_r, H_s]
		&= \delta_{r+s,0} \frac{\kappa_r}{r}(C^{r}-C^{-r}).
		\label{1.16} 
	\end{align}
	
	\subsection{Coproduct}
	
	The DIM algebra has a \say{coproduct} 
	$ \Delta : 	U_{\gq,\gd}(\widehat{\widehat{\mathfrak{gl}}}_1) \rightarrow
	U_{\gq,\gd}(\widehat{\widehat{\mathfrak{gl}}}_1) \otimes U_{\gq,\gd}(\widehat{\widehat{\mathfrak{gl}}}_1)	 $ defined by 
	\begin{align}
		\Delta(E(z)) &= E(z) \otimes 1 + K^{-} (C_1z) \otimes E(C_1z),  \label{cpE} \\
		\Delta(F(z)) &=  F(C_2 z) \otimes K^{+}(C_2z)  +  1 \otimes F(z), \label{cpF} \\
		\Delta(K^{+} (z)) &= K^{+}(z) \otimes K^{+}(C_1^{-1} z), \label{cpK+} \\
		\Delta(K^{-} (z)) &= K^{-}(C_2^{-1} z) \otimes K^{-}(z), \label{cpK-} \\
		\Delta(C) &= C \otimes C,
	\end{align}
	where $C_1 = C \otimes 1$ and $C_2 = 1 \otimes C$.
	
	We will see that the coproduct structure plays an important role in defining the tensor product representation of DIM. 
	However, we would like to give a remark here that $ \Delta $ is not a coproduct in the strict sense. 
	The reason is that it contains an infinite summation of elements which is not defined in general. 
	Hence, when we use $ \Delta $ for defining the action of tensor product representation, 
	we have to check every time that the action is well-defined. 
	
	\subsection{Grading operators}
	
	The DIM algebra has a bi-grading defined by 
	two grading operators $ d_1, d_2 $ which satisfy
	\begin{gather}
		\label{eq:22}
		[d_1 , E(z)] = -E(z),\qquad [d_1, F(z)] = F(z),\qquad [d_1, H(z)] = 0,\\
		[d_2, E(z)] = z \partial_z E(z),\qquad   [d_2, F(z)] = z \partial_z F(z),
		\qquad   [d_2, H(z)] = z \partial_z H(z).
	\end{gather}
	In \cite{FJMM} the degrees with respect to $d_1$ and $d_2$ are called principal degree and
	homogeneous degree, respectively. 
	Hence, the degree of the vertical spectral parameter $z$ counts the $d_2$-grading and
	the generators $E_k, F_k$ and $H_r$ have gradings $(-1,k), (1,k)$ and $(0,r)$, respectively.
	The generators with higher $d_1$-grading are given by multiple commutators of $E_k$ and $F_k$.
	Later we will introduce the horizontal spectral parameter $u$ which counts the $d_1$-grading.
	The $SL(2, \mathbb{Z})$ automorphism of the DIM algebra acts on this bi-grading \cite{M}.
	The grading operators are important in the expression of the universal $ R $-matrix 
	which is the main tool for deriving the generalized Knizhnik-Zamolodchikov equation for MacMahon intertwiner. 
	
	\subsection{Universal $R$ matrix}
	\label{universal}
	
	The quantum toroidal algebra allows a quantum (Drinfeld) double construction \cite{BS}. 
	Consequently it has a quasi-triangular structure, which implies the existence of a universal $R$ matrix. 
	According to \cite{FJMM} (see also \cite{Negut}), the universal $ R $-matrix $ \cals{R}$ of DIM algebra factorizes as follows;
	\begin{equation}
		\cals{R} = \gq^{c^\perp \otimes d^\perp + d^\perp \otimes c^\perp} \cals{R}_{+} \cals{R}_{0} \cals{R}_{-},
	\end{equation} 
	where $\gq^{c^\perp}= K_0^{-}$ and $d^\perp= d_1$ (the grading operator for the principal degree).
	What is most relevant in the present paper is the Cartan factor of $ \cals{R}$\footnote{The definition of $\kappa_n$ in this paper
		is $-\kappa_n$ in  \cite{FJMM}.};
	\begin{align}\label{universalR}
		\gq^{c^\perp \otimes d^\perp + d^\perp \otimes c^\perp} \cdot \cals{R}_0 
		= (K_0^{-} \otimes \gq^ {d_1})(\gq^{d_1} \otimes K_0^{-})\exp\big\{						
		\sum_{n=1}^{\infty}n \kappa_n h_{-n} \otimes h_{n}
		\big\},
	\end{align} 
	where $ h_{\pm n} $ is defined via $\kappa_n h_{\pm n} = \pm H_{\pm n}$.
	Note that the universal $R$ matrix we use in the present paper is $P\cals{R}$ with $P(a \otimes b) = b \otimes a$ in \cite{FJMM}.

	It is known that DIM algebra acts on the equivariant $ K $-theory ;
	$ \displaystyle{\oplus_{n=0}^\infty} K_G(\mathrm{Hilb}^n (\mathbb{C}^2))$ 
	of the Hilbert schemes of $ n $ points on $ \mathbb{C}^2 $.
	Hence, we can also define the $ R $-matrix by using the ideas coming from geometry \cite{MO}. 
	It is interesting to find that the $ R $-matrix $ [\rho^{V_1} \otimes \rho^{V_2}] (\mathcal{R}_0) $ featured in our generalized KZ equation \eqref{DIMKZ}
	coincides with the infinite slope $ R $-matrix $ R_\infty $ which is ubiquitous 
	in the Khoroshkin-Tolstoy factorization of the slope $ s $ $ R $-matrix introduced in \cite{OS}.
	As noticed in \cite{OS} $R_\infty$ corresponds to multiplication by a class of normal
	bundles in $K$-theory and is diagonal in the fixed point basis of the torus action.
	It is an intriguing challenge to work out a possible link of the quantum difference 
	equation \eqref{DIMKZ} to those in \cite{OS}.
	
	
	\section{Representations of DIM algebra}
	\label{section3}
	
	In this section we concern with the representation of DIM algebra. Note that in this paper we use the word representation and module interchangeably. 
	We start this section with the definition of level.
	\begin{dfn}[Level]
		Let $ V $ be a representation of DIM algebra $U_{\gq,\gd}(\widehat{\widehat{\mathfrak{gl}}}_1)$. 
		We say that the representation $ V $ is of level $ (\gamma_1,\gamma_2) \in \mathbb{C}^2 $ if $ C $ 
		and $ K_0^- $ act as constant multiplications by $ \gamma_1 $ and $ \gamma_2 $, respectively. 
	\end{dfn}
	We will call a representation with $ \gamma_1 = 1 $ \textbf{\textit{vertical representation}}.
	Note that the condition $ \gamma_1 = 1 $ is kept intact under taking the tensor product.
	By \eqref{1.16} $H_r$ are mutually commuting for the vertical representation. Hence the vertical representation 
	allows a basis which simultaneously diagonalizes the Cartan modes $H_r$. 
	There are three natural vertical representations of DIM algebra; vector, Fock and MacMahon representations.
	A basis which diagonalizes $H_r$ is labeled by 1d, 2d and 3d Young diagrams.

	\subsection{Vector representations}
	\label{subsecvectorrep}
	We start with the vector representation which is considered as the simplest vertical representation.
	Though DIM algebra is completely symmetric in parameters $(q_1, q_2, q_3)$, the symmetry is broken at the level of
	representation in general. In order to define the vector representation we have to choose one of three parameters as \say{prefered}. 
	Accordingly there are three kinds of vector representations $ \rho^{V^{(k)}}~(k=1,2,3)$ \cite{Z18}.
	For a parameter $ v \in \mathbb{C} $ we consider a vector space $ V (v) $ over $ \mathbb{C} $ with a basis
	$ \{	[v]_i \big| \, i \in \mathbb{Z} \}$.	
	The vector representation $V^{(k)}$ is defined as follows : 
	\begin{align}
		K^+(z) [v]_i &= \tilde{\psi}_k(q_k^i v/z) [v]_i, \label{vecKp} \\
		K^-(z) [v]_i &= \tilde{\psi}_k(q_k^{-i-1} z/v) [v]_i, \label{vecKm} \\
		E(z) [v]_i &= (1-q_k)^{-1} ~\delta(q_k^{i+1} v/z) [v]_{i+1}, \label{vecE} \\
		F(z) [v]_{i+1} &= (1-q_k^{-1})^{-1} ~\delta(q_k^{i+1} v/z) [v]_{i}. \label{vecF}
	\end{align}
	Here 
	\begin{align}
		\tilde{\psi}_1(z) = \frac{(1-q_2^{-1}z)(1-q_3^{-1}z)}{(1-z)(1-q_1 z)},
		\label{2.5}
	\end{align}
	and $\tilde{\psi}_2(z)$ and $\tilde{\psi}_3(z)$ are defined by the cyclic permutation of $(q_1, q_2, q_3)$.
	The original vector representation in \cite{FFJMM1} is $  \rho^{V^{(1)}}$. 
	In the following we will choose the same one and simply denote it by $  \rho^{V}$. 
	If we introduce $ \displaystyle \psi(z) = \gq\frac{1-q^{-1}_{3}z}{1-z} $, then we can express the $ \tilde{\psi}_1(z) $ as\footnote{
		Note that we can exchange $q_2$ and $q_3$.}
	\begin{gather}
		\tilde{\psi}_1(z) = \psi(z)\psi(q^{-1}_{2}z)^{-1}. 
	\end{gather}
	It is $ \psi(z) $ rather than $ \tilde{\psi}_1(z) $ which plays a main role later, for example see
	Eqs.\eqref{2.22} and \eqref{2.23} in section \ref{subsectionFock}. 
	It is straightforward to show that $ V(v) $ with the above action really forms an irreducible representation
	with level $ (1,1) $. 
	\begin{rem}
		We give two remarks here. 
		\begin{enumerate}
			\item In the equations \eqref{vecKp} -- \eqref{vecF}, $ \tilde{\psi} $ and $ \delta $ are formal power series. 
			\item The vector representation is not a highest-weight representation. 
		\end{enumerate}
	\end{rem}

	\subsection{Tensor product representation of the vector representation}
	
	Now perform the tensor product 
	\begin{gather}
		V(v_1) \otimes V(v_2) \otimes \cdots \otimes V(v_n).
	\end{gather}
	So we obtain a representation of $ \underbrace{U_{\gq,\gd}(\widehat{\widehat{\mathfrak{gl}}}_1) \otimes \cdots \otimes U_{\gq,\gd}(\widehat{\widehat{\mathfrak{gl}}}_1)}_{\text{ n times}} $. By using $ \Delta^{n-1} $, 
	\begin{align}
		\Delta^{n-1} (K^{\pm} (z)) &= \overbrace{K^{\pm} (z) \otimes \cdots \otimes K^{\pm}(z)}^{n}, \label{DNK}
		\\
		\Delta^{n-1} (E(z)) &= \sum_{k=1}^n \overbrace{K^{-}(z) \otimes \cdots K^{-}(z)}^{k-1}
		\otimes E(z) \otimes \overbrace{1 \otimes \cdots \otimes 1}^{n-k}, \label{2.8} \\
		\Delta^{n-1} (F(z)) &= \sum_{k=1}^n \overbrace{1 \otimes \cdots \otimes 1}^{k-1}
		\otimes F(z) \otimes \overbrace{K^{+}(z) \otimes \cdots \otimes K^{+}(z)}^{n-k} \label{DNF}.
	\end{align}
	we obtain a representation of $ U_{\gq,\gd}(\widehat{\widehat{\mathfrak{gl}}}_1) $.
	
	However, there are two issues here. First, we see from Eq.\eqref{2.8} that the action of $ E(z) $ results 
	in a product of the formal power series $ \tilde{\psi} $ and $ \delta $ which is not defined in general,
	since it contains an infinite summation of elements in we obtain a representation of $ U_{\gq,\gd}(\widehat{\widehat{\mathfrak{gl}}}_1) $. 
	The solution to this issue is to perform a regularization. More precisely, in the tensor product representation 
	we treat $ \tilde{\psi} $ to be a function instead of formal power series. 
	
	Second, after regularization, $ \tilde{\psi} $ is a function. From Eq.\eqref{2.5}, we see that it contains poles. 
	If we choose $ v_1, v_2, \dots, v_n  $ not carefully, then it might hit the poles generated from $ \tilde{\psi} $. 
	This means that the parameters $ v_1, v_2, \dots, v_n  $ can not be chosen arbitrarily. 
	Then, it is natural to ask what is a condition for $ v_1, v_2, \dots, v_n  $ to assure that the action defined by $ \Delta^{n-1} $ on
	$ V(v_1) \otimes V(v_2) \otimes \cdots \otimes V(v_n) $ does not hit the poles.
	An answer is given in the following lemma which we take from the paper \cite{FFJMM1}. 
	
	\begin{lem}
		If the parameters $ v_1, v_2, \dots, v_n  \in \mathbb{C} $ satisfy the condition that for any
		$ 1 \leq i < j \leq n $, 
		\begin{gather*}
			\frac{v_j}{v_i} \neq q^k_1 \quad \forall k \in \mathbb{Z}, 
		\end{gather*}
		then $ V(v_1) \otimes V(v_2) \otimes \cdots \otimes V(v_n) $ is a well-defined representation of DIM algebra. 
		\label{lemma2.3}
	\end{lem}
	
	As a consequence of Lemma \ref{lemma2.3}, we get that the representation $ V^n(v) $ defined by
	\begin{gather}
		V^n(v) = V(v) \otimes V(q_2v) \otimes \cdots \otimes V(q^{n-1}_2v)
		\label{2.11}
	\end{gather}
	is well-defined.
	For each $ \lambda = (\lambda_1, \dots, \lambda_n) \in \mathbb{Z}^n $, 
	we define the state
	\begin{gather}
		|\lambda\rangle := [v]_{\lambda_1-1} \otimes [q_2v]_{\lambda_2 -1 } \otimes \cdots \otimes
		[q^{n-1}_{2}v]_{\lambda_n - 1}. 
	\end{gather}
	It is easy to see that $ \{	|\lambda\rangle	\big|\,	\lambda = (\lambda_1, \dots, \lambda_n) \in \mathbb{Z}^n		\} $ forms a basis of $ V^n(v) $. 
	
	\medskip 
	
	Note that the representation $ V^n(v) $ is reducible since it contains a nonzero proper DIM-submodule $ W^n(v) $ which is defined by 
	\begin{gather}
		W^n(v) = \operatorname{span}_{\mathbb{C}}\{	|\lambda\rangle 
		\big|\, \lambda \in \mathcal{P}^n
		\}.
	\end{gather}
	Here
	\begin{gather}
		\mathcal{P}^n = \{	\lambda = (\lambda_1,\cdots,\lambda_n) \in \mathbb{Z}^N \big| \lambda_1 \geq \cdots \geq \lambda_n							\}. 
	\end{gather}

	\subsection{Fock representation}
	\label{subsectionFock}
	
	Once we have constructed the $ n $-th tensor product representation $ V^n(v) $ for any
	$ n \in \mathbb{Z}^{\geq 1} $, it is natural to think about the representation of DIM algebra which collects 
	all of the $ V^n(v) \,\, (n \in \mathbb{Z}^{\geq 1}) $ together under an identification
	\begin{align}
		|(\lambda_1, \dots, \lambda_n) \rangle \sim |(\lambda_1, \dots, \lambda_n, 0) \rangle \sim
		|(\lambda_1, \dots, \lambda_n,0,0) \rangle \sim \cdots
		\label{2.14}
	\end{align}
	and so on. This recalls us the notion of inductive limit. Of course, to perform the inductive limit we need to have the inductive system in our hand first.

	Here we construct the inductive system. First we define $ W^{n,+}(v) $ to be the subspace
	\begin{gather}
		\operatorname{span}_{\mathbb{C}}\{	|\lambda\rangle 
		\big|\, \lambda \in \mathcal{P}^{n,+}
		\} \subseteq V^n(v)
		\label{2.16}
	\end{gather}
	where
	$ \mathcal{P}^{n,+} = \{	\lambda = (\lambda_1,\cdots,\lambda_n) \in \mathbb{Z}^N \big| \lambda_1 \geq \cdots \geq \lambda_n \geq 0							\} $. 
	Then we construct the inductive system of vector spaces. 
	\begin{gather}
		W^{1,+}(v) \xrightarrow{\tau_1} W^{2,+}(v) \xrightarrow{\tau_2} W^{3,+}(v) \cdots 
		\label{inductivesystem}
	\end{gather}
	where $ \tau_n : W^{n,+}(v) \rightarrow W^{n+1,+}(v) $ sends
	$ |(\lambda_1,\dots,\lambda_n)\rangle $ to $ |(\lambda_1,\dots,\lambda_n,0)\rangle $. Then, we take the inductive limit of the above inductive system. 
	\begin{gather}
		\mathcal{F}(v) = \inlim{} W^{n,+}(v).
		\label{inductivelimit}
	\end{gather}
	Being an inductive limit, the vector space $ \mathcal{F}(v) $ is spanned by
	$ \{ |\lambda\rangle 	\big|\, \lambda \in \mathcal{P}^+			\} $ where
	\begin{gather}
		\mathcal{P}^+ =
		\{ \lambda = (\lambda_1,\lambda_2,\dots)
		\big|\, \lambda_i \geq \lambda_{i+1}, \lambda_i \in \mathbb{Z},
		\lambda_i = 0 \text{ for sufficiently large } i
		\}. 
	\end{gather}

	The next task is to endow the structure of  $U_{\gq,\gd}(\widehat{\widehat{\mathfrak{gl}}}_1)$-module on $ \mathcal{F}(v) $.
	From Eq.\eqref{inductivelimit}, we know that $ \cals{F}(v) $ is the disjoint union of $ W^{n,+}(v) \,\, (n \in \mathbb{Z}^{\geq 1}) $ modulo the identification $ \sim $ in \eqref{2.14}. Accordingly, it is natural to use the $U_{\gq,\gd}(\widehat{\widehat{\mathfrak{gl}}}_1)$-module structure on each $ W^{n,+} $
	to construct the $U_{\gq,\gd}(\widehat{\widehat{\mathfrak{gl}}}_1)$-module structure on $ \cals{F}(v) $. More precisely, if a partition $ \lambda = (\lambda_1,\lambda_2,\dots,\lambda_l,0,0,\dots) $, then we regard $ |\lambda\rangle $ as an element of $ V^l(v) $ and use the $U_{\gq,\gd}(\widehat{\widehat{\mathfrak{gl}}}_1)$-module structure of $ V^l(v) $. 
	
	Unfortunately, the above action 
	is \textbf{not compatible} with the inductive system \eqref{inductivesystem}. That is, if we denote the representation $ V^l $ by 
	$ \rho^l : U_{\gq,\gd}(\widehat{\widehat{\mathfrak{gl}}}_1) \rightarrow \End V^l $ and 
	if $ \lambda = (\lambda_1,\lambda_2,\dots,\lambda_l,0,0,\dots) $, then we find that 
	\begin{gather}
		\rho^l\big(	 K^{\pm}(z)				\big)|(\lambda_1,\lambda_2,\dots,\lambda_l)\rangle
		\neq \rho^{l+1}\big(	 K^{\pm}(z)				\big)|(\lambda_1,\lambda_2,\dots,\lambda_l,0)\rangle,
	\end{gather}
	and 
	\begin{gather}
		\rho^l\big(	 F(z)				\big)|(\lambda_1,\lambda_2,\dots,\lambda_l)\rangle
		\neq \rho^{l+1}\big(	 F(z)				\big)|(\lambda_1,\lambda_2,\dots,\lambda_l,0)\rangle. 
	\end{gather}

	To make it compatible with the inductive system, we need to define the action
	$ \bar{\rho} : U_{\gq,\gd}(\widehat{\widehat{\mathfrak{gl}}}_1)
	\rightarrow \operatorname{End}\big(	\cals{F}(v)		\big)  $ as follows : for each partition $ \lambda = (\lambda_1,\lambda_2,\dots,\lambda_l,0,0,\dots) $ we determine  
	\begin{gather}
		\bar{\rho}\big(		K^{\pm}(z)		\big)|\lambda\rangle
		= \beta^{\pm}_{l+1}( (v/z)^{\pm}		)\cdot\rho^{l+1}\big(		K^{\pm}(z)		\big)|\lambda\rangle,
		\notag \\
		\bar{\rho}\big(	E(z)			\big)|\lambda\rangle = \rho^{l+1}\big(	E(z)		\big)|\lambda\rangle,
		\notag \\
		\bar{\rho}\big(		F(z)	\big)|\lambda\rangle = \beta^{+}_{l+1}(v/z)\rho^{l+1}\big(
		F(z)
		\big)|\lambda\rangle, 
		\label{2.22}
	\end{gather}
	where
	\begin{gather}
		\beta^+_l(v/z) = \psi(	q^{-1}_{1}q^{l-1}_2v/z		)^{-1}
		\quad \quad
		\beta^-_l(z/v) = \psi(	q^{-l}_{2}z/v			). 
		\label{2.23}
	\end{gather}
	
	This action is well-defined and compatible with the structure of inductive system. This is assured by the following theorem. 
	\begin{thm}
		Let $ \lambda = (\lambda_1,\lambda_2,\dots,\lambda_l,0,0,\dots) $ be a partition. Then, for any $ k \in \mathbb{Z}^{\geq 1} $
		\begin{gather}
			\bar{\rho}\big(		K^{\pm}(z)		\big)|\lambda\rangle
			= \beta^{\pm}_{l+k}( (v/z)^{\pm}		)\cdot\rho^{l+k}\big(		K^{\pm}(z)		\big)|\lambda\rangle,
			\notag \\
			\bar{\rho}\big(	E(z)			\big)|\lambda\rangle = \rho^{l+k}\big(	E(z)		\big)|\lambda\rangle,
			\notag \\
			\bar{\rho}\big(		F(z)	\big)|\lambda\rangle = \beta^{+}_{l+k}(v/z)\rho^{l+k}\big(
			F(z)
			\big)|\lambda\rangle.
		\end{gather}
	\end{thm}
	Thus, we have equipped a $U_{\gq,\gd}(\widehat{\widehat{\mathfrak{gl}}}_1)$-module structure to $ \cals{F}(v) $. We call $ \big(	\bar{\rho} : U_{\gq,\gd}(\widehat{\widehat{\mathfrak{gl}}}_1) \rightarrow \cals{F}(v)	, \cals{F}(v)		\big) $ the Fock representation of DIM algebra. 
	The Fock representation is irreducible representation of level $ (1,\gq) $ \cite{FFJMM1} \cite{MacMahon}.

	\subsection{MacMahon representation}
	Up to now we have constructed the Fock representation from the vector representations. 
	Now we construct the MacMahon representation from the Fock representations by a similar process 
	used in section \ref{subsectionFock}.
	
	Analogous to Eq.\eqref{2.11}, we define 
	\begin{gather}
		\cals{F}^n(v) = \cals{F}(v) \tens \cals{F}(q_3v) \tens \cdots \tens \cals{F}(q_3^{n-1}v). 
	\end{gather}
	It is clear that the subset 
	\begin{gather}
		\{ |\Lambda\rangle := |\Lambda^{(1)}\rangle \tens \cdots \tens |\Lambda^{(n)}\rangle
		\big|\, \Lambda^{(1)},\dots,\Lambda^{(n)} \in \cals{P}^+
		\}
	\end{gather}
	forms a basis of $ \cals{F}^n(v) $. Note that sometimes we write $ |\Lambda^{(1)},\dots,\Lambda^{(n)}\rangle $ for $ |\Lambda^{(1)}\rangle \tens \cdots \tens |\Lambda^{(n)}\rangle  $.

	Next we construct the subspace $ \cals{M}^n(v) $ of $ \cals{F}^n(v) $ spanned by plane partitions, i.e. the
	$|\Lambda^{(1)},\dots,\Lambda^{(n)}\rangle $ which satisfy the condition 
	\begin{align}
		\Lambda^{(k)}_i \geq \Lambda^{(k+1)}_i \quad \quad \forall i, k.
	\end{align}
	This step is analogous to Eq.\eqref{2.16}. Then, we collect these subspaces $ \cals{M}^n(v) $ 
	together by running $ n $ over $ \mathbb{Z}^{\geq 1} $, and then form an inductive system as \eqref{inductivesystem}.

	Now we would like to endow a structure of $U_{\gq,\gd}(\widehat{\widehat{\mathfrak{gl}}}_1)$-module to the vector space $ \ind_n \cals{M}^n(v)  $ 
	which is the inductive limit of the above-mentioned inductive system. 
	As usual we first try to use the action as in the representation $ \cals{F}^n(v) $. 
	Again the problem arises : the action is not compatible with the structure of the inductive system. So we need a modification. 
	
	To make it compatible with the inductive system, we need to define the action $ \bar{\varrho} : U_{\gq,\gd}(\widehat{\widehat{\mathfrak{gl}}}_1)
	\rightarrow \operatorname{End}\big(	\ind_n \cals{M}^n(v) 	\big) $ as follows :
	for each 3d partition $ \Lambda = \big(
	\Lambda^{(1)},\dots,\Lambda^{(l)}, 0, 0, \dots 
	\big)  $
	we determine 
	\begin{gather}
		\bar{\varrho}\big(	K^{\pm}(z)				\big)|\Lambda\rangle = 
		\gamma^{\pm}_{l+1}(	(v/z)^{\pm}		)\cdot \varrho^{l+1}\big(	K^{\pm}(z)				\big)|\Lambda\rangle,
		\notag \\ 
		\bar{\varrho}\big(	E(z)					\big)|\Lambda\rangle = \varrho^{l+1}\big(
		E(z)
		\big)|\Lambda\rangle,
		\notag \\ 
		\bar{\varrho}\big(	F(z)					\big)|\Lambda\rangle = \gamma^+_{l+1}(v/z)\cdot 
		\varrho^{l+1}\big(		F(z)		\big)|\Lambda\rangle,
	\end{gather}
	where
	\begin{gather}
		\gamma^+_{l}(v/z) =
		\frac{	K^{-1/2}(1-Kv/z)			}{\gq^{-l}(1- q_3^lv/z)},
		\quad \quad
		\gamma^-_{l}(z/v) = 
		\frac{	K^{1/2}	\big( 1 - \frac{z}{Kv}				\big)				}{\gq^l \big(	1 - \frac{z}{q^l_3v}			\big)}. 
	\end{gather}
	Here $ K $ is an arbitrary parameter which arises by the prescription of making the inductive system consistent \cite{FFJMM1}.
	The appearance of this continuous parameter $K$ is one of the most intriguing aspects of the MacMahon representation.

	It is straightforward to check that the action $ \bar{\varrho} $ is well-defined and compatible with the structure of the inductive system. 
	Thus we have equipped the $U_{\gq,\gd}(\widehat{\widehat{\mathfrak{gl}}}_1)$-module structure to $ \ind_n \cals{M}^n(v) $. We call
	$ \big(	\bar{\varrho}	: 	U_{\gq,\gd}(\widehat{\widehat{\mathfrak{gl}}}_1) \rightarrow \End(\ind_n \cals{M}^n(v) ), 
	\ind_n \cals{M}^n(v) 			\big) $ the MacMahon module.
	It is easy to see that the MacMahon module has level $ (1,K^{1/2}) $. Because of this, from now on we denote it by $ \cals{M}(K;v) $. 
	As was shown in \cite{FJMM11} \cite{BFM}, when $K=q_1^a q_2^b q_3^c,~a,b,c \in \mathbb{Z}_{\geq 0}$, 
	the MacMahon representation is reducible. We can reduce the representation space to that 
	spanned by plane partitons with a \say{pit} at $(a+1, b+1, c+1)$.  
	In particular for $K=q_3$ the \say{pit} is at $(1,1,2)$ and the plane partition has only the first layer. 
	Thus the representation is reduced to the Fock representation.

	In the cases of the vector and the Fock representations, one of the parameters $(q_1, q_2, q_3)$ of the DIM algebra plays 
	a distinguished role. Consequently there are three kinds of the vector and the Fock representations\footnote{The conventional
		choice, which we follow in this paper, is to choose $q_1$ for the vertical representation and $q_3$ for the Fock representation.
		For the existence of the intertwiner we have to choose the parameters of the vector and the Fock representations differently.}.
	On the other hand in the MacMahon representation three parameters are treated as an equal footing and the triality of 
	DIM algebra is manifest.

	\subsection{Horizontal Fock representation}
	\label{subsechorizontal}
	Up to now we only discussed about vertical representations. In this subsection we construct a horizontal representation of DIM algebra
	with $ C \neq 1$. As in the case of the vector and the vertical Fock representations, there are three kinds of horizontal Fock representations,
	for which $C = q_k^{\frac{1}{2}}$. In the following we fix $k$ and write $\gq=q_k^{\frac{1}{2}}$. The conventional horizontal representation
	corresponds to the choice $k=3$. When $C =  \gq$, from Eq.\eqref{1.16} we obtain the Heisenberg algebra
	\begin{align}\label{Heisenberg}
		[H_r, H_s]
		&= \delta_{r+s,0} \frac{\kappa_r}{r} (\gq^r -  \gq^{-r}) .
	\end{align}
	
	There is a well-known representation of Heisenberg algebra whose representation space $ \cals{F} $ is the Fock space 
	of a (deformed) free boson with a creation operators $ a_{-r} \,\, ( r > 0) $ acting on the vacuum state $ |0\rangle $. 
	The vacuum state is annihilated by the annihilation operator $ a_r \,\, (r>0) $. 
	Namely, we define the horizontal Fock representation by\footnote{Since the only difference of $H_r$ and $a_r$ is the normalization,
		we will use them interchangeably in this paper.}
	\begin{align}
		\rho_H^{(\gq,1)} (H_r) := 
		\frac{\kappa_r}{r} a_r , \qquad
		[ a_r, a_s ] = \delta_{r+s,0}\frac{r}{\kappa_r} (\gq^{r} - \gq^{-r}).
	\end{align}

	Now we try to endow $ \cals{F} $
	with the $U_{\gq,\gd}(\widehat{\widehat{\mathfrak{gl}}}_1)$-module structure. 
	We construct the action $ \rho_H^{(\gq,1)} : U_{\gq,\gd}(\widehat{\widehat{\mathfrak{gl}}}_1) 
	\rightarrow \End\cals{F}
	$ by 
	\begin{gather}\label{Vertex}
		\rho_H^{(\gq,1)}\big(		E(z)		\big) = V^{-}(\gq^{-1/2}z)V^+(\gq^{1/2}z),
		\notag \\ 
		\rho_H^{(\gq,1)}\big(	F(z)			\big) = V^{-}(\gq^{1/2}z)^{-1}V^+(	\gq^{-1/2}z		)^{-1},
		\notag \\ 
		\rho_H^{(\gq,1)}\big(	K^{\pm}(\gq^{1/2}z)				\big) = V^{\pm}(\gq^{\pm 1}z)V^{\pm}(\gq^{\mp 1}z)^{-1},
	\end{gather}
	where
	\begin{gather}
		V^{\pm}(z) = \exp\bigg(
		\mp\sum_{r=1}^{\infty}\frac{1}{r} \frac{\kappa_r}{\gq^{r} - \gq^{-r}}  a_{\pm r} z^{\mp r}
		\bigg).
	\end{gather}
	It is straightforward to show that
	$ \big(	\rho_H^{(\gq,1)}, \cals{F}	\big) $ is a representation of DIM algebra of level $ (\gq,1) $.

	Now for a given $ \gamma_2 \in \mathbb{C}\backslash\{0\} $ we try to construct 
	a horizontal representation of DIM algebra of level $ (\gq,\gamma_2) $ by generalizing the above formula.
	While the representation space is the same as above, say $ \cals{F} $, the action is modified to be 
	$ \rho_H^{(\gq,\gamma_2)} : U_{\gq,\gd}(\widehat{\widehat{\mathfrak{gl}}}_1) 
	\rightarrow \End\cals{F}
	$ by multiplying the zero mode factors $ \textbf{e}(z), \textbf{f}(z), \textbf{k}^{\pm}(z) $;
	\begin{gather}
		\rho_H^{(\gq,\gamma_2)}\big(		E(z)		\big) = \rho_H^{(\gq,1)}\big(		E(z)		\big) \textbf{e}(z),
		\notag \\ 
		\rho_H^{(\gq,\gamma_2)}\big(	F(z)			\big) = \rho_H^{(\gq,1)}\big(	F(z)			\big)\textbf{f}(z),
		\notag \\ 
		\rho_H^{(\gq,\gamma_2)}\big(	K^{\pm}(\gq^{1/2}z)				\big) 
		= \rho_H^{(\gq, 1)}\big(	K^{\pm}(\gq^{1/2}z)				\big) \textbf{k}^{\pm}(z)
	\end{gather}
	where $ \textbf{e}(z), \textbf{f}(z), \textbf{k}^{\pm}(z) $ should satisfy the condition 
	\begin{align}
		\textbf{e}(z)\textbf{f}(\gq^{\mp 1}z) = \textbf{k}^{\pm}(\gq^{\mp 1/2}z),
		\label{3.34}
	\end{align}
	and 
	\begin{align}
		\textbf{k}^{\pm}(z) = \textbf{k}^{\pm}(0) = \gamma_2^{\mp 1}. 
		\label{3.35}
	\end{align}
	It is straightforward to check that
	$ \big( \rho_H^{(\gq,\gamma_2)}, \cals{F} \big) $ is a representation of DIM algebra of level $ (\gq,\gamma_2) $.

	When $ \gamma_2 = \gq^N $, by introducing the spectral parameter $ u \in \mathbb{C} $ of 
	the horizontal representation, we can solve the conditions \eqref{3.34} and \eqref{3.35} as follows;
	\begin{gather}\label{levelN}
		\textbf{e}(z) = \bigg(	\frac{\gq}{z}		\bigg)^Nu, \quad \quad
		\textbf{f}(z) = \bigg(		\frac{\gq}{z}			\bigg)^{-N}u^{-1}, \quad \quad 
		\textbf{k}^{\pm}(z) = \gq^{\mp N},
	\end{gather}
	which was originally employed in \cite{AFS}.
	We denote the representation with this choice of $ \textbf{e}(z), \textbf{f}(z)$ and $\textbf{k}^{\pm}(z) $ 
	by $ \cals{F}^{(\gq,\gq^N)}_{u} $, which is a representation of DIM algebra of level $ (\gq,\gq^N) $. 
	Recall that the MacMahon representation has a continuous parameter $K$ for the second level $\gamma_2$.
	Hence we can no longer use \eqref{levelN} when the MacMahon representation is involved. This is the reason
	why we need mode general zero mode algebra of $ \textbf{e}(z), \textbf{f}(z), \textbf{k}^{\pm}(z) $.

	
	\section{Trivalent intertwiners}
	\label{intertwiners}
	
	In this section, we give a quick review on the trivalent intertwiners and the dual intertwiners. A trivalent intertwiner 
	$ \Psi : \cals{V} \tens \cals{H} \rightarrow \cals{H}^{\prime} $, where $ \cals{V} $ is a vertical representation 
	and $ \cals{H},\cals{H}^{\prime} $ are horizontal representations, is determined by the intertwining relation 
	\begin{gather}
		a\Psi = \Psi\Delta(a) \quad \quad \forall a \in U_{\gq,\gd}(\widehat{\widehat{\mathfrak{gl}}}_1).
		\label{intertwiningrelation}
	\end{gather}

	Taking a basis $ \{	\alpha		\} $ of $ \cals{V} $ which diagonalizes $H_r$,
	we define the $ \alpha $-component of the intertwiner $ \Psi_{\alpha}(\bullet) $ by 
	\begin{gather}
		\Psi_{\alpha}(\bullet) = \Psi(		\alpha \tens \bullet						) : \cals{H} \rightarrow \cals{H}^{\prime}.
	\end{gather}
	By using $ \Psi_{\alpha}(\bullet) $ we can express the intertwining relation \eqref{intertwiningrelation} as 
	\begin{align}
		K^{+}(z) \Psi_\alpha &= (\alpha \vert  K^{+}(z) \vert \alpha) ~\Psi_\alpha K^{+}(z), 
		\label{intK+}
		\\
		K^{-}(\gq z) \Psi_\alpha &= (\alpha \vert  K^{-}(z) \vert \alpha) ~\Psi_\alpha K^{-}(\gq z),
		\label{intK-}
		\\
		E(z) \Psi_\alpha &= \sum_{\beta} (\beta \vert E(z) \vert \alpha) ~\Psi_\beta + (\alpha \vert K^{-}(z) \vert \alpha) ~\Psi_\alpha E(z), 
		\label{4.5}
		\\
		F(z) \Psi_\alpha &= \sum_{\beta} (\beta \vert F(\gq z) \vert \alpha) ~\Psi_\beta K^{+} (\gq z) + \Psi_\alpha F(z).
		\label{4.6}
	\end{align}

	When $ \cals{V} $ is the vertical vector/Fock/MacMahon modules, we call the intertwiner 
	$ \Psi : \cals{V} \tens \cals{H} \rightarrow \cals{H}^{\prime} $ vector/Fock/MacMahon intertwiner, respectively. 
	In the following we will summarize explicit expressions of these intertwiners which were derived in \cite{MacMahon}.
	We also determine the dual intertwiners (see \eqref{dintK+} -- \eqref{dintF} for the definition) which were not given in \cite{MacMahon}.

	\subsection{Vector Intertwiner}
	\label{vecint}
	In terms of the basis  $ \{	[v]_n		\big|	n \in \mathbb{Z}			\} $ of $ V(v) $ (see section \ref{subsecvectorrep}), 
	we define the components of the intertwiner by 
	\begin{align}\label{vcomponent}
		\mathbb{I}_n(v)(\bullet) = \mathbb{I}([v]_{n-1} \otimes \bullet) \colon \mathcal{H} \to \mathcal{H}'.
	\end{align}
	A solution to the intertwining relations is 
	\begin{align}\label{vintertwin}
		\mathbb{I}_n(v)
		&= z_n \tilde{\mathbb{I}}_n(v), \quad
		\tilde{\mathbb{I}}_n(v) = \tilde{\mathbb{I}}_0(q_1^n v), \quad
		n \in \mathbb{Z}, \\
		\tilde{\mathbb{I}}_0(v)
		&= \exp\left( -\sum_{r=1}^{\infty}\frac{H_{-r}}{\gq^r - \gq^{-r}} \frac{\gq^{-r/2}}{1-q_1^r} v^r \right)
		\exp\left( \sum_{r=1}^{\infty}\frac{H_{r}}{\gq^r - \gq^{-r}} \frac{\gq^{-r/2}}{1-q_1^{-r}} v^{-r} \right).
	\end{align}
	Here 
	\begin{align}
		z_0(v) = 1, \quad
		z_n(v) = q_2^{-n} \prod_{j=1}^{n} \mathbf{e}(q_1^{j-1}v) \quad (n>0), \quad
		z_n(v) = q_2^{-n} \prod_{j=n}^{-1} \mathbf{e}(q_1^{j}v)^{-1} \quad (n<0).
	\end{align}
	For the existence of the intertwiner 
	the zero modes $ \textbf{e}(z), \textbf{f}(z), \textbf{k}^{\pm}(z) $ of $ \cals{H} $ and 
	$ \textbf{e}^{\prime}(z), \textbf{f}^{\prime}(z), \textbf{k}^{\pm\prime}(z) $ of $ \cals{H}^{\prime} $ have to be
	related by\footnote{Recall that
		$ \textbf{k}^{\pm\prime}(z) = \gamma^{\mp 1}_2 $. 
	} \cite{FHHSY};
	\begin{gather}
		\gamma^{\prime}_2 = \gamma_2, \quad \quad
		\textbf{e}^{\prime}(z) = q^{-1}_{2}\textbf{e}(z), \quad \quad 
		\textbf{f}^{\prime}(z) = q_2\textbf{f}(z).
		\label{vectorshift}
	\end{gather}
	From now on we are going to write $ \gamma $ for $ \gamma_2 $ for the sake of convenience.

	\subsection{Fock intertwiner}
	Similarly for the set of partitions $ \{	|\lambda\rangle	 				\} $ which forms a basis of $ \cals{F}(v) $, the $ \lambda $-component of the Fock intertwiner is defined by 
	\begin{gather}
		\Phi_{\lambda}(v) = \Phi\big(	|\lambda\rangle \tens \bullet				\big) : \cals{H} \rightarrow \cals{H}^{\prime}. 
	\end{gather}
	A solution to the intertwining relations is
	\begin{gather}
		\Phi_{\lambda}(v) = z_{\lambda}\cals{G}_{\lambda}^{-1}
		\exp\left( \sum_{r=1}^{\infty} \frac{H_{-r}}{\gq^r - \gq^{-r}}\gq^{-r/2}v^r\left( \sum_{(i,j)\in\lambda}x_{i,j}^r - \frac{1}{(1-q_1^r)(1-q_2^r)} \right) \right) 
		\notag
		\\
		\times \exp\left( -\sum_{r=1}^{\infty} \frac{H_{r}}{\gq^r - \gq^{-r}}\gq^{-r/2}v^{-r}
		\left( \sum_{(i,j)\in\lambda}x_{i,j}^{-r} - \frac{1}{(1-q_1^{-r})(1-q_2^{-r})} \right) \right),
	\end{gather}
	where $x_{i,j} = q_1^{j-1} q_2^{i-1}$.
	The zero mode factor is
	\begin{align}
		z_\lambda(v)
		= \prod_{i=1}^{\ell(\lambda)} \prod_{j=1}^{\lambda_i} \left( -\gq q_2^{i-1} x_{i,j}^{-1} \right) \mathbf{e}(x_{i,j}v)
		= q_2^{n(\lambda)} (-\gq)^{|\lambda|} \prod_{(i,j)\in\lambda} x_{i,j}^{-1} \mathbf{e}(x_{i,j}v), 
	\end{align}
	and the normalization factor is
	\begin{align}
		\mathcal{G}_\lambda
		&= \prod_{\square \in \lambda} \left( 1-q_1^{-a_\lambda(\square)}q_2^{l_\lambda(\square)+1} \right).
	\end{align}

	The Fock intertwiner exists, if and only if the following relations between the zero modes of $ \cals{H} $ and $ \cals{H}^{\prime} $
	are satisfied \cite{AFS};
	\begin{align}\label{Fockzeromodes}
		\gamma' = \gq \gamma, \quad \quad
		\mathbf{e}'(z) = (-\gq v/z) \mathbf{e}(z), \quad \quad 
		\mathbf{f}'(z) = (-\gq v/z)^{-1} \mathbf{f}(z).
	\end{align}
	Compared with \eqref{levelN}, this means that the level and the spectral parameter are shifted 
	by $\gq^N \to \gq^{N+1}$ and $u \to -uv$.

	\subsection{MacMahon intertwiner}
	Now let's consider the MacMahon intertwiner whose vertical representation is $ \cals{M}(K;v) $. 
	We know that the set of plane partitions $ \{	|\Lambda\rangle			\} $ forms a basis of
	$ \cals{M}(K;v) $. We define the $ \Lambda $-component of the MacMahon intertwiner by 
	\begin{gather}
		\Xi_{\Lambda}(K;v) = \Xi\big(
		|\Lambda\rangle \tens \bullet
		\big) : \cals{H} \rightarrow \cals{H}^{\prime}. 
	\end{gather}
	A solution to the intertwining relations is 
	\begin{align}
		\Xi_{\Lambda}(K;v) = z_{\Lambda}(K;v)\mathcal{M}^{[n]}(K)\tilde{\Phi}^{[n]}_{\Lambda}(v)\Gamma_n(K;v), \qquad n > h(\Lambda),
		\label{macmahonansatttz}
	\end{align}
	where
	\begin{gather}
		\Gamma_n(K;v)
		= \exp\left( \sum_{r=1}^{\infty}\frac{H_{-r}}{\gq^r - \gq^{-r}} \frac{q_3^{nr}-K^r}{\kappa_r} \gq^{-r/2} v^r \right)
		\exp\left( \sum_{r=1}^{\infty}\frac{H_{r}}{\gq^r - \gq^{-r}} \frac{q_3^{-nr}-K^{-r}}{\kappa_r} \gq^{-r/2} v^{-r} \right)
		\label{4.3}
	\end{gather}
	is the vacuum contribution which survives for $\Lambda = \varnothing$.
	The zero mode factor is 
	\begin{gather}
		\label{Maczeromode}
		z_\Lambda(K;v)
		= \prod_{k=1}^{h(\Lambda)} \prod_{(i,j)\in\Lambda^{(k)}}
		\frac{K^{1/2}}{\gq^{k-1}} \frac{\theta_{q_3}(q_3^{k-1}/x_{ijk})}{\theta_{q_3}(K/x_{ijk})} \mathbf{e}(x_{ijk}v),
	\end{gather}
	where $x_{ijk} = q_1^i q_2^j q_3^k$.
	Note that for later convenience we have exchanged $q_1$ and $q_2$ in the original formula in \cite{MacMahon}.
	This is equivalent to taking the transpose of the Young diagram $\lambda$ which labels the basis of the Fock representation.
	Or we may exchange the role of $q_1$ and $q_2$ in the construction of the Fock representation in section \ref{subsectionFock}.
	Namely we begin with the vector representation $\rho^{V^{(2)}}$ and choose $q_1$ as the shift parameter in the tensor product. 
	Furthermore, 
	\begin{align}
		\label{Fockcomposition}
		\tilde{\Phi}^{[n]}_\Lambda(v)
		&= \tilde{\Phi}_{\Lambda^{(1)}}(v) \circ \cdots \circ \tilde{\Phi}_{\Lambda^{(n)}}(q_3^{n-1}v)
	\end{align}
	is a composition of the Fock intertwiners $\tilde{\Phi}$ without the zero mode factor.
	Finally the factor $ \mathcal{M}^{[n]}(K) $ is defined via 
	\begin{align}
		:\tilde{\Phi}^{[n]}_{\varnothing}(v) \Gamma_n(K;v): = \mathcal{M}^{[n]}(K)\tilde{\Phi}^{[n]}_{\varnothing}(v) \Gamma_n(K;v). 
	\end{align}
	The notation $ n > h(\Lambda) $ in Eq.\eqref{macmahonansatttz} means 
	that we can choose any integer $ n $ which is greater than $ h(\Lambda) $. All of them give the same $ \Xi_{\Lambda}(K;v) $. 
	It can be shown that 
	\begin{align}
		\Xi_{\Lambda}(K;v) = z_{\Lambda}\cals{G}_{\Lambda}^{-1}
		&\exp\left( \sum_{r=1}^{\infty} \frac{H_{-r}}{\gq^r - \gq^{-r}}\gq^{-r/2}v^r\left( \sum_{(i,j,k)\in\Lambda}x_{ijk}^r + \frac{1-K^r}{\kappa_r} \right) \right) \notag \\ 
		& \quad \times \exp\left( -\sum_{r=1}^{\infty} \frac{H_{r}}{\gq^r - \gq^{-r}}\gq^{-r/2}v^{-r}\left( \sum_{(i,j,k)\in\Lambda}x_{ijk}^{-r} - \frac{1-K^{-r}}{\kappa_r} \right) \right),
	\end{align}
	where the factor $ \cals{G}_{\Lambda} $ is defined by 
	\begin{align}
		:\tilde{\Phi}^{[n]}_{\Lambda}(v) \Gamma_n(K;v): = \cals{G}_{\Lambda}\mathcal{M}^{[n]}(K)\tilde{\Phi}^{[n]}_{\Lambda}(v) \Gamma_n(K;v). 
	\end{align}
	
	As motivated from the Fock intertwiners, we can construct the MacMahon intertwiner from the Fock intertwiners by composing them together with the correction intertwiner \cite{MacMahon}. From this method of construction, we obtain the explicit form of $ \cals{G}_{\Lambda} $. Nevertheless, in this paper, 
	the explicit form is not relevant and we do not state it here. 
	We obtain the following relations between the zero modes of $ \cals{H} $ and $ \cals{H}^{\prime} $ 
	\begin{align}\label{Maczeromodes}
		\gamma' = K^{1/2} \gamma, \quad
		\mathbf{e}'(z) = K^{1/2}\frac{\theta_{q_3}(v/z)}{\theta_{q_3}(Kv/z)} \mathbf{e}(z), \quad
		\mathbf{f}'(z) = \frac{\theta_{q_3}(\gq Kv/z)}{\theta_{q_3}(\gq v/z)} \mathbf{f}(z),
	\end{align}
	Note that when $K=q_3$, \eqref{Maczeromodes} reduces to \eqref{Fockzeromodes}.

	\subsection{Dual intertwiners}
	\label{dualintertwiner}

	Now we turn to the construction of the dual intertwiner $ \Psi^* : \cals{H}^{\prime} \rightarrow \cals{H} \tens \cals{V} $ ,
	which is determined from the intertwining relation : 
	\begin{gather}
		\Delta(a)\Psi^* = \Psi^*a, \quad \quad \forall a \in U_{\gq,\gd}(\widehat{\widehat{\mathfrak{gl}}}_1).
		\label{dualinter}
	\end{gather}
	Taking the same basis $ \{	\alpha		\} $ of $ \cals{V} $
	used to define the components of $ \Psi$, we can define the $ \alpha$-component 
	of the dual intertwiner $ \Psi^*_{\alpha} : \cals{H}^{\prime} \rightarrow \cals{H} $ by 
	\begin{gather}
		\Psi^*(\bullet) = \sum_{\alpha} \Psi^*_{\alpha}(\bullet) \tens |\alpha\rangle,
		\quad \quad \bullet \in \cals{H}^{\prime}. 
		\label{dualcomponent}
	\end{gather}
	We can express the intertwining relation \eqref{dualinter} in terms of the $ \alpha$-component as 
	\begin{align}
		\label{dintK+}
		\Psi^*_{\alpha}K^+(\gq z) &= ( \alpha | K^+(z) |\alpha) K^+(\gq z)\Psi^*_{\alpha}, 
		\\
		\label{dintK-}
		\Psi^*_{\alpha}K^-(z) &= ( \alpha | K^-(z) |\alpha) K^-(z)\Psi^*_{\alpha},
		\\
		\label{dintE}
		\Psi^*_{\alpha}E(z) &= K^-(\gq z)\sum_\beta(\alpha | E(\gq z)| \beta )\Psi^*_{\beta}
		+ E(z)\Psi^*_{\alpha},
		\\
		\label{dintF}
		\Psi^*_{\alpha}F(z) &= \sum_{\beta} (\alpha|F(z)|\beta)\Psi^*_{\beta} + 
		(\alpha|K^+(z)|\alpha)F(z)\Psi^*_{\alpha}. 
	\end{align}
	When $ \cals{V} $ is the vertical vector/Fock/MacMahon modules, we call the dual intertwiner 
	$ \Psi^* : \cals{H}^{\prime} \rightarrow \cals{H} \tens \cals{V} $ the vector/Fock/MacMahon dual intertwiner, respectively. 
	In the following we give their explicit expressions in each case.

	We know that $ \big\{	[v]_n		\big|~	 n \in \mathbb{Z}	\big\} $ forms a basis of $ V(v) $. Then, according to \eqref{dualcomponent}, 
	the $ n$-component $ \bb{I}^*_n(v) : \cals{H}^{\prime} \rightarrow \cals{H}  $ of the dual intertwiner satisfies 
	\begin{gather}
		\bb{I}^*(v) = \sum_{n \in \bb{Z}}\bb{I}^*_n(v) \tens [v]_{n-1}. 
	\end{gather}
	Thus, the intertwining relations become 
	\begin{align}
		\bb{I}^*_n(v)K^+(\gq z) &= \tilde{\psi}(q^{n-1}_{1}v/z)K^+(\gq z)\bb{I}^*_n(v), 
		\\
		\bb{I}^*_n(v)K^-(z) &= \tilde{\psi}(q^{-n}_{1}z/v)K^-(z)\bb{I}^*_n(v), 
		\\
		\bb{I}^*_n(v)E(z) &= (1-q_1)^{-1}\delta(\gq^{-1}q^{n-1}_{1}v/z)K^-(\gq z)\bb{I}^*_{n-1}(v)
		+ E(z)\bb{I}^*_n(v),
		\\
		\bb{I}^*_n(v)F(z) &= (1-q^{-1}_{1})^{-1}\delta(	q^n_1v/z		)\bb{I}^*_{n+1}(v)
		+ \tilde{\psi}(q^{n-1}_{1}v/z)F(z)\bb{I}^*_n(v). 
	\end{align}
	A solution to the intertwining relations is 
	\begin{gather}
		\bb{I}^*_n(v) = z^*_n\exp\bigg(
		\sum_{r=1}^{\infty}\frac{H_{-r}}{\gq^r - \gq^{-r}}\frac{\gq^{r/2}(q^n_1v	)^r}{1-q^r_1}
		\bigg)\exp\bigg(
		- \sum_{r=1}^{\infty}\frac{H_r}{\gq^r - \gq^{-r}}\frac{\gq^{r/2}(q^{n}_{1}v	)^{-r}}{1-q^{-r}_1}
		\bigg),
	\end{gather}
	where 
	\begin{gather}
		z^*_n =
		\begin{cases}
			q^n_2\prod_{j=1}^{n}\textbf{f}(q^{j-1}_{1}v) \quad ; ~  n > 0,
			\\
			q^n_2\prod_{j = n}^{-1}\textbf{f}(q^j_1v)^{-1} \quad ; ~ n < 0.
		\end{cases}
	\end{gather}
	As the condition for the existence of the dual intertwiner, we get the following relations \cite{FHHSY};
	\begin{gather}
		\gamma^{\prime} = \gamma, \quad \quad
		\textbf{e}^{\prime}(z) = q^{-1}_{2}\textbf{e}(z), \quad \quad
		\textbf{f}^{\prime}(z) = q_2\textbf{f}(z). 
	\end{gather}

	We define the $ \lambda$-component of the dual Fock intertwiner $ \Phi^*_{\lambda}(v) $ by 
	\begin{gather}
		\Phi^*(v) = \sum_{\lambda}\Phi^*_{\lambda}(v) \tens |\lambda\rangle.
	\end{gather}
	A solution to the dual intertwining relations is
	\begin{align}
		\Phi^{*}_{\lambda}(v) = z^*_{\lambda}(v)\cals{G}^{[n]*}B^*_n(v)\tilde{\bb{I}}^{[n]*}_{\lambda}(v), ~~ n > \ell(\lambda),
	\end{align}
	where
	\begin{align}
		B^*_n(v) = \exp\bigg(
		\sum_{r=1}^{\infty}\frac{1}{k_{r}}(	\gq^{3/2}q^n_2v	)^rH_{-r}
		\bigg)\exp\bigg(
		-\sum_{r=1}^{\infty}\frac{1}{\kappa_r}(	\gq^{1/2}q^n_2v	)^{-r}H_r
		\bigg),
	\end{align}
	$ \tilde{\bb{I}}^{[n]*}_{\lambda} := \tilde{\bb{I}}^{*}_{\lambda_n}(q^{n-1}_{2}v) \comp \cdots \comp \tilde{\bb{I}}^{*}_{\lambda_1}(v) $, 
	and the normalization factor $ \cals{G}^{[n]*} $ is defined by 
	$ :	B^*_n(v)\tilde{\bb{I}}^{[n]*}_{\emptyset}(v): = 	\cals{G}^{[n]*}B^*_n(v)\tilde{\bb{I}}^{[n]*}_{\emptyset}(v) 				$. 
	The factor $ z^*_{\lambda}(v) $ is 
	\begin{gather}
		z^*_{\lambda}(v) = \gq^{|\lambda|}\prod_{i=1}^{l(\lambda)}\prod_{j=1}^{\lambda_i}
		\big(	-\gq q^{i-1}_{2}x^{-1}_{i,j}				\big)^{-1}\textbf{f}(	x_{ij}v		).
	\end{gather}
	It can also be shown that 
	\begin{gather}
		\Phi^*_{\lambda}(v) = z^*_{\lambda}\cals{G}^{*\,-1}_{\lambda}
		\exp\bigg[
		\sum_{r=1}^{\infty}\frac{H_{-r}}{\gq^r - \gq^{-r}}\gq^{r/2} v ^r
		\Big(		\frac{q^{lr}_{2}}{(1- q^r_1 )(1- q^r_2 )} + \sum_{k=1}^{l} \frac{(q^{\lambda_k}_{1}q^{k-1}_{2})^r}{1-q^r_1}								\Big)
		\bigg]
		\notag \\
		\times 
		\exp\bigg[
		\sum_{r=1}^{\infty}\frac{H_r}{\gq^r - \gq^{-r}}\gq^{r/2}v^{-r}
		\Big(
		\sum_{k=1}^{l}\frac{	(	q^{\lambda_k - 1}_{1}q^{k-1}_{2}		)^{-r}			}{1-q^r_1}
		- \frac{q^{-rl}_{2}}{(1- q^{-r}_1 )(1 - q^{-r}_2 )}
		\Big)
		\bigg].
		\label{5.23}
	\end{gather} 
	Here $ \cals{G}^{*}_{\lambda} $ is defined via the equation
	\begin{gather}
		:	B^*_n(v)\tilde{\bb{I}}^{[n]*}_{\lambda}(v): = 	\cals{G}^{*}_{\lambda} \cals{G}^{[n]*}B^*_n(v)\tilde{\bb{I}}^{[n]*}_{\lambda}(v). 
	\end{gather}

	Similarly we can define the $ \Lambda$-component of the MacMahon intertwiner $ \Xi^*_{\Lambda}(K, v) $ by 
	\begin{gather}
		\Xi^*(K;v) = \sum_{\Lambda} \Xi^*_{\Lambda}(K;v) \tens |\Lambda\rangle. 
	\end{gather}
	A solution to the intertwining relations is 
	\begin{align}
		\label{dualMacint}
		\Xi_{\Lambda}^{*}(K;v) = z_{\Lambda}^*(K;v)\cals{M}^{[n]*}(K)\tilde{\Phi}^{[n]*}_{\Lambda}(v)\Gamma^{*}_{n}(K;v),
	\end{align}
	where 
	\begin{gather}
		\Gamma^{*}_{n}(K;v) = 
		\exp\bigg(
		-\sum_{r=1}^{\infty}\frac{H_{-r}}{\gq^r - \gq^{-r}}\frac{q^{nr}_3 - K^r}{\kappa_r}\gq^{r/2}v^r
		\bigg)
		\exp\bigg(
		-\sum_{r=1}^{\infty}\frac{H_r}{\gq^r - \gq^{-r}}\frac{q^{-nr}_3 - K^{-r}}{\kappa_r}\gq^{r/2}v^{-r}
		\bigg),
		\label{4.36}
		\\
		\label{dualMaczeromode}
		z^*_{\Lambda}(K;v) = \prod_{k=1}^{h(\Lambda)}\prod_{(i,j) \in \Lambda^{(k)}}
		\Bigg(
		\frac{K^{1/2}}{\gq^k}\frac{	 \theta_{q_3}(q^k_3x^{-1}_{ijk})			}{\theta_{q_3}(Kx^{-1}_{ijk})}
		\Bigg)^{-1}\textbf{f}(x_{ijk}v),
	\end{gather}
	and
	\begin{gather}
		\label{dualFockcomposition}
		\tilde{\Phi}^{[n]*}_{\Lambda}(v) =
		\tilde{\Phi}^{*}_{\Lambda^{(1)}}(v) \comp \cdots \comp \tilde{\Phi}^{*}_{\Lambda^{(n)}}(q^{n-1}_{3}v). 
	\end{gather}
	The factor $ \cals{M}^{[n]*}(K) $ is defined by the equation
	\begin{gather}
		: \tilde{\Phi}^{[n]*}_{\emptyset}(v)\Gamma^{*}_{n}(K;v)			: = \cals{M}^{[n]*}(K)\tilde{\Phi}^{[n]*}_{\emptyset}(v)\Gamma^{*}_{n}(K;v).  
	\end{gather}
	As in the case of the dual Fock intertwiner, we can determine 
	an explicit form of $ \Xi^*_{\Lambda}(K;v) $ like Eq.\eqref{5.23}, which we will compute in section 6.
	
	
	\section{Generalized KZ equation for the vector intertwiners}
	\label{section6}
	
	In this section we derive the generalized KZ equation for
	the vector intertwiners as a warm-up exercise before attacking the more complicated case of the MacMahon intertwiners. 
	The first step is to construct the shift operator $ \cals{T} $.

	\subsection{Construction of the shift operator}
	
	In \cite{qtKZdim} the shift operator is constructed as the composition of the intertwiner and the dual intertwiner with appropriate
	spectral parameter. We first recall the relation \eqref{vectorshift} on the level of the vector intertwiner discussed in the previous section.
	Comparing the choice of the zero modes \eqref{levelN} with \eqref{vectorshift} saying 
	there is no change of the level of $ \mathcal{H} $ and $ \mathcal{H}^{\prime} $,
	we see $ u^{\prime} = q^{-1}_{2}u $.
	We define the operator $ \mathcal{T}^{m}_{n}$ by the composition of the intertwiner and the dual intertwiner;
	\begin{gather}
		\mathcal{T}^{m}_{n}(N,u|v,w) := \mathbb{I}^{*}_n(w)\mathbb{I}_m(v) \quad 
		:  \quad \mathcal{F}_u^{(\gq, \gq^N)} \longrightarrow  \mathcal{F}_u^{(\gq, \gq^N)}.
	\end{gather}
	Pictorially, it is represented as Fig.\ref{dsoivoidfsuosiufdsoifusio}. 
	
	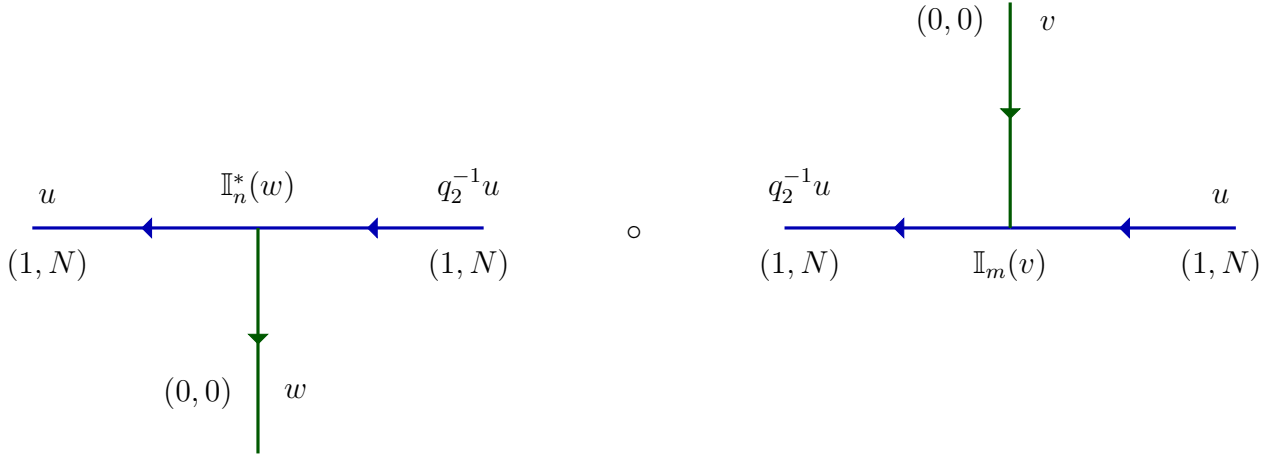
\begin{figure}[h]
		\unitlength 2mm
		\begin{center}
		\begin{tikzpicture}
			\begin{scope}[very thick, every node/.style={sloped,allow upside down}]
				\draw[color=blue!70!black] (0,0)-- node {\midarrow} (-3,0);
				\draw [color=blue!70!black](3,0)-- node {\midarrow} (0,0);
				\draw[color=green!35!black] (0,0) -- node {\midarrow} (0,-3);
			\end{scope}
			every node/.style=draw,
			every label/.style=draw
			]
			\node [label={[shift={(0.5,-2.55)}]$ w $}] {};
			\node [label={[shift={(-0.8,-2.7)}]$ (0,0) $}] {};
			
			\node [label={[shift={(2.8,-1.0)}]$ (1,N) $}] {};
			\node [label={[shift={(-2.8,-1.0)}]$ (1,N) $}] {};
			
			\node [label={[shift={(0,0.04)}]$ \bb{I}^*_n(w) $}] {};
			\node [label={[shift={(2.8,0.04)}]$ q_2^{-1}u $}] {};
			\node [label={[shift={(-2.8,0.04)}]$ u $}] {};
			
			\node [label={[shift={(5,-0.45)}]$ \comp $}] {};
			
			\begin{scope}[very thick, every node/.style={sloped,allow upside down}]
				\draw (10,0)[color=blue!70!black]-- node {\midarrow} (7,0);
				\draw (13,0)[color=blue!70!black]-- node {\midarrow} (10,0);
				\draw[color=green!35!black] (10,3) -- node {\midarrow} (10,0);
			\end{scope}
			\node [label={[shift={(10.5,2.35)}]$ v $}] {};
			\node [label={[shift={(9.2,2.25)}]$ (0,0) $}] {};
			
			\node [label={[shift={(12.8,-1.0)}]$ (1,N) $}] {};
			\node [label={[shift={(7.2,-1.0)}]$ (1,N) $}] {};
			
			\node [label={[shift={(10,-1.0)}]$ \bb{I}_m(v) $}] {};
			\node [label={[shift={(12.8,0.04)}]$ u $}] {};
			\node [label={[shift={(7.2,0.04)}]$ q_2^{-1}u $}] {};
		\end{tikzpicture}
		\end{center}
		\caption{Composition $\circ$ of the vector intertwiner and the dual vector intertwiner. 
			Contrary to the Fock intertwiner, the level $(1,N)$ of the horizontal Fock space is kept intact
			and the shift of the spectral parameter $u$ is independent of the spectral parameters
			$v,w$ of the vertical representation. }
		\label{dsoivoidfsuosiufdsoifusio}
	\end{figure}
	Using the commutation relation \eqref{Heisenberg}, we see that 
	\begin{align}
		\mathcal{T}^{m}_{n}(N,u|v,w) =
		z^{*}_n & \big(	u, N \big| w		\big)z_m\big(	 u, N \big| v		\big) \nn \\
		& \cdot \exp\Big(
		-\sum_{r=1}^{\infty}
		\frac{	q_{1}^{-r(n-1)}w^{-r}(q^m_1v)^r						}{r}
		\frac{	\mathfrak{q}^r	(	1 - q^r_2	)			}{1-q^r_1}
		\Big) : \tilde{\mathbb{I}}^{*}_n(w) \tilde{\mathbb{I}}_m(v):~,
		\label{mathToperatudosifuweioruewoiruewioudfsiofusdiofdsuo}
	\end{align}
	where $: \quad :$ denotes the normal ordering. 
	The next step is to determine the value of the shift parameter
	$ p $ such that we can express $ \mathbb{I}_n(pv) $ and $ \mathbb{I}_n^*(pv) $ by the action of $ \mathcal{T} $ operator. 
	Consider the quantity $ \mathbb{I}_{n}(pv) $. From the explicit expression of the vector intertwiner explained in section \ref{vecint}, 
	we see that 
	\begin{align}
		\mathbb{I}_n(pv)
		=&\frac{
			z_n(u,N | pv)
		}{
			z_n(u,N | v)
		} \cdot\exp\Big(
		-\sum_{r=1}^{\infty}\frac{H_{-r}}{\gq^r - \gq^{-r}}
		\frac{	(		\mathfrak{q}^{-1/2}q^n_1v		)^r			}{1-q^r_1}
		(	p^r - 1 				)
		\Big)
		\nn \\
		& \cdot \mathbb{I}_n(v) \cdot
		\exp\Big(
		\sum_{r=1}^{\infty}
		\frac{	H_r			}{\gq^r - \gq^{-r}}
		\frac{		(	\mathfrak{q}^{1/2}q^{n-1}_{1}v		)^{-r}			}{   1 - q^r_1  }
		(	1 - p^{-r}				)
		\Big). \label{pshift} 
	\end{align}

	To express the exponential factors in \eqref{pshift} by $ \mathcal{T} $ operator, 
	let us consider Eq.\eqref{mathToperatudosifuweioruewoiruewioudfsiofusdiofdsuo} in the case $ m = n  $. 
	\begin{align}
		\mathcal{T}^{n}_{n}(N,u|v,w) =
		&z^{*}_n\big(	u, N		\big| w \big)z_n\big(	u, N		\big| v \big)
		\exp\Big(
		-\sum_{r=1}^{\infty}
		\frac{	q_{1}^{r}w^{-r}v^r					}{r}
		\frac{	\mathfrak{q}^r	(	1 - q^r_2	)			}{1-q^r_1}
		\Big)
		\nn \\
		&\cdot \exp\Big(
		\sum_{r=1}^{\infty}
		\frac{H_{-r}}{(\gq^r - \gq^{-r})(1-q^r_1)}
		\big(
		(	\mathfrak{q}^{1/2}q^n_1w				)^r - (	 \mathfrak{q}^{-1/2}q^n_1v				)^r
		\big)
		\Big)
		\nn \\
		&\cdot 
		\exp\Big(
		\sum_{r=1}^{\infty}
		\frac{H_r}{(\gq^r - \gq^{-r})(1-q^r_1)}
		\big(
		(	\mathfrak{q}^{-1/2}q^{n-1}_1w		)^{-r} - 
		(	\mathfrak{q}^{1/2}q^{n-1}_{1}v			)^{-r}
		\big)
		\Big).
	\end{align}
	Thus, if $ w = \mathfrak{q}v $, the positive modes disappear 
	\begin{align}
		\mathcal{T}^{n}_{n}(N,u|v,\mathfrak{q}v) =
		& \exp\Big(
		- \sum_{r=1}^{\infty}
		\frac{q^r_1}{r}
		\frac{	1 - q_2^r				}{	1 - q^r_1				}
		\Big)
		\cdot 
		\exp\Big(
		\sum_{r=1}^{\infty} \frac{H_{-r} \gq^{r/2}}{1 - q^r_1}
		(	q^n_1v		)^r
		\Big).
		\label{seventennfsofudsiofudsoifudsoifudsofidsfusdfs}
	\end{align}
	where we have used $z^{*}_n\big(	u, N		\big| \gq^{\pm 1} v \big)z_n\big(	u, N		\big| v \big) =1$ for $\gamma=1$. 
	On the other hand, if $ w = \gq^{-1}v $, the negative modes disappear 
	\begin{align}
		\mathcal{T}^{n}_{n}(N,u|v,\gq^{-1}v) =
		&\exp\Big(
		-\sum_{r=1}^{\infty}
		\frac{	q^r_1\gq^{2r}(1-q^r_2)						}{	r(1 - q^r_1)			}
		\Big)
		\cdot \exp\Big(
		\sum_{r=1}^{\infty}
		\frac{H_{r}\gq^{r/2}}{1 - q^r_1}
		(	q^{n-1}_1v			)^{-r}
		\Big). 
		\label{eightteendsfjdsofusdoifudsoifudsfosdufoisdu}
	\end{align}
	
	By tuning the spectral parameters of $\mathbb{I}^{*}_n(w)$ and $\mathbb{I}_m(v)$, 
	we define\footnote{See section \ref{section7} for the reason why we define them in this way.}
	\begin{align}
		\cals{T}^+_n(N,u | v) &:= \mathcal{T}^{n}_{n}(N,u|\gq^{1/2}v,\gq^{-1/2}v)^{-1},
		\notag \\ 
		\cals{T}^-_n(N,u | v) &:= \mathcal{T}^{n}_{n}(N,u|\gq^{-1/2}v,\gq^{1/2}v).
		\label{defi5.6}
	\end{align}
	Thus, 
	\begin{align}
		\mathbb{I}_n(\gq^{-2}v) =
		& \frac{
			z_n(u,N | \gq^{-2} v)
		}{
			z_n(u,N | v)
		} \cdot \exp
		\bigg(
		\sum_{r=1}^{\infty}
		\frac{q^r_1}{r}
		\frac{	1 - q^r_2			}{	1 - q^r_1		}(		1 - \mathfrak{q}^{2r}		)
		\bigg)
		\notag \\ 
		&
		\cdot \mathcal{T}^{-}_{n}(N,q_2^{-1}u|\mathfrak{q}^{-3/2}v) \cdot 
		\mathbb{I}_n(v)
		\cdot \mathcal{T}^{+}_{n}(N,u|\gq^{-1/2}v). 
		\label{sdweoruiodsufiosudoiuf}
	\end{align}
	It is remarkable that the shift parameter is fixed to be $p=\gq^{-2}$. 
	In the above tuning the spectral parameters of the intertwiners and the dual intertwiners are
	shifted by $\gq^{\pm 1}$. This reminds us of exactly the same shift of the spectral parameters
	in the \say{Higgsed} DIM network \cite{Taki} \cite{Z18} \cite{FOS1} \cite{FOS2}\footnote{Upon
		the completion of the present work, we noticed \cite{Z20}, which shows the universal $R$ matrix of DIM algebra plays
		a significant role in the \say{Higgsed} network calculus. In particular it is interesting that the shift operator in the present paper 
		is physically identified with the crossing of branes in \cite{Z20}. }. It seems that this coincidence is related to the fact
	the shift operator \eqref{defi5.6} is obtained from the universal $R$ matrix as we show later in section \ref{section7}.

	Similarly from the explicit expression of the dual vector intertwiner explained in section \ref{dualintertwiner}, 
	we see that 
	\begin{align*}
		\mathbb{I}^{*}_{n}(pv)
		=
		&\frac{
			z^*_n(u,N | pv)
		}{
			z^*_n(u,N | v)
		} \cdot
		\exp\bigg[
		\sum_{r=1}^{\infty} \frac{H_{-r}}{\gq^r - \gq^{-r}}
		\frac{		\gq^{r/2}(q^n_1v			)^r				}{	1 - q^r_1			}
		(	p^r - 1	)
		\bigg]
		\\
		&
		\cdot \mathbb{I}^{*}_{n}(v) \cdot
		\exp\bigg[
		- \sum_{r=1}^{\infty}\frac{H_r}{\gq^r - \gq^{-r}}
		\frac{	\gq^{r/2}(q^{n}_1v		)^{-r}					}{	1 - q^{-r}_1			}
		(	p^{-r} - 1		)
		\bigg]. 
		\label{aqewriueworudioudsfoidsufdosifusdoifuociucxofudsofdsu}
	\end{align*}
	From Eqs.\eqref{seventennfsofudsiofudsoifudsoifudsofidsfusdfs} and \eqref{eightteendsfjdsofusdoifudsoifudsfosdufoisdu}, we get that 
	\begin{align}
		\mathbb{I}^{*}_{n}(\gq^{-2}v) =
		&\frac{
			z^*_n(u,N | \gq^{-2} v)
		}{
			z^*_n(u,N | v)
		} \cdot
		\exp\bigg[
		\sum_{r=1}^{\infty}
		\frac{		q^r_1( 1 - q^r_2 )				}{r ( 1 - q^r_1 )}
		(	\gq^{2r} - 1		)
		\bigg]
		\notag \\
		&
		\mathcal{T}^{-}_{n}(N, q_2u|\gq^{-1/2}v)^{-1} \cdot 
		\mathbb{I}^{*}_{n}(v) \cdot
		\mathcal{T}^{+}_{n}(N, u|\mathfrak{q}^{-3/2}v)^{-1}. 
	\end{align}

	\subsection{Commutation of vector intertwiner and the shift operators}
	
	From the previous subsection, we see that $ \mathcal{T} $ operators are expressed as the composition of the intertwiner and 
	the dual vector intertwiner. 
	Thus, to calculate the commutation of the vector intertwiner and $ \mathcal{T} $ operators, 
	we first calculate the commutation between the vector intertwiners and the dual vector intertwiners themselves.
	
	
	
	\begin{figure}[b]
		\unitlength 2mm
		\begin{center}
%
		\begin{tikzpicture}
			\begin{scope}[very thick, every node/.style={sloped,allow upside down}]
				\draw[color=blue!70!black] (0,0)-- node {\midarrow} (-3,0);
				\draw [color=blue!70!black](3,0)-- node {\midarrow} (0,0);
				\draw[color=green!35!black] (10,0) -- node {\midarrow} (10,-3);
			\end{scope}
			every node/.style=draw,
			every label/.style=draw
			]
			\node [label={[shift={(10.5,-2.55)}]$ w $}] {};
			\node [label={[shift={(9.2,-2.7)}]$ (0,0) $}] {};
			
			\node [label={[shift={(2.8,-1.0)}]$ (1,N) $}] {};
			\node [label={[shift={(-2.8,-1.0)}]$ (1,N) $}] {};
			
			\node [label={[shift={(10,0.04)}]$ \bb{I}^*_n(w) $}] {};
			\node [label={[shift={(2.8,0.04)}]$ q_2u $}] {};
			\node [label={[shift={(-2.8,0.04)}]$ u $}] {};
			
			\node [label={[shift={(5,-0.45)}]$ \comp $}] {};
			
			\begin{scope}[very thick, every node/.style={sloped,allow upside down}]
				\draw (10,0)[color=blue!70!black]-- node {\midarrow} (7,0);
				\draw (13,0)[color=blue!70!black]-- node {\midarrow} (10,0);
				\draw[color=green!35!black] (0,3) -- node {\midarrow} (0,0);
			\end{scope}
			\node [label={[shift={(0.5,2.35)}]$ v $}] {};
			\node [label={[shift={(-0.8,2.25)}]$ (0,0) $}] {};
			
			\node [label={[shift={(12.8,-1.0)}]$ (1,N) $}] {};
			\node [label={[shift={(7.2,-1.0)}]$ (1,N) $}] {};
			
			\node [label={[shift={(0,-1.0)}]$ \bb{I}_m(v) $}] {};
			\node [label={[shift={(12.8,0.04)}]$ u $}] {};
			\node [label={[shift={(7.2,0.04)}]$ q_2u $}] {};
		\end{tikzpicture}
		\end{center}
		\caption{Composition $\circ$ of the dual vector intertwiner and the vector intertwiner. 
			Contrary to the Fock intertwiner, the level $(1,N)$ of the horizontal Fock space is kept intact
			and the shift of the spectral parameter $u$ is independent of the spectral parameters
			$v,w$ of the vertical representation. }
		\label{figuretwofsdofusdoifusdoifusdiofusouoidsuf}
	\end{figure}
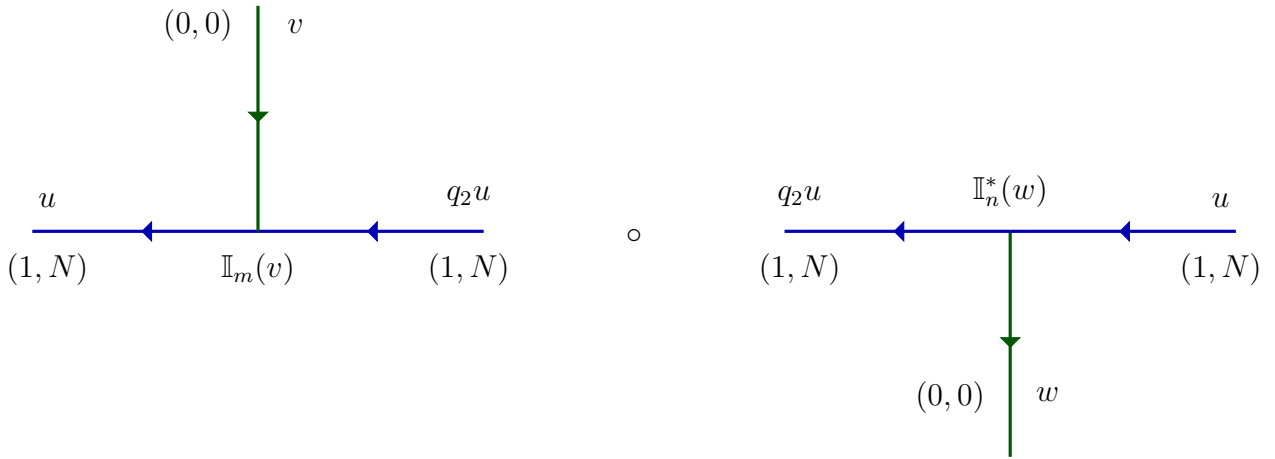
	

	The result of the composition $ \mathbb{I}^{*}_n(w)\mathbb{I}_m(v) $ of Fig.\ref{dsoivoidfsuosiufdsoifusio}
	has already given in eq.\eqref{mathToperatudosifuweioruewoiruewioudfsiofusdiofdsuo}.
	On the other hand the composition $ \mathbb{I}_m(v)\mathbb{I}^{*}_n(w) $ 
	of Fig.\ref{figuretwofsdofusdoifusdoifusdiofusouoidsuf} is
	\begin{align}
		\mathbb{I}_m(v)\mathbb{I}^{*}_n(w) 
		= z_m&\big( q_2 u, N \big| v		\big)z^{*}_{n}(	q_2 u, N \big| w		)
		\nn \\
		& \cdot \exp\bigg(
		-\sum_{r=1}^{\infty}
		\frac{1}{r}
		(	\gq^{1/2}q^{m-1}_1v			)^{-r}( \gq^{1/2}q^n_1w )^{r}
		\gq^r
		\frac{	1 - q^r_2			}{	1 - q^r_1				}
		\bigg) : \tilde{\mathbb{I}}_m(v) \tilde{\mathbb{I}}^{*}_n(w):~.
	\end{align}
	Since 
	\begin{align}\label{propertyofstackzeromodes}
		z_m\big(	q_2 u, N \big| v			\big) = q_2^m z_m\big(u,N\big|v\big), \qquad
		z^{*}_n\big(	q_2 u, N \big| w		\big) = q_2^{-n}z^{*}_n(u,N\big|w),
	\end{align}
	we get that 
	\begin{align}
		\mathbb{I}_m(v)\mathbb{I}^{*}_n(w)
		= q_2^{m-n}\exp\bigg(
		\sum_{r=1}^{\infty}
		\frac{	q^r_1 \gq^r		}{r}
		\frac{	1-q^r_2		}{1-q^r_1}
		\big[
		q^{-rn}_1q^{mr}_1(v/w)^r - q^{-mr}_1q^{nr}_1
		(w/v)^r\big]
		\bigg)
		\mathbb{I}^{*}_{n}(w)\mathbb{I}_m(v). 
		\label{rtipoipodfidspofidspofisopfisdopfidspof}
	\end{align}

	\begin{figure}[h]
		\unitlength 2mm
		\begin{center}
		\begin{tikzpicture}
			\begin{scope}[very thick, every node/.style={sloped,allow upside down}]
				\draw[color=blue!70!black] (0,0)-- node {\midarrow} (-3,0);
				\draw [color=blue!70!black](3,0)-- node {\midarrow} (0,0);
				\draw[color=green!35!black] (0,3) -- node {\midarrow} (0,0);
			\end{scope}
			every node/.style=draw,
			every label/.style=draw
			]
			\node [label={[shift={(0.5,2.35)}]$ v^{\prime} $}] {};
			\node [label={[shift={(-0.8,2.25)}]$ (0,0) $}] {};
			
			\node [label={[shift={(2.8,-1.0)}]$ (1,N) $}] {};
			\node [label={[shift={(-2.8,-1.0)}]$ (1,N) $}] {};
			
			\node [label={[shift={(0,-1.0)}]$ \bb{I}_m(v^{\prime}) $}] {};
			\node [label={[shift={(2.8,0.04)}]$ q_2^{-1}u $}] {};
			\node [label={[shift={(-2.8,0.04)}]$ q_2^{-2}u $}] {};
			
			\node [label={[shift={(5,-0.45)}]$ \comp $}] {};
			
			\begin{scope}[very thick, every node/.style={sloped,allow upside down}]
				\draw (10,0)[color=blue!70!black]-- node {\midarrow} (7,0);
				\draw (13,0)[color=blue!70!black]-- node {\midarrow} (10,0);
				\draw[color=green!35!black] (10,3) -- node {\midarrow} (10,0);
			\end{scope}
			\node [label={[shift={(10.5,2.35)}]$ v $}] {};
			\node [label={[shift={(9.2,2.25)}]$ (0,0) $}] {};
			
			\node [label={[shift={(12.8,-1.0)}]$ (1,N) $}] {};
			\node [label={[shift={(7.2,-1.0)}]$ (1,N) $}] {};
			
			\node [label={[shift={(10,-1.0)}]$ \bb{I}_n(v) $}] {};
			\node [label={[shift={(12.8,0.04)}]$ u $}] {};
			\node [label={[shift={(7.2,0.04)}]$ q_2^{-1}u $}] {};
		\end{tikzpicture}
		\end{center}
		\caption{Composition $\circ$ of the vector intertwiner and the vector intertwiner. 
			Contrary to the Fock intertwiner, the level $(1,N)$ of the horizontal Fock space is kept intact
			and the shift of the spectral parameter $u$ is independent of the spectral parameters
			$v,v^{\prime}$ of the vertical representation. }
		\label{vectorintertwinervsvectorintertwiner1}
	\end{figure}

	Next let us calculate $ \mathbb{I}_m(v^{\prime})\mathbb{I}_n(v) $ 
	of Fig.\ref{vectorintertwinervsvectorintertwiner1}. 
	The result is 
	\begin{align} 
		\mathbb{I}_m(v^{\prime})\mathbb{I}_n(v)
		=
		z_{m}&\big(	q_2^{-1}u, N \big| v^{\prime}			\big)z_n(u,N \big| v) \nn \\
		& \cdot \exp\bigg( 
		\sum_{r=1}^{\infty}
		(	q^{m-1}_{1}v^{\prime}			)^{-r}(		q^{n}_{1}v				)^{r}
		\frac{1}{r} \cdot 
		\frac{		1 - q^r_2					}{		1 - q^r_1			}
		\bigg) : \tilde{\mathbb{I}}_m(v^{\prime}) \tilde{\mathbb{I}}_n(v):~.
	\end{align}
	
	The formula for $\mathbb{I}_n(v)\mathbb{I}_m(v^{\prime})$ is obtained by
	interchanging the parameters $ v \leftrightarrow v^{\prime} $, and $ n \leftrightarrow m $.
	Hence, by \eqref{propertyofstackzeromodes}, we see that 
	\begin{align}
		\mathbb{I}_n(v)\mathbb{I}_m(v^{\prime}) 
		= q_2^{-n+m}\exp\bigg(
		\sum_{r=1}^{\infty}\frac{q^r_1}{r}\frac{1-q^r_2}{1-q^r_1}
		\big[
		\frac{		(q^m_1v^{\prime})^r			}{	(	q^n_1v		)^r		}
		-
		\frac{			(	q^n_1v		)^r				}{		(q^m_1v^{\prime})^r			}
		\big]
		\bigg)
		\mathbb{I}_m(v^{\prime})\mathbb{I}_n(v). 
		\label{qweopqwiepoifdopsfidspofidspfdsifpodsifsp}
	\end{align}
	A similar computation gives
	\begin{align}
		\mathbb{I}^{*}_{m}(w^{\prime})\mathbb{I}^{*}_{n}(w)
		=
		q_2^{n-m}
		\exp\bigg(
		\sum_{r=1}^{\infty}
		\frac{		\gq^{2r}q^r_1				}{r}
		\frac{	1 - q^r_2				}{1 - q^r_1}
		\big[
		\frac{	(q^n_1w)^r		}{	(q^m_1w^{\prime})^r		} - \frac{(q^m_1w^{\prime})^r}{		(q^n_1w)^r			}
		\big]
		\bigg)
		\mathbb{I}^{*}_{n}(w)\mathbb{I}^{*}_{m}(w^{\prime}). 
		\label{woeudsofudsiofudsiofdsufodsufdoisfusdio}
	\end{align}
	By introducing the function
	\begin{align}
		\Upsilon\big(	\alpha \big| x			\big)
		= \exp\bigg(
		\sum_{r=1}^{\infty}
		\frac{	\alpha^r(1-q^r_2)				}{	r(1 - q^r_1)			}
		(	x^r - x^{-r}	)
		\bigg),
		\label{definitiondsofdsufiosdufodisfudsoifuosiduf}
	\end{align}
	we can write the equations 
	\eqref{rtipoipodfidspofidspofisopfisdopfidspof},
	\eqref{qweopqwiepoifdopsfidspofidspfdsifpodsifsp} and
	\eqref{woeudsofudsiofudsiofdsufodsufdoisfusdio}
	more compactly;
	\begin{align}
		\mathbb{I}_m(v)\mathbb{I}^{*}_n(w)
		&= q_2^{m-n}
		\Upsilon\big(
		q_1\gq~\big|~\frac{v}{w}q^{m-n}_{1}
		\big)
		\mathbb{I}^{*}_{n}(w)\mathbb{I}_m(v),
		\label{34sdfiusfodsufiodsfudisofusdo}
		\\
		\mathbb{I}_n(v)\mathbb{I}_m(v^{\prime}) 
		&= q_2^{-n+m}
		\Upsilon\big(
		q_1~\big|~q^{m-n}_1\frac{v^{\prime}}{v}
		\big)
		\mathbb{I}_m(v^{\prime})\mathbb{I}_n(v), 
		\label{35sdofudsoifudsoifusiofuiouoisdfods}
		\\
		\mathbb{I}^{*}_{m}(w^{\prime})\mathbb{I}^{*}_{n}(w)
		&=
		q_2^{n-m}
		\Upsilon\big(
		\gq^2q_1~\big|~q_1^{n-m}\frac{w}{w^{\prime}}
		\big)
		\mathbb{I}^{*}_{n}(w)\mathbb{I}^{*}_{m}(w^{\prime}). 
		\label{36rwepriewopidspofidsopfidspofisdpofspdif}
	\end{align}

	
	Now we are ready to calculate the commutation relation between
	$ \mathcal{T} $ operator and the vector intertwiner $ \mathbb{I}_n(v) $. 
	Consider the composition $ \mathcal{T}^{m}_{n}\big(	N, q^{-1}u \big| v, w			\big)\mathbb{I}_k(z) $ 
	corresponding to Fig.\ref{figure6naja} below. 
	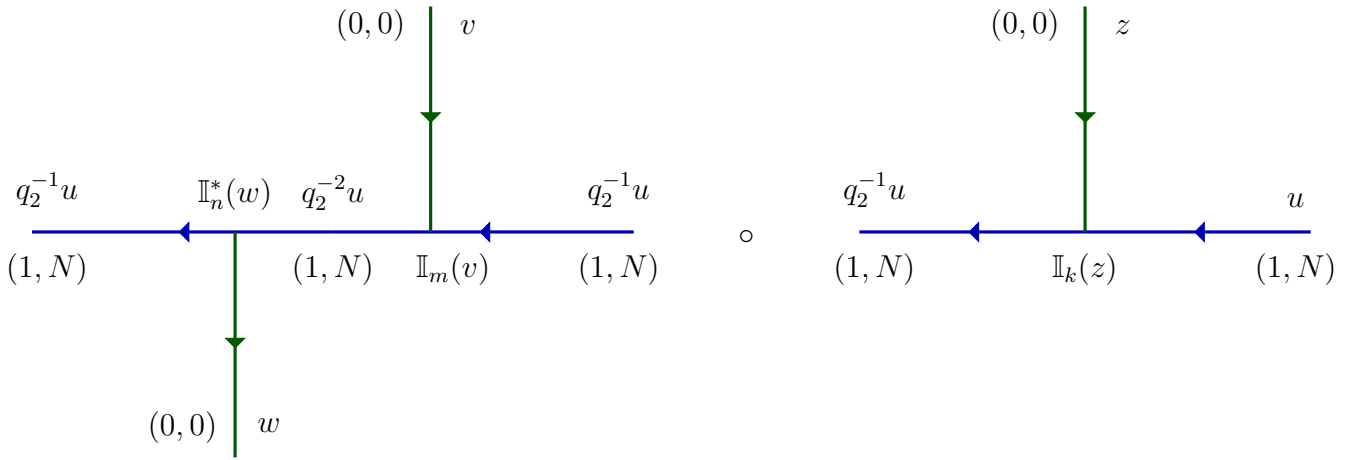
\begin{figure}[h]
		\unitlength 3.3mm
		\begin{center}
		\begin{tikzpicture}
			\begin{scope}[very thick, every node/.style={sloped,allow upside down}]
				\draw[color=blue!70!black] (0,0)-- node {\midarrow} (-4,0);
				\draw [color=blue!70!black](4,0)-- node {\midarrow} (0,0);
				\draw[color=green!35!black] (1.3,3) -- node {\midarrow} (1.3,0);
				\draw[color=green!35!black] (-1.3,0) -- node {\midarrow} (-1.3,-3);
			\end{scope}
			every node/.style=draw,
			every label/.style=draw
			]
			\node [label={[shift={(1.8,2.35)}]$ v $}] {};
			\node [label={[shift={(0.5,2.25)}]$ (0,0) $}] {};
			
			\node [label={[shift={(-1.3,0.01)}]$ \bb{I}^*_n(w) $}] {};
			\node [label={[shift={(-0.85,-2.95)}]$ w $}] {};
			\node [label={[shift={(-2.0,-3.1)}]$ (0,0) $}] {};
			
			\node [label={[shift={(0,0)}]$ q_2^{-2}u $}] {};
			\node [label={[shift={(0,-1.0)}]$ (1,N) $}] {};
			
			\node [label={[shift={(3.8,-1.0)}]$ (1,N) $}] {};
			\node [label={[shift={(-3.8,-1.0)}]$ (1,N) $}] {};
			
			\node [label={[shift={(1.6,-1.0)}]$ \bb{I}_m(v) $}] {};
			\node [label={[shift={(3.8,0.04)}]$ q_2^{-1}u $}] {};
			\node [label={[shift={(-3.8,0.04)}]$ q_2^{-1}u $}] {};
			
			\node [label={[shift={(5.5,-0.45)}]$ \comp $}] {};
			
			\begin{scope}[very thick, every node/.style={sloped,allow upside down}]
				\draw (10,0)[color=blue!70!black]-- node {\midarrow} (7,0);
				\draw (13,0)[color=blue!70!black]-- node {\midarrow} (10,0);
				\draw[color=green!35!black] (10,3) -- node {\midarrow} (10,0);
			\end{scope}
			\node [label={[shift={(10.5,2.35)}]$ z $}] {};
			\node [label={[shift={(9.2,2.25)}]$ (0,0) $}] {};
			
			\node [label={[shift={(12.8,-1.0)}]$ (1,N) $}] {};
			\node [label={[shift={(7.2,-1.0)}]$ (1,N) $}] {};
			
			\node [label={[shift={(10,-1.0)}]$ \bb{I}_k(z) $}] {};
			\node [label={[shift={(12.8,0.04)}]$ u $}] {};
			\node [label={[shift={(7.2,0.04)}]$ q_2^{-1}u $}] {};
		\end{tikzpicture}
		\end{center}
		\caption{Composition $\circ$ of the vector intertwiner and the shift operator. }
		\label{figure6naja}
	\end{figure}
	
	Then, by using \eqref{34sdfiusfodsufiodsfudisofusdo} and \eqref{35sdofudsoifudsoifusiofuiouoisdfods}, we obtain 
	\begin{align}
		\mathcal{T}^{m}_{n}\big(	N, q^{-1}_2u \big| v, w			\big)\mathbb{I}_k(z)
		&= \mathbb{I}^{*}_{n}(w)\mathbb{I}_m(v)\mathbb{I}_k(z)
		\notag 
		= \mathbb{I}^{*}_{n}(w)
		\bigg(
		q_2^{-m+k}\Upsilon\big(	q_1 \big| q^{k-m}_{1}\frac{z}{v}		\big)
		\bigg)
		\mathbb{I}_k(z)\mathbb{I}_m(v)
		\\
		&= q_2^{-m+k}\Upsilon\big( q_1
		\big|
		q_1^{k-m}\frac{z}{v}
		\big)
		q_2^{n-k}\Upsilon\big( q_1\gq
		\big|
		q^{k-n}_{1}\frac{z}{w}
		\big)^{-1}
		\mathbb{I}_k(z)\mathbb{I}^{*}_{n}(w)\mathbb{I}_m(v)
		\notag 
		\\
		&= q_2^{n-m}\Upsilon\big( q_1
		\big|
		q_1^{k-m}\frac{z}{v}
		\big)
		\Upsilon\big( q_1\gq
		\big|
		q^{n-k}_{1}\frac{w}{z}
		\big)
		\mathbb{I}_k(z)
		\mathcal{T}^{m}_{n}\big( N, u 
		\big|
		v, w 
		\big).
		\label{red37}
	\end{align}
	Here we have used the property
	$ \Upsilon\big(\alpha\big|x\big)^{-1} = \Upsilon\big(\alpha\big|\frac{1}{x}\big) $ which follows directly from the definition \eqref{definitiondsofdsufiosdufodisfudsoifuosiduf}. 
	Note that, during the calculation, we do not need to be careful too much about the horizontal spectral parameters,
	since we have dealt with this issue when we calculated the commutation of the vector intertwiner and dual vector intertwiner. 
	The only thing we need to be careful is the initial point, 
	$ \mathcal{T}^{m}_{n}\big(	N, q_2^{-1}u \big| v, w			\big)\mathbb{I}_k(z) $, 
	and the final point, $ \mathbb{I}_k(z)
	\mathcal{T}^{m}_{n}\big( N, u 
	\big|
	v, w 
	\big) $.

	Now consider the composition $ \mathcal{T}^{m}_{n}\big(	N, q_2 u \big| v,w	\big)\mathbb{I}^{*}_{k}(z) $ 
	corresponding to Fig.\ref{figurednaja} below.
	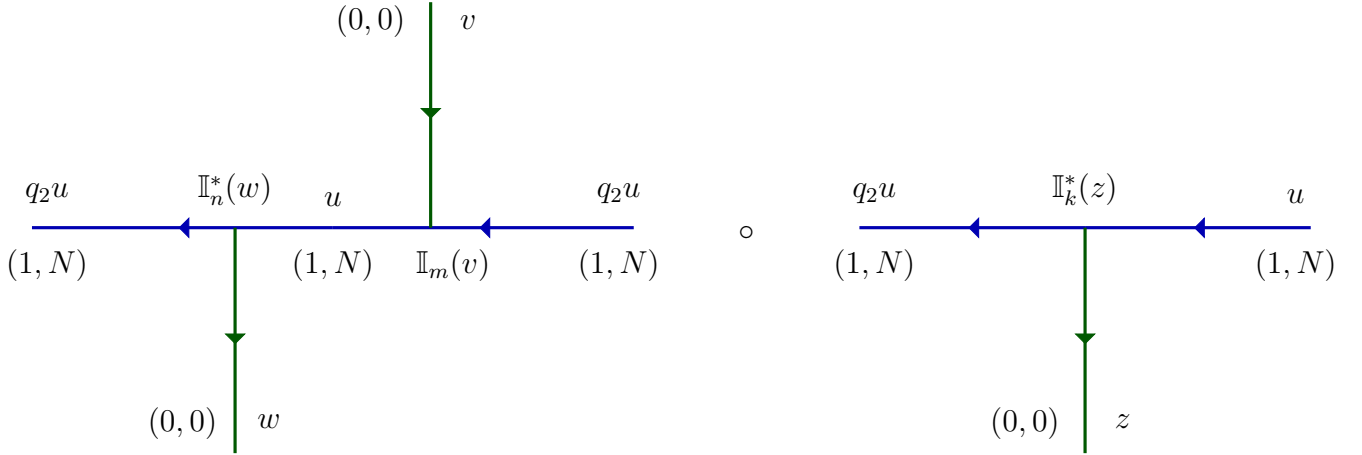
\begin{figure}[h]
		\unitlength 3.3mm
		\begin{center}
		\begin{tikzpicture}
			\begin{scope}[very thick, every node/.style={sloped,allow upside down}]
				\draw[color=blue!70!black] (0,0)-- node {\midarrow} (-4,0);
				\draw [color=blue!70!black](4,0)-- node {\midarrow} (0,0);
				\draw[color=green!35!black] (1.3,3) -- node {\midarrow} (1.3,0);
				\draw[color=green!35!black] (-1.3,0) -- node {\midarrow} (-1.3,-3);
			\end{scope}
			every node/.style=draw,
			every label/.style=draw
			]
			\node [label={[shift={(1.8,2.35)}]$ v $}] {};
			\node [label={[shift={(0.5,2.25)}]$ (0,0) $}] {};
			
			\node [label={[shift={(-1.3,0.01)}]$ \bb{I}^*_n(w) $}] {};
			\node [label={[shift={(-0.85,-2.95)}]$ w $}] {};
			\node [label={[shift={(-2.0,-3.1)}]$ (0,0) $}] {};
			
			\node [label={[shift={(0,0)}]$ u $}] {};
			
			\node [label={[shift={(3.8,-1.0)}]$ (1,N) $}] {};
			\node [label={[shift={(-3.8,-1.0)}]$ (1,N) $}] {};
			\node [label={[shift={(0,-1.0)}]$ (1,N) $}] {};
			
			\node [label={[shift={(1.6,-1.0)}]$ \bb{I}_m(v) $}] {};
			\node [label={[shift={(3.8,0.04)}]$ q_2u $}] {};
			\node [label={[shift={(-3.8,0.04)}]$ q_2u $}] {};
			
			\node [label={[shift={(5.5,-0.45)}]$ \comp $}] {};
			
			\begin{scope}[very thick, every node/.style={sloped,allow upside down}]
				\draw (10,0)[color=blue!70!black]-- node {\midarrow} (7,0);
				\draw (13,0)[color=blue!70!black]-- node {\midarrow} (10,0);
				\draw[color=green!35!black] (10,0) -- node {\midarrow} (10,-3);
			\end{scope}
			\node [label={[shift={(10.5,-2.95)}]$ z $}] {};
			\node [label={[shift={(9.2,-3.1)}]$ (0,0) $}] {};
			
			\node [label={[shift={(12.8,-1.0)}]$ (1,N) $}] {};
			\node [label={[shift={(7.2,-1.0)}]$ (1,N) $}] {};
			
			\node [label={[shift={(10,0.01)}]$ \bb{I}^*_k(z) $}] {};
			\node [label={[shift={(12.8,0.04)}]$ u $}] {};
			\node [label={[shift={(7.2,0.04)}]$ q_2u $}] {};
		\end{tikzpicture}
		\end{center}
		\caption{Composition $\circ$ of the dual vector intertwiner and the shift operator. }
		\label{figurednaja}
	\end{figure}
	We can compute
	\begin{align}
		\mathcal{T}^{m}_{n}\big(
		N,q_2 u \big| v, w 
		\big)
		\mathbb{I}^{*}_{k}(z)
		&=
		\mathbb{I}^{*}_{n}(w)\mathbb{I}_m(v)\mathbb{I}^{*}_{k}(z)
		=
		q_2^{m-k}\Upsilon\big(
		q_1\gq \big| \frac{v}{z}q^{m-k}_{1}
		\big)
		\mathbb{I}^{*}_{n}(w)\mathbb{I}^{*}_{k}(z)\mathbb{I}_m(v)
		\notag 
		\\
		&= q_2^{m-k}\Upsilon\big( q_1\gq
		\big|\frac{v}{z}q^{m-k}_{1}
		\big)
		q_2^{k-n}\Upsilon\big(\gq^2q_1
		\big|
		q^{k-n}_{1}\frac{z}{w}\big)
		\mathbb{I}^{*}_{k}(z)\mathbb{I}^{*}_{n}(w)\mathbb{I}_m(v)
		\notag
		\\
		&= q_2^{m-n}\Upsilon\big( q_1\gq
		\big|\frac{v}{z}q^{m-k}_{1}
		\big)
		\Upsilon\big(\gq^2q_1
		\big|
		q^{k-n}_{1}\frac{z}{w}\big)
		\mathbb{I}^{*}_{k}(z)
		\mathcal{T}^{m}_{n}\big(
		N, u \big| v, w 
		\big). 
		\label{red38}
	\end{align}
	
	From Eqs. \eqref{red37}, \eqref{red38} and the definition \eqref{defi5.6} we see that
	\begin{align}
		\cals{T}^+_n(N, u \big| v)\bb{I}_k(z) &= 
		\frac{		\Upsilon(	q_1 \big| q^{n-k}_1\gq^{1/2}\frac{v}{z}		)			}{	\Upsilon(q_1\gq \big| q^{n-k}_1\gq^{-1/2}\frac{v}{z})			}\bb{I}_k(z)\cals{T}^+_n(N, q_2u \big| v),
		\label{pro5.20}
		\\
		\cals{T}^-_n(N, u \big| v)\bb{I}_k(z) &=
		\frac{	\Upsilon(	q_1 \big| q^{k-n}_1\gq^{1/2}\frac{z}{v})}
		{	\Upsilon(q_1\gq \big| q^{k-n}_1\gq^{-1/2}\frac{z}{v})					}
		\bb{I}_k(z)\cals{T}^-_n(N, q_2u \big| v),
		\label{pro5.21}
		\\
		\cals{T}^+_n(N, u \big| v)\bb{I}^*_k(z) &=
		\frac {	\Upsilon(	q_1\gq \big| q^{k-n}_1\gq^{-1/2}\frac{z}{v}		)		
		}{	\Upsilon(	q_1\gq^2 \big| q^{k-n}_1\gq^{1/2}\frac{z}{v}		)		
		}
		\bb{I}^*_k(z)\cals{T}^+_n(N, q_2^{-1}u \big| v),
		\label{pro5.22}
		\\
		\cals{T}^-_n(N, u \big| v)\bb{I}^*_k(z) &=
		\frac{
			\Upsilon(   q_1\gq \big| q^{n-k}_1\gq^{-1/2}\frac{v}{z}				)}
		{ \Upsilon(	q_1\gq^2 \big| q^{n-k}_1\gq^{1/2}\frac{v}{z}				)}
		\bb{I}^*_k(z)\cals{T}^-_n(N, q^{-1}_2u \big| v). 
		\label{pro5.23}
	\end{align}
	The function $\Upsilon(\alpha \big| x)$ involves both the positive and the negative powers of the parameter $x$.
	However the combinations appearing in Eqs.\eqref{pro5.20} -- \eqref{pro5.23} have only the positive or the negative
	powers. Hence, we expect that the proportional factors in these equations
	can be written in terms of the $ R $-matrix of the vector representation.
	We will see this is in fact the case in the next subsection.

	\subsection{$ R $-matrix of the vector representation}
	
	We can define the $ R $-matrix of the vector representation of the DIM algebra 
	in terms of the universal $ R $-matrix $ \cals{R}_0 $ \eqref{universalR};
	\begin{gather}
		\bigg[\big(	\rho^V_{v_1} \tens \rho^V_{v_2}				\big)\big(	\cals{R}_0		\big)\bigg]
		\bigg[		[v_1]_m \tens [v_2]_l					\bigg]
		= R_{ml}\big(	\frac{v_1}{v_2}				\big)\bigg[		[v_1]_m \tens [v_2]_l					\bigg]. 
	\end{gather}
	Thus, to determine the $ R_{ml}(\frac{v_1}{v_2}) $ we have to calculate the following quantity : 
	\begin{gather}
		\exp\bigg(
		\sum_{r=1}^{\infty}r \kappa_r \rho^V_{v_1}(h_{-r}) \tens \rho^V_{v_2}(h_r)
		\bigg)\bigg[
		[v_1]_m \tens [v_2]_l
		\bigg].
	\end{gather}
	From Eqs.\eqref{2.12},  \eqref{vecKp} and \eqref{vecKm}, we see that 
	\begin{equation}
		\rho^V_v(H_{\pm r})[v]_i = \pm \frac{(q^i_1v)^{\pm r}}{r}q^{(1\pm 1)r/2}_1(q^{\pm r}_2 - 1)(q^{\pm r}_3 - 1)[v]_i.
	\end{equation}
	Hence, 
	\begin{equation}\label{Lemma512}
		\rho^V_{v}(h_{\pm r})[v]_i = \pm \frac{1}{\kappa_r}\rho^V_{v}(H_{\pm r}) [v]_i
		=  \frac{(q^i_1v)^{\pm r}}{r \kappa_r}q^{(1\pm 1)r/2}_1(q^{\pm r}_2 - 1)(q^{\pm r}_3 - 1)[v]_i.
	\end{equation}
	Accordingly, we obtain 
	\begin{align}\label{6.39}
		R_{ml}\Big(		\frac{v_1}{v_2}			\Big)
		&= \exp\bigg[
		- \sum_{r=1}^{\infty}\frac{1}{r}\Big(	q^{m-l}_{1}\frac{v_1}{v_2}			\Big)^{-r}
		\frac{	q^r_1(1-q^r_2)(1-q^r_3)			}{1-q^r_1}
		\bigg]
		\nn \\
		&= \frac{\bigg(	q^{l-m+1}_{1}\frac{v_2}{v_1}			; q_1		\bigg)_{\infty}
			\bigg(q^{l-m}_{1}\frac{v_2}{v_1} ; q_1
			\bigg)_{\infty}}
		{\bigg(	\frac{q^{l-m}_{1}}{q_2}\frac{v_2}{v_1}	; q_1		\bigg)_{\infty}
			\bigg(	\frac{q^{l-m}_{1}}{q_3}\frac{v_2}{v_1}	; q_1		\bigg)_{\infty}},
	\end{align}
	where we have used the formula \eqref{q-factorial}.
	
	By direct calculation, it is easy to see that the following relation 
	\begin{gather}\label{U-Rrelation}
		R_{nl}(x) = \frac{\Upsilon\big(	q_1\gq^{m} \big| q^{n-l}_1\gq^{m}x			\big)		}
		{	\Upsilon\big(	q_1\gq^{m+1} \big| q^{n-l}_1\gq^{m-1}x			\big)					}
	\end{gather}
	holds for any $ m \in \bb{Z} $. Combining \eqref{U-Rrelation} with $ \Upsilon(\alpha\big|x)^{-1} = \Upsilon(\alpha\big| x^{-1}) $, 
	we can write Eqs.\eqref{pro5.20} -- \eqref{pro5.23} in terms of the $R$ matrix;
	\begin{align}
		\cals{T}^+_n(N, u \big| v)\bb{I}_k(z) &= 
		R_{nk}(	 \gq^{1/2}\frac{v}{z}	)\bb{I}_k(z)\cals{T}^+_n(N, q_2u \big| v),
		\label{prrrrro5.20}
		\\
		\cals{T}^-_n(N, u \big| v)\bb{I}_k(z) &=
		R_{kn}(	 \gq^{1/2}\frac{z}{v}				)
		\bb{I}_k(z)\cals{T}^-_n(N, q_2u \big| v),
		\label{prrrrrro5.21}
		\\
		\cals{T}^+_n(N, u \big| v)\bb{I}^*_k(z) &=
		\frac{1}{
			R_{nk}(	 \gq^{-1/2}\frac{v}{z}			)
		}
		\bb{I}^*_k(z)\cals{T}^+_n(N, q_2^{-1}u \big| v),
		\label{prrrrrro5.22}
		\\
		\cals{T}^-_n(N, u \big| v)\bb{I}^*_k(z) &=
		\frac{1}{
			R_{kn}(	 \gq^{-1/2}\frac{z}{v}					)
		}
		\bb{I}^*_k(z)\cals{T}^-_n(N, q^{-1}_2u \big| v). 
		\label{prrrrrro5.23}
	\end{align}
	
	\subsection{KZ equation for the vector intertwiners}
	
	We define the correlation function of vector intertwiners and its duals by 
	\begin{align}
		\mathcal{G}_{\lambda_1,\cdots,\lambda_n}^{\mu_1,\cdots,\mu_m}
		\big(	&u_n \big|~{w_1,\cdots,w_m, z_1,\cdots,z_n} 		\big)
		\notag \\
		:= 
		&\langle\emptyset|
		\mathbb{I}^{*}_{\mu_1}\big(	q_2^{-(n-m)}u_n \big| w_1			\big)
		\cdots \mathbb{I}^{*}_{\mu_{m-1}}\big(	q_2^{-(n-2)}u_n \big| w_{m-1}				\big)\mathbb{I}^{*}_{\mu_m}\big(
		q_2^{-(n-1)}u_n		\big| w_m
		\big)
		\notag
		\\
		&\cdot~
		\mathbb{I}_{\lambda_1}\big(	q_2^{-(n-1)}u_n \big| z_1		\big)
		\cdots\mathbb{I}_{\lambda_k}\big(	q_2^{-(n-k)}u_n	\big|		z_k
		\big)
		\cdots 
		\mathbb{I}_{\lambda_n}\big(	 u_n	\big|	z_n		\big)
		|\emptyset\rangle.
	\end{align}
	Consider the shift of the spectral parameter of the vector intertwiner;
	\begin{align}
		(\gq^{-2})^{z_k\frac{\partial}{	\partial z_k		}}
		&\mathcal{G}_{\lambda_1,\cdots,\lambda_n}^{\mu_1,\cdots,\mu_m}
		\big(	u_n \big|~{w_1,\cdots,w_m, z_1,\cdots,z_n} 		\big)
		= 
		\notag 
		\\
		\langle\emptyset|
		&\mathbb{I}^{*}_{\mu_1}\big(	q_2^{-(n-m)}u_n \big| w_1			\big)
		\cdots \mathbb{I}^{*}_{\mu_{m-1}}\big(	q_2^{-(n-2)}u_n \big| w_{m-1}				\big)\mathbb{I}^{*}_{\mu_m}\big(
		q_2^{-(n-1)}u_n		\big| w_m
		\big)
		\notag 
		\\
		&\cdot~
		\mathbb{I}_{\lambda_1}\big(	q_2^{-(n-1)}u_n \big| z_1		\big)
		\cdots\mathbb{I}_{\lambda_k}\big(	q_2^{-(n-k)}u_n	\big|	\gq^{-2}	z_k
		\big)
		\cdots 
		\mathbb{I}_{\lambda_n}\big(	 u_n	\big|	z_n		\big)
		|\emptyset\rangle. 
		\label{whalerpewipofdfisdpfosdpfsdis}
	\end{align}
	Here we explicitly write the information of incoming horizontal for each vector intertwiner. 
	The only independent horizontal spectral parameter is of the incoming horizontal of $ \mathbb{I}_{\lambda_n} $. 
	The other parameters are determined automatically by the property of vector intertwiners. 
	From \eqref{sdweoruiodsufiosudoiuf} we know that 
	\begin{align}
		\mathbb{I}_{\lambda_k}\big( &q_2^{-(n-k)}u_n
		\big| \gq^{-2}z_k
		\big)
		=
		\exp\bigg(
		\sum_{r=1}^{\infty}\frac{q^r_1}{r}\frac{1-q^r_2}{1-q^r_1}(1-\gq^{2r})
		\bigg)
		\notag 
		\\
		&\cdot~
		\cals{T}^{-}_{\lambda_k}\big(
		N,q^{-(n-k)}_{2}u_n \big| \gq^{-3/2}z_k
		\big) \cdot
		\mathbb{I}_{\lambda_k}\big(
		q_2^{-(n-k)}u_n \big| z_k
		\big) \cdot
		\cals{T}^{+}_{\lambda_k}\big(
		N, q^{-(n-k)}_{2}u_n \big| \gq^{-1/2}z_k
		\big).
		\label{40fortydsofusioreuosdufoidsufiof}
	\end{align}
	Hence, we can rewrite $ \mathbb{I}_{\lambda_k}\big(	\gq^{-(n-k)}u_n	\big|	\gq^{-2}	z_k
	\big) $ in \eqref{whalerpewipofdfisdpfosdpfsdis} by using \eqref{40fortydsofusioreuosdufoidsufiof}, which implies
	\begin{align}
		(\gq^{-2})^{z_k\frac{\partial}{	\partial z_k		}}
		&\mathcal{G}_{\lambda_1,\cdots,\lambda_n}^{\mu_1,\cdots,\mu_m}
		\big(	u_n \big|~{w_1,\cdots,w_m, z_1,\cdots,z_n} 		\big)
		\notag 
		\\
		= 
		\langle\emptyset|
		&\mathbb{I}^{*}_{\mu_1}\big(	q_2^{-(n-m)}u_n \big| w_1			\big)
		\cdots \mathbb{I}^{*}_{\mu_{m-1}}\big(	q_2^{-(n-2)}u_n \big| w_{m-1}				\big)\mathbb{I}^{*}_{\mu_m}\big(
		q_2^{-(n-1)}u_n		\big| w_m
		\big)
		\notag
		\\
		&\cdot~
		\mathbb{I}_{\lambda_1}\big(	q_2^{-(n-1)}u_n \big| z_1		\big)
		\cdots
		\mathbb{I}_{\lambda_{k-1}}\big(
		q_2^{-(n-k+1)}u_n  \big| z_{k-1}
		\big)
		\notag 
		\\
		&\cdot~
		\cals{T}^{-}_{\lambda_k}\big(
		N,q^{-(n-k+1)}_{2}u_n \big| \gq^{-3/2}z_k
		\big) \cdot
		\mathbb{I}_{\lambda_k}\big(
		q_2^{-(n-k)}u_n \big| z_k
		\big) \cdot
		\cals{T}^+_{\lambda_k}\big(
		N, q^{-(n-k)}_{2}u_n \big| \gq^{-1/2}z_k
		\big)
		\notag 
		\\
		&\cdot~
		\mathbb{I}_{\lambda_{k+1}}\big(	q_2^{-(n-k-1)}u_n	\big|	z_{k+1}
		\big)
		\cdots 
		\mathbb{I}_{\lambda_n}\big(	u_n \big|z_n\big)
		|\emptyset\rangle
		\notag \\
		&\cdot 
		\exp\bigg(
		\sum_{r=1}^{\infty}\frac{q^r_1}{r}\frac{1-q^r_2}{1-q^r_1}(1-\gq^{2r})
		\bigg).
		\label{42forweioudoisfuoiruewoiusdofiudsoewr}
	\end{align}
	
	Now by using \eqref{prrrrro5.20}, \eqref{prrrrrro5.21} and \eqref{prrrrrro5.23}, 
	we can transport $ \cals{T}^{-}_{\lambda_k}$ and $ \cals{T}^+_{\lambda_k}$
	in \eqref{42forweioudoisfuoiruewoiusdofiudsoewr} until they hit the left and the right vacua.
	Then, using \eqref{seventennfsofudsiofudsoifudsoifudsofidsfusdfs} and
	\eqref{eightteendsfjdsofusdoifudsoifudsfosdufoisdu}, we calculate
	$ \langle\emptyset|
	\cals{T}^{-}_{\lambda_k}\big( N, q^{m-n}_{2}u_n
	\big| \gq^{-3/2}z_k
	\big) $ and
	$ \cals{T}^{+}_{\lambda_k}\big(
	N, u_n			\big| \gq^{-1/2}z_k
	\big)
	|\emptyset\rangle $. 
	The result cancels the exponential factor in 
	\eqref{42forweioudoisfuoiruewoiusdofiudsoewr}.
	Hence, we finally obtain
	\begin{gather}
		(\gq^{-2})^{z_k\frac{\partial}{	\partial z_k		}}
		\mathcal{G}_{\lambda_1,\cdots,\lambda_n}^{\mu_1,\cdots,\mu_m}
		\big(	u_n \big|~{w_1,\cdots,w_m, z_1,\cdots,z_n} 		\big)
		= A_k \cdot \mathcal{G}_{\lambda_1,\cdots,\lambda_n}^{\mu_1,\cdots,\mu_m}
		\big(	u_n \big|~{w_1,\cdots,w_m, z_1,\cdots,z_n} 		\big),
		\notag \\
		A_k := \prod_{j=1}^{m}
		\Big[
		R_{\mu_j\lambda_k}\Big(	\frac{\gq w_j}{z_k}				\Big)
		\Big]
		\prod_{j=1}^{k-1}
		\Big[
		R_{\lambda_j\lambda_k}\Big(
		\gq^2 \frac{z_j}{z_k}
		\Big)^{-1}
		\Big]
		\prod_{j=k+1}^{n}
		\Big[
		R_{\lambda_k\lambda_j}\Big( \frac{z_k}{z_j}
		\Big)
		\Big].\label{5.41}
	\end{gather}

	Now we consider the shift of the spectral parameter of the dual vector intertwiner;
	\begin{align}
		( \gq^{-2} )^{w_j\frac{	\partial		}{\partial w_j}}&\mathcal{G}_{\lambda_1,\cdots,\lambda_n}^{\mu_1,\cdots,\mu_m}
		\big(	u_n \big|~{w_1,\cdots,w_m, z_1,\cdots,z_n} 		\big)
		\notag \\ 
		=
		&\langle\emptyset|
		\mathbb{I}^{*}_{\mu_1}\big(
		q^{-(n-m)}_2u_n		\big| w_1
		\big)
		\cdots
		\mathbb{I}^{*}_{\mu_{j-1}}
		\big(
		q_2^{	-(n-m+j-2)		}u_n		\big| w_{j-1}
		\big)
		\notag 
		\\
		&\cdot ~
		\mathbb{I}^{*}_{\mu_j}\big(
		q_2^{-(n-m+j-1)}u_n			\big| \gq^{-2}w_j
		\big)
		\mathbb{I}^{*}_{\mu_{j+1}}\big(
		q_2^{-(n-m+j)}u_n			\big| w_{j+1}
		\big)
		\cdots
		\mathbb{I}^{*}_{\mu_m}\big(
		q_2^{-(n-1)}u_n\big| w_m
		\big)
		\notag \\ 
		&\cdot~ \mathbb{I}_{\lambda_1}\big(
		q_2^{-(n-1)}u_n\big|z_1
		\big) \cdots 
		\mathbb{I}_{\lambda_n}\big(
		u_n\big|z_n
		\big)
		|\emptyset\rangle .
	\end{align}
	We can rewrite 
	\begin{align}
		\mathbb{I}^{*}_{\mu_j}&\big(
		q_2^{-(n-m+j-1)}u_n			\big| \gq^{-2}w_j
		\big)
		=
		\exp\bigg[
		\sum_{r=1}^{\infty}
		\frac{		q^r_1( 1 - q^r_2 )				}{r ( 1 - q^r_1 )}
		(	\gq^{2r} - 1		)
		\bigg]
		\notag \\
		&\cdot~
		\cals{T}^{-}_{\mu_j}\big(
		N, q^{-(n-m+j-1)}_{2}u_n			\big| \gq^{-1/2}w_j
		\big)^{-1} \cdot
		\mathbb{I}^{*}_{\mu_j}(w_j) \cdot
		\cals{T}^{+}_{\mu_j}\big(
		N, q^{-(n-m+j)}_{2}u_n\big| \gq^{-3/2}w_j
		\big)^{-1}. 
		\label{5.43}
	\end{align}
	Now we move the operator
	$ \cals{T}^+_{\mu_j} $ to the rightmost and
	$ \cals{T}^-_{\mu_j} $ to the leftmost as above. 
	By using \eqref{seventennfsofudsiofudsoifudsoifudsofidsfusdfs} and \eqref{eightteendsfjdsofusdoifudsoifudsfosdufoisdu}, we can calculate
	$ \langle\emptyset|
	\cals{T}^{-}_{\mu_j}\big(
	N, q^{-(n-m)}_{2}u_n\big| \gq^{-1/2}w_j
	\big)^{-1} $ and  $ 
	\cals{T}^+_{\mu_j}\big(
	N,u_n\big|\gq^{-3/2}w_j
	\big)^{-1}
	|\emptyset\rangle   $. The result of this is the cancellation of the exponential factor in \eqref{5.43}. 
	The final result is 
	\begin{gather}
		( \gq^{-2} )^{w_j\frac{	\partial		}{\partial w_j}}\mathcal{G}_{\lambda_1,\cdots,\lambda_n}^{\mu_1,\cdots,\mu_m}
		\big(	u_n \big|~{w_1,\cdots,w_m, z_1,\cdots,z_n} 		\big)
		= A_j^{*} \cdot \mathcal{G}_{\lambda_1,\cdots,\lambda_n}^{\mu_1,\cdots,\mu_m}
		\big(	u_n \big|~{w_1,\cdots,w_m, z_1,\cdots,z_n} 		\big),
		\notag \\ 
		A_j^{*} :=
		\prod_{l=1}^{j-1}
		\Big[
		R_{\mu_l\mu_j}\Big(
		\frac{w_l}{w_j}
		\Big)^{-1}
		\Big]
		\prod_{l=j+1}^{m}
		\Big[
		R_{\mu_j\mu_l}\Big( \gq^{-2}\frac{w_j}{w_l}
		\Big)
		\Big]
		\prod_{k=1}^{n}
		\Big[
		R_{\mu_j\lambda_k}\Big(
		\gq^{-1}\frac{w_j}{z_k}
		\Big)^{-1}
		\Big].\label{5.44}
	\end{gather}
	We call the system of Eqs.\eqref{5.41} and \eqref{5.44} generalized KZ equation for the vector representation.
	

	\subsection{Shift operator from the universal $ R $-matrix}
	\label{section7}
	
	Before proceeding to the generalized KZ equation for the MacMahon representation, in this section we provide an alternative method to obtain
	the shift operator. This method is based on the use of the universal $ R $-matrix of DIM algebra $\cals{R}_0 $ \eqref{universalR}.
	Consider the quantity $ \bigg\{[\rho^{V}_v \tens \rho^H](\cals{R}_0)\bigg\}[v]_{n-1} \tens |\bullet\rangle $. 
	Here $ \rho^V_v $ is the vector representation with the spectral parameter $ v $ described in section \ref{subsecvectorrep}, 
	and $ \rho^H $ is the horizontal representation described in section \ref{subsechorizontal}. Up to a proportional factor, we see that 
	\begin{gather}
		\bigg\{[\rho^{V}_v \tens \rho^H](\cals{R}_0)\bigg\}[v]_{n-1} \tens |\bullet\rangle 
		\sim 
		\bigg\{
		\exp\big[ \sum_{r=1}^{\infty} r\kappa_r\rho^V(h_{-r}) \tens \rho^H(h_r)									\big]
		\bigg\}\bigg(	[v]_{n-1} \tens |\bullet \rangle					\bigg).
	\end{gather}
	Eq.\eqref{Lemma512} tells us that 
	\begin{align}
		\bigg\{[\rho^{V}_v \tens \rho^H](&\cals{R}_0)\bigg\}[v]_{n-1} \tens |\bullet\rangle 
		\notag \\ 
		&\sim 
		\bigg\{
		\exp\big[  \sum_{r=1}^{\infty}\frac{	q^{-nr}_{1}( 1 - q^{-r}_{3})(	1 - q^{-r}_{2}	)			}{v^r} \tens \frac{\rho^H(H_r)}{\kappa_r}							\big]
		\bigg\}\bigg(	[v]_{n-1} \tens |\bullet \rangle					\bigg)
		\notag \\
		&= 
		\displaystyle \bigg\{
		\exp\big[ - \sum_{r=1}^{\infty}\frac{q^{-(n-1)r}_{1}}{v^r} \frac{H_r}{1 - q^r_1}					\big]
		\bigg\} 
		\bigg(	[v]_{n-1} \tens |\bullet \rangle					\bigg).
	\end{align}
	%
	Note that here we suppress the tensor product notation and $ \rho^H $. 
	The context will tell us either $ H_r $ plays a role of generator of the DIM algebra or the operator on a representation. 
	If we compare this operator with Eq.\eqref{eightteendsfjdsofusdoifudsoifudsfosdufoisdu}, we see 
	\begin{equation}\label{T+fromR}
		[\rho^{V}_v \tens \rho^H](\cals{R}_0) \sim \mathcal{T}^{n}_{n}(N,u|\gq^{1/2}v,\gq^{-1/2}v)^{-1} = \cals{T}^+_n\big(N,u \big| v\big).
	\end{equation}

	Now consider the quantity $ \bigg\{[\rho^{V}_v \tens \rho^H](P\cals{R}_0)\bigg\} $. 
	Here the operator $ P $ switches the position between the elements in tensor product, i.e. $ P(h_{-r} \tens h_r) = h_r \tens h_{-r} $. 
	This time Eq.\eqref{Lemma512} implies that 
	\begin{align}
		\bigg\{
		\exp\big[	\sum_{r=1}^{\infty} r\kappa_r\rho^V(h_{r}) \tens &\rho^H(h_{-r})									\big]
		\bigg\}\bigg(	[v]_{n-1} \tens |\bullet \rangle					\bigg)
		\notag \\ 
		&=  \displaystyle \bigg\{
		\exp\big[
		\sum_{r=1}^{\infty}(q^n_1v)^r\frac{1}{1- q^r_1}H_{-r}
		\big]
		\bigg\} 
		\bigg(	[v]_{n-1} \tens |\bullet \rangle					\bigg).
	\end{align}
	If we compare this operator with Eq.\eqref{seventennfsofudsiofudsoifudsoifudsofidsfusdfs},
	we see that
	\begin{equation}\label{T-fromR}
		[\rho^{V}_v \tens \rho^H](P\cals{R}_0) \sim
		\mathcal{T}^{n}_{n}(N,u|\gq^{-1/2}v,\mathfrak{q}^{1/2}v)= \cals{T}^{-}_{n}\big(N , u \big| v			\big).
	\end{equation}
	
	
	\section{Generalized KZ equation for MacMahon intertwiners}
	\label{section8}

	\subsection{MacMahon shift operator}
	
	In the last section we learn that the formulas \eqref{T+fromR} and \eqref{T-fromR} provide us
	the shift operators $ \cals{T}^{\pm} $ from the universal $R$ matrix.
	Replacing the vector representation with the MacMahon representation, we thus consider the quantity 
	\begin{align}
		\Bigg\{
		\bigg[	&\rho^M_{v,K_1} \otimes \rho^H						\bigg]\big(
		\cals{R}_0
		\big)
		\Bigg\}\big(
		|\Pi, v, K_1 \rangle \otimes | \bullet \rangle 
		\big)
		\notag 
		\\
		&= \Bigg\{
		(K_0^{-} \otimes \gq^{d_1})(\gq^{d_1} \otimes K_0^{-})
		\exp\big[ \sum_{n=1}^{\infty} n \kappa_n \rho^{M}_{v,K_1}(h_{-n})	\otimes \rho^H(h_n)								\big]
		\Bigg\}
		\big(
		|\Pi,v,K_1\rangle \otimes |\bullet\rangle 
		\big),
	\end{align}
	where $ \rho^{M}_{v,K_1} $ represents the vertical MacMahon module with the spectral parameter $ v $ 
	and level $ K_1 $, and $ H $ denotes the horizontal Fock module with the spectral parameter $ u $.
	To calculate $ \bigg[	\rho^{M}_{v,K_1}(h_{-n})	\otimes \rho^H(h_n)											\bigg]\Big(
	|\Pi,v,K_1\rangle \otimes |\bullet\rangle 
	\Big) $ we need to know $ \bigg[		\rho^{M}_{v,K_1}(h_{-n})			\bigg]\Big(		|\Pi,v,K_1\rangle 						\Big) $,
	which has been calculated in \cite{MacMahon}. We systematically state the result as the following lemma\footnote
	{Recall the exchange of $q_1$ and $q_2$ in the formulas of \cite{MacMahon}. We have corrected typos in \cite{MacMahon}.};
	\begin{lem}
		For any $ n (\neq 0) \in \mathbb{Z} $ 
		\begin{align*}
			\rho^{M}_{v,K}(h_n)|\Pi,v,K\rangle 
			= \frac{v^n}{n}\bigg[
			\frac{1-K^n}{\kappa_n} + \sum_{	(i,j,k) \in \Pi				} q^{ni}_1q^{nj}_2q^{nk}_3
			\bigg]|\Pi,v,K\rangle. 
		\end{align*}
		\label{lemma3.1}
	\end{lem}
	Thus, we obtain
	\begin{align}
		\Bigg\{
		&\bigg[	\rho^M_{v,K_1} \otimes \rho^H						\bigg]\big(
		\cals{R}_0
		\big)
		\Bigg\}\big(
		|\Pi, v, K_1 \rangle \otimes | \bullet \rangle 
		\big)
		\notag 
		\\
		= \alpha\,\Bigg\{
		&\exp\bigg[
		- \sum_{n=1}^{\infty} n \kappa_n \bigg(
		\frac{v^{-n}}{n}\bigg[
		\frac{1-K^{-n}}{\kappa_{-n}} + \sum_{	(i,j,k) \in \Pi				} q^{-ni}_1q^{-nj}_2q^{-nk}_3
		\bigg]
		\otimes
		\frac{	H_n				}{\kappa_n}
		\bigg)
		\bigg]
		\Bigg\}
		\big(
		|\Pi, v, K_1 \rangle \otimes | \bullet \rangle 
		\big).
	\end{align}
	where we have used the fact that $ d_1 $ commutes with $ H_r $ and 
	$ \alpha $ is the multiplicative constant in front of the action 
	of $(K_0^{-} \otimes \gq^{d_1})(\gq^{d_1} \otimes K_0^{-})$
	on the state $ |\Pi, v, K_1 \rangle \otimes | \bullet \rangle  $. 
	Thus\footnote{The notation $ \cals{T}^+_{\Pi}(N \big| v, K) $ resembles the notation of 
		$ \cals{T} $ operator of the vector module. However, it is just a matter of notation.},
	\begin{align}
		\cals{T}^{+}_{\Pi}(  v, K			)
		\sim 
		\exp\bigg(
		\sum_{n=1}^{\infty} \bigg[
		\frac{1-K^{-n}}{\kappa_{n}} - \sum_{	(i,j,k) \in \Pi				} q^{-ni}_1q^{-nj}_2q^{-nk}_3
		\bigg] v^{-n} H_n
		\bigg),
		\label{3.6}
	\end{align}
	where we have suppressed the tensor product.
	
	Similarly,
	\begin{align}
		&\Bigg\{
		\bigg[	\rho^M_{v,K_1} \otimes \rho^H						\bigg]\big(
		\cals{P}\cals{R}_0
		\big)
		\Bigg\}\big(
		|\Pi, v, K_1 \rangle \otimes | \bullet \rangle 
		\big)
		\notag \\
		&= \Bigg\{
		(K_0^{-} \otimes \gq^{d_1})(\gq^{d_1} \otimes K_0^{-})
		\exp\bigg[	\sum_{n=1}^{\infty} \bigg( n \kappa_n	
		\frac{v^n}{n}\bigg[
		\frac{1-K^n}{\kappa_n} + \sum_{	(i,j,k) \in \Pi				} q^{ni}_1q^{nj}_2q^{nk}_3
		\bigg] 
		\notag \\ 
		&\quad\quad \otimes 
		(-1)\frac{	H_{-n}			}{\kappa_n}
		\bigg)
		\bigg]
		\Bigg\}
		\big(
		|\Pi,v,K_1\rangle \otimes |\bullet\rangle 
		\big).
	\end{align}
	Note that we have applied Lemma \ref{lemma3.1} as before. Thus, 
	\begin{align}
		&\Bigg\{
		\bigg[	\rho^M_{v,K_1} \otimes \rho^H						\bigg]\big(
		\cals{P}\cals{R}_0
		\big)
		\Bigg\}\big(
		|\Pi, v, K_1 \rangle \otimes | \bullet \rangle 
		\big)
		\notag \\ 
		&= \alpha \,\, \Bigg\{
		\exp\bigg[	- \sum_{n=1}^{\infty} \bigg( n \kappa_n	
		\frac{v^n}{n}\bigg[
		\frac{1-K^n}{\kappa_n} + \sum_{	(i,j,k) \in \Pi				} q^{ni}_1q^{nj}_2q^{nk}_3
		\bigg] 
		\otimes 
		\frac{	H_{-n}			}{\kappa_n}
		\bigg)
		\bigg]
		\Bigg\}
		\big(
		|\Pi,v,K_1\rangle \otimes |\bullet\rangle 
		\big).
	\end{align}
	Suppressing the tensor product, we are able to conclude that
	\begin{align}
		\cals{T}^{-}_{\Pi}(  v, K			)
		\sim
		\exp\bigg(	- \sum_{n=1}^{\infty} 	\bigg[
		\frac{1-K^n}{\kappa_n} + \sum_{	(i,j,k) \in \Pi				} q^{ni}_1q^{nj}_2q^{nk}_3
		\bigg] v^nH_{-n}
		\bigg).
		\label{3.10}
	\end{align}
	
	
	\subsection{Effect of the shift operator}
	
	As in the case of the vector and the Fock intertwiners we expect that we are able to write the shift of the spectral parameter 
	in term of the shift operators $ \cals{T}^+$ and $\cals{T}^- $. Schematically, 
	\begin{align}
		\Xi(pv) = \cals{T}^{-}\cdot \Xi(v) \cdot \cals{T}^{+}.
		\label{schemeofscaling}
	\end{align}
	However, the value of $ p $ in  \eqref{schemeofscaling} is not arbitrary but only certain appropriate values are allowed.
	We first determine these appropriate values of $ p $.

	Recall that the MacMahon intertwiner is given by (see \eqref{macmahonansatttz})
	\begin{align}
		\Xi_{\Lambda}(K;v) = z_{\Lambda}(K;v)\mathcal{M}^{[n]}(K)\tilde{\Phi}^{[n]}_{\Lambda}(v)\Gamma_n(K;v)
	\end{align}
	where $\Gamma_n(K;v)$ and $z_\Lambda(K;v)$ are given by \eqref{4.3} and \eqref{Maczeromode}.
	Note also that  (see \eqref{Fockcomposition})
	\begin{align}
		\tilde{\Phi}^{[n]}_\Lambda(v)
		&= \tilde{\Phi}_{\Lambda^{(1)}}(v) \circ \cdots \circ \tilde{\Phi}_{\Lambda^{(n)}}(q_3^{n-1}v),
	\end{align}
	where the tilde means the intertwiner without the zero mode factor. 
	In the following, introducing the horizontal spectral parameter $u$, we use the zero mode factor
	\begin{align}
		\textbf{e}(z) = K^{1/2} \frac{	\theta_{q_3}(z^{-1})		}{	\theta_{q_3}(K z^{-1})			}u.
	\end{align}
	so that we can write 
	\begin{align}
		z_\Lambda(K;v)
		= \prod_{k=1}^{h(\Lambda)} \prod_{(i,j)\in\Lambda^{(k)}}
		\frac{K}{\gq^{k-1}} \frac{\theta_{q_3}(q_3^{k-1}/x_{ijk})}{\theta_{q_3}(K/x_{ijk})}
		\frac{	\theta_{q_3}(1/(x_{ijk}v))		}{	\theta_{q_3}(K/(x_{ijk}v))			}u. 
	\end{align}
	
	Let us calculate $ \Xi_{\Lambda}(K;pv) $. Firstly,  we see that 
	\begin{align}
		\Gamma_n(K;pv)
		= &\exp\bigg(
		\sum_{r=1}^{\infty} \frac{H_{-r}}{\gq^r - \gq^{-r}}\frac{q^{nr}_3 - K^r}{\kappa_r}
		(\gq^{-1/2}v)^r(p^r - 1)
		\bigg)
		\notag 
		\\
		&\cdot 
		\Gamma_n(K;v)
		\exp\bigg(
		\sum_{r=1}^{\infty}\frac{H_r}{\gq^r - \gq^{-r}}\frac{	q^{-nr}_3 - K^{-r}			}{\kappa_r}
		(\gq^{1/2}v)^{-r}(p^{-r}-1)
		\bigg),
		\label{4.8}
	\end{align}
	and 
	\begin{align}
		z_\Lambda(K;pv)
		&= \prod_{k=1}^{h(\Lambda)} \prod_{(i,j)\in\Lambda^{(k)}}
		\frac{K}{\gq^{k-1}} \frac{\theta_{q_3}(q_3^{k-1}/x_{ijk})}{\theta_{q_3}(K/x_{ijk})}
		\frac{	\theta_{q_3}(1/(x_{ijk}pv))		}{	\theta_{q_3}(K/(x_{ijk}pv))			}u
		\notag 
		\\
		&= z_{\Lambda}(K;v)\prod_{k=1}^{h(\Lambda)}\prod_{(i,j)\in\Lambda^{(k)}}
		\frac{	\theta_{q_3}(1/(x_{ijk}pv))		}{	\theta_{q_3}(K/(x_{ijk}pv))			}
		\frac{	\theta_{q_3}(K / (x_{ijk}v))			}{\theta_{q_3}(1/(x_{ijk}v)	)}. 
		\label{6.13}
	\end{align}
	To calculate $ \tilde{\Phi}^{[n]}_{\Lambda}(pv) $ we need to know an explicit expression of $ \tilde{\Phi}^{[n]}_{\Lambda}(v) $ first. 
	It is convenient to introduce a notation
	\begin{equation}\label{chi(k)}
		\chi_{\pm r}^{(k)} (q_1, q_2) :=
		\sum_{	(i,j) \in \Lambda^{(k)}			}x^{\pm r}_{ij} - \frac{1}{(	1- q^{\pm r}_1 )( 1-	q^{\pm r}_2 )},
	\end{equation}
	Using this notation, we state its expression systematically as the following lemma;
	\begin{lem}
		\begin{align*}
			\tilde{\Phi}^{[n]}_{\Lambda}(v) 
			= c(\Lambda) \,
			\exp\big[
			\sum_{r=1}^{\infty}\frac{H_{-r}}{\gq^r - \gq^{-r}}
			\sum_{k=1}^{n}(	\gq^{-1/2}q^{k-1}_3v				)^r \chi_r^{(k)} (q_1, q_2)
			\big]
			\notag \\ 
			\exp\big[
			-\sum_{r=1}^{\infty}\frac{H_r}{\gq^r - \gq^{-r}}\sum_{k=1}^{n}(	\gq^{1/2}q^{k-1}_3v		)^{-r}
			\chi_{-r}^{(k)} (q_1, q_2)
			\big],
		\end{align*}
		where $ c(\Lambda) $ is a constant which may depend on $ \Lambda $. 
		\label{lemma4.1}
	\end{lem}
	\begin{proof}
		By definition, 
		\begin{align}
			\tilde{\Phi}^{[n]}_\Lambda(v)
			&= \tilde{\Phi}_{\Lambda^{(1)}}(v) \circ \cdots \circ \tilde{\Phi}_{\Lambda^{(n)}}(q_3^{n-1}v) 
			\notag \\
			&= \cals{G}^{-1}_{\Lambda^{(1)}}\cals{G}^{-1}_{\Lambda^{(2)}}\cdots\cals{G}^{-1}_{\Lambda^{(n)}}\big(
			\text{exponential factor}
			\big)
		\end{align}
		where (exponential factor) in the above equation is equal to 
		\begin{align}
			\prod_{k=1}^{n}
			\Bigg[
			\exp\big[
			\sum_{r=1}^{\infty}\frac{H_{-r}}{\gq^r - \gq^{-r}}
			(	\gq^{-1/2}q^{k-1}_3v				)^r
			\big(	\sum_{	(i,j) \in \Lambda^{(k)}			}	x^r_{ij} 	
			- \frac{1}{	(1-q^r_1)(1-q^r_2)			}
			\big)
			\big]
			\notag \\ 
			\exp\big[
			-\sum_{r=1}^{\infty}\frac{H_r}{\gq^r - \gq^{-r}}(	\gq^{1/2}q^{k-1}_3v		)^{-r}
			\big(
			\sum_{	(i,j) \in \Lambda^{(k)}			}x_{ij}^{-r} - 
			\frac{1}{	(1-q^{-r}_1)(1-q^{-r}_2)				}
			\big)
			\big]
			\Bigg].
			\label{4.11}
		\end{align}
		Applying the formula \eqref{CBH}
		to Eq.\eqref{4.11}, we obtain the desired result. 
	\end{proof}
	
	\begin{rem}
		Recall that the aim of this subsection is to calculate the appropriate value of $ p $ 
		by comparing with the exponential factor appearing in \eqref{3.6} and \eqref{3.10}. 
		For this purpose, only the exponential factor of $ H_{\pm r}$ is instrumental, and we do not pay attention to the form of the proportional factor $ c $. 
		For the sake of convenience we introduce the notation $ \sim $ to mean equal up to a proportional factor. 
	\end{rem}
	
	By Lemma \ref{lemma4.1}
	\begin{align}
		\tilde{\Phi}^{[n]}_{\Lambda}&(pv)
		= \exp\big[
		\sum_{r=1}^{\infty}\frac{(p^r-1)H_{-r}}{\gq^r - \gq^{-r}}
		\sum_{k=1}^{n}(	\gq^{-1/2}q^{k-1}_3v				)^r 
		\chi_{r}^{(k)} (q_1, q_2) \big]
		\notag \\ 
		&\cdot \tilde{\Phi}^{[n]}_{\Lambda}(v)
		\exp\big[
		-\sum_{r=1}^{\infty}\frac{(	p^{-r} - 1) H_r}{\gq^r - \gq^{-r}}\sum_{k=1}^{n}(	\gq^{1/2}q^{k-1}_3v		)^{-r}
		\chi_{-r}^{(k)} (q_1, q_2) \big].
		\label{4.12}
	\end{align}
	Merging \eqref{4.8} and \eqref{4.12}, we see that 
	\begin{equation}\label{4.13}
		\Xi_{\Lambda}(K;pv)
		\sim \widetilde{\cals{T}}^{-}_\Lambda(p) \cdot  \, \Xi_{\Lambda}(K;v)\cdot  \,  \widetilde{\cals{T}}^{+}_\Lambda(p),
	\end{equation}
	where $ \widetilde{\cals{T}}^{-}_\Lambda(p)$ and $\widetilde{\cals{T}}^{+}_\Lambda(p)$ are
	\begin{equation}
		\exp\bigg(
		\sum_{r=1}^{\infty} \frac{(p^r-1) H_{-r}}{(\gq^r - \gq^{-r})}
		(\gq^{-1/2}v)^r \Big\{
		\frac{q^{nr}_3 - K^r}{\kappa_r}
		+ \sum_{k=1}^n  q^{r(k-1)}_3
		\chi_{r}^{(k)} (q_1, q_2)
		\Big\}
		\bigg), \nonumber
	\end{equation}
	and
	\begin{equation}
		\exp\bigg(
		\sum_{r=1}^{\infty}\frac{(	p^{-r}-1) H_r}{(\gq^r - \gq^{-r})}
		(\gq^{1/2}v)^{-r} \Big\{
		\frac{	q^{-nr}_3 - K^{-r}			}{\kappa_r}
		- \sum_{k=1}^{n} q^{-r(k-1)}_3 \chi_{- r}^{(k)} (q_1, q_2)
		\Big\}
		\bigg), \nonumber
	\end{equation}
	respectively.
	Again note that when we switch the $ \exp(H_{-r}) $ factor of $ \Gamma_n(K;pv) $ 
	and the $ \exp(-H_r) $ factor of $ \tilde{\Phi}^{[n]}_{\Lambda}(pv) $ 
	the additional proportional factor appears from the formula \eqref{CBH}. 
	However, under the sign $ \sim $ we are able to  pay no attention to it.

	Now compare the $ \exp(H_{-r}) $ factor of \eqref{4.13} with those of \eqref{3.10}. We would like to find a value of $ p $ 
	which renders them to be equal up to a shift of spectral parameter $ v $. That is, we require that 
	\begin{equation}
		\widetilde{\cals{T}}^{-}_\Lambda(p) =
		\exp\bigg(	- \sum_{n=1}^{\infty} \bigg[
		\frac{1-K^n}{\kappa_n} + \sum_{	(i,j,k) \in \Lambda				} q^{ni}_1q^{nj}_2q^{nk}_3
		\bigg] 
		\frac{v^n}{\lambda^n}H_{-n}
		\bigg). 
	\end{equation}
	for some constant $ \lambda $. 
	We find that the above equality holds if and only if 
	\begin{align}
		\frac{p^r - 1}{\gq^r - \gq^{-r}}&\bigg\{
		(\gq^{-1/2}v)^r\big(
		\frac{q^{nr}_3 - K^r}{\kappa_r}
		\big)
		+ \sum_{k=1}^n (	\gq^{-1/2}q^{k-1}_3v		)^r \chi_{r}^{(k)} (q_1, q_2)
		\bigg\}
		\notag \\ 
		&= - \frac{v^r}{\lambda^r}\bigg[
		\frac{1-K^n}{\kappa_n} + \sum_{	(i,j,k) \in \Lambda				} q^{ni}_1q^{nj}_2q^{nk}_3
		\bigg].
	\end{align}
	Recall that we have exchanged $q_1$ and $q_2$ in the Fock representation; $ x_{ij} := q^{i-1}_{1}q^{j-1}_{2} $, hence 
	$ \displaystyle{\sum_{k=1}^n} q^{(k-1)r}_3 \displaystyle{\sum_{(i,j)\in\Lambda^{(k)}}} x^r_{ij} = \sum_{(i,j,k) \in \Lambda}x^r_{ijk}$. 
	We see that for each $ r \in \mathbb{Z}^{\geq 1} $, we should have
	\begin{align}
		p^r = \frac{\lambda^r - \gq^{3r/2} + \gq^{-r/2}}{\lambda^r}.
	\end{align}
	To make the value of $ p $ independent of $ r $ we have to choose $ \lambda = \gq^{3/2} $, and hence we get that $ p = \gq^{-2} $. 
	Thus we see that the shift parameter is fixed. Thus, we have confirmed that 
	\begin{equation}
		\widetilde{\cals{T}}^{-}_\Lambda(\gq^{-2} ) = \cals{T}^{-}_{\Lambda}( \gq^{-3/2} v, K		).
	\end{equation}
	
	
	We expect that the value $ p = \gq^{-2} $ also works for the $ \exp(H_r) $ as well. More precisely,
	when $ p = \gq^{-2} $ we should check that there exists $ \lambda \in \mathbb{C} $ such that 
	\begin{equation}
		\widetilde{\cals{T}}^{+}_\Lambda(p) 
		= \cals{T}^{+}_{\Lambda}( \frac{v}{\lambda}, K			),
		\label{4.18}
	\end{equation}
	which means
	\begin{align}
		\frac{(p^{-r}-1) H_r}{\gq^r - \gq^{-r}}&(\gq^{1/2}v)^{-r}
		\bigg\{
		\frac{	q^{-nr}_3 - K^{-r}			}{\kappa_r}
		-
		\sum_{k=1}^{n}(	q^{k-1}_3				)^{-r} \chi_{-r}^{(k)} (q_1, q_2)
		\bigg\}
		\bigg)
		\notag \\ 
		&= \kappa_r \bigg(
		\frac{v^{-r}}{\lambda^{-r}}\bigg[
		\frac{1-K^{-r}}{\kappa_{r}} - \sum_{	(i,j,k) \in \Lambda			} q^{-ri}_1q^{-rj}_2q^{-rk}_3
		\bigg]
		\frac{	H_n				}{\kappa_n}
		\bigg).
	\end{align} 
	By using the summation formula for geometric series, the above equality becomes 
	\begin{gather}
		( p^{-r} - 1		) = \lambda^r(	\gq^{3r/2} - \gq^{-r/2}	). 
	\end{gather}
	If $ p = \gq^{-2} $, we obtain $\lambda = \gq^{1/2}$.
	Hence, 
	\begin{equation}
		\widetilde{\cals{T}}^{+}_\Lambda(\gq^{-2}) 
		= \cals{T}^{+}_{\Lambda}(  \gq^{-1/2}v, K	)
		\label{4.22}
	\end{equation}
	is satisfied. In summary we get the relation 
	\begin{align}
		\Xi_{\Lambda}(	K ; \gq^{-2}v		) = K^{|\Lambda|}~
		\cals{T}^{-}_{\Lambda}(  \gq^{-3/2} v, K	) \cdot \Xi_{\Lambda}(K ; v) \cdot
		\cals{T}^{+}_{\Lambda}(  \gq^{-1/2}v, K			),
		\label{4.35}
	\end{align}
	where we fix the proportional factor $K^{|\Lambda|}$ in appendix. 
	
	\subsection{Case of the dual intertwiner}
	
	Recall that the dual MacMahon intertwiner is;
	\begin{align}
		\Xi_{\Lambda}^{*}(K;v) = z_{\Lambda}^*(K;v)\cals{M}^{[n]*}(K)\tilde{\Phi}^{[n]*}_{\Lambda}(v)\Gamma^{*}_{n}(K;v),
	\end{align}
	where $\Gamma^{*}_{n}(K;v)$ and $z^*_{\Lambda}(K;v)$ are given by \eqref{4.36} and \eqref{dualMaczeromode}.
	We also have (see \eqref{dualFockcomposition});
	\begin{gather}
		\tilde{\Phi}^{[n]*}_{\Lambda}(v) =
		\tilde{\Phi}^{*}_{\Lambda^{(1)}}(v) \comp \cdots \comp \tilde{\Phi}^{*}_{\Lambda^{(n)}}(q^{n-1}_{3}v). 
	\end{gather}
	We will use the zero mode
	\begin{align}
		\textbf{f}(z) = \frac{	\theta_{q_3}(\gq K z^{-1})			}{\theta_{q_3}(\gq z^{-1})}u^{-1},
	\end{align}
	so that we can write 
	\begin{align}
		z^*_{\Lambda}(K;v) = \prod_{k=1}^{h(\Lambda)}\prod_{(i,j) \in \Lambda^{(k)}}
		\Bigg(
		\frac{K^{1/2}}{\gq^k}\frac{	 \theta_{q_3}(q^k_3x^{-1}_{ijk})			}{\theta_{q_3}(Kx^{-1}_{ijk})}
		\Bigg)^{-1}\frac{	\theta_{q_3}(\gq K \cdot(x_{ijk}v)^{-1})			}{\theta_{q_3}(\gq \cdot(x_{ijk}v)^{-1})}u^{-1}.
	\end{align}
	
	In order to calculate $ \Xi^*_{\Lambda}(K;pv) $ we need to calculate
	$ z_{\Lambda}^*(K;pv), \tilde{\Phi}^{[n]*}_{\Lambda}(pv) $ and $\Gamma^{*}_{n}(K;pv) $.
	From the above equations we see that 
	\begin{align}
		\Gamma^*_n(K;pv) =
		&\exp\bigg(
		-\sum_{r=1}^{\infty}\frac{H_{-r}}{\gq^r - \gq^{-r}}\frac{q^{nr}_3 - K^r}{\kappa_r}
		(\gq^{1/2}v)^r(p^r - 1)
		\bigg)
		\notag \\ 
		&\cdot \Gamma^*_n(K;v)
		\exp\bigg(
		-\sum_{r=1}^{\infty}\frac{H_r}{\gq^r - \gq^{-r}}\frac{	q^{-nr}_3 - K^{-r}		}{\kappa_r}
		(\gq^{-1/2}v)^{-r}(p^{-r}-1)
		\bigg),
		\label{4.30}
	\end{align}
	and 
	\begin{gather}
		z^*_{\Lambda}(K;pv) = z^{*}_{\Lambda}(K;v)
		\prod_{k=1}^{h(\Lambda)}\prod_{(i,j) \in \Lambda^{(k)}}
		\bigg[
		\frac{	\theta_{q_3}( \gq\cdot(x_{ijk}v)^{-1}			)			}{	 \theta_{q_3}(\gq K \cdot(x_{ijk}v)^{-1} )			}
		\frac{	 \theta_{q_3}(	\gq K \cdot(x_{ijk}pv)^{-1}				)						}{
			\theta_{q_3}(\gq\cdot(x_{ijk}pv)^{-1}) }
		\bigg]. 
		\label{6.33}
	\end{gather}
	Using the notation \eqref{chi(k)}, we can state an explicit expression of 
	$\tilde{\Phi}^{[n]*}_{\Lambda}(v) $ as the following lemma;
	\begin{lem}
		\begin{align*}
			\tilde{\Phi}^{[n]*}_{\Lambda}(v)
			= \tilde{c}(\Lambda)
			&\exp\bigg[
			-\sum_{m=1}^{n}\sum_{r=1}^{\infty}\frac{H_{-r}}{\gq^r - \gq^{-r}}(\gq^{1/2}q^{m-1}_{3}v)^r  \chi_{r}^{(m)} (q_1, q_2)
			\bigg]
			\notag \\ 
			&\cdot \exp\bigg[
			\sum_{m=1}^{n}\sum_{r=1}^{\infty}\frac{H_r}{\gq^r - \gq^{-r}}(	\gq^{-1/2}q^{m-1}_{3}v	)^{-r} \chi_{-r}^{(m)} (q_1, q_2)
			\bigg]. 
		\end{align*}
		where $ \tilde{c}(\Lambda) $ is a constant which may depend on $ \Lambda $. 
		\label{lemma4.3}
	\end{lem}
	\begin{proof}
		By definition
		\begin{align}
			\tilde{\Phi}^{[n]*}_{\Lambda}(v) &=
			\tilde{\Phi}^{*}_{\Lambda^{(1)}}(v) \circ \cdots \circ \tilde{\Phi}^{*}_{\Lambda^{(n)}}(	q^{n-1}_{3}v		) 
			\notag \\ 
			&= \tilde{\cals{G}}^{-1}_{\Lambda^{(1)}}\cdots \tilde{\cals{G}}^{-1}_{\Lambda^{(n)}}
			\big(
			\text{exponential factor}
			\big)
		\end{align}
		where $ \big(\text{exponential factor} \big) $ is 
		\begin{gather}
			\prod_{m=1}^{n}\Bigg[
			\exp\big[
			\sum_{r=1}^{\infty}\frac{H_{-r}}{\gq^r - \gq^{-r}}(	\gq^{1/2}q^{m-1}_{3}v	)^r\bigg(
			\frac{1}{( 1 - q^r_1)(1 - q^r_2)} - \sum_{(i,j) \in \Lambda^{(m)}}x^r_{ij}
			\bigg)
			\big]
			\notag \\ 
			\cdot \exp\big[
			\sum_{r=1}^{\infty}\frac{H_r}{\gq^r - \gq^{-r}}(	 \gq^{-1/2}q^{m-1}_{3}v			)^{-r}
			\bigg(
			\sum_{	(i,j) \in \Lambda^{(m)}			}x^{-r}_{ij}
			- \frac{1}{( 1- q^{-r}_1)( 1 - q^{-r}_2 )}
			\bigg)
			\big]
			\Bigg]. 
			\label{4.33}
		\end{gather}
		Applying \eqref{CBH} to Eq.\eqref{4.33}, we obtain the desired result. 
	\end{proof}
	
	Now Lemma \ref{lemma4.3} implies that 
	\begin{align}
		\tilde{\Phi}^{[n]*}_{\Lambda}&(pv) 
		= \exp\bigg[
		-\sum_{m=1}^{n}\sum_{r=1}^{\infty}
		\frac{(p^r - 1)H_{-r}}{\gq^r - \gq^{-r}}(	\gq^{1/2}q^{m-1}_{3}v		)^r \chi_{r}^{(m)} (q_1, q_2)
		\bigg]
		\notag \\ 
		&\cdot \tilde{\Phi}^{[n]*}_{\Lambda}(v)
		\cdot
		\exp\bigg[
		\sum_{m=1}^{n}\sum_{r=1}^{\infty}\frac{(p^{-r} - 1)H_r}{\gq^r - \gq^{-r}}(	\gq^{-1/2}q^{m-1}_{3}v		)^{-r} 
		\chi_{-r}^{(m)} (q_1, q_2)
		\bigg].
		\label{4.34}
	\end{align}
	
	Now we merge the results \eqref{4.30} and \eqref{4.34} together. Then, we see that
	\begin{align}
		&\Xi^{*}_{\Lambda}(K;pv)
		\sim
		\notag \\ 
		&\exp\bigg[
		-\sum_{r=1}^{\infty}\frac{(p^r - 1) H_{-r}}{\gq^r - \gq^{-r}}(\gq^{1/2}v)^r
		\big\{
		\sum_{m=1}^nq^{(m-1)r}_{3} \chi_{r}^{(m)} (q_1, q_2)
		+ \frac{q^{nr}_3 - K^r}{\kappa_r}
		\big\}
		\bigg]
		\cdot \Xi^{*}_{\Lambda}(K;v)
		\notag \\
		&\cdot \exp\bigg[
		\sum_{r=1}^{\infty}\frac{(p^{-r} - 1) H_r}{\gq^r - \gq^{-r}}(\gq^{-1/2}v)^{-r}
		\big\{
		\sum_{m=1}^{n}q^{-(m-1)r}_{3} \chi_{-r}^{(m)} (q_1, q_2)
		- \frac{	q^{-nr}_{3} - K^{-r}		}{\kappa_r}
		\big\}
		\bigg].
		\label{4.46}
	\end{align}
	It is straightforward to see that if we take $p=\gq^{-2}$,
	\begin{align}
		\Xi^{*}_{\Lambda}(K;\gq^{-2}v)
		= K^{-|\Lambda|}~
		\cals{T}^{-}_{\Lambda}\big(  \gq^{-1/2}v, K	\big)^{-1} \cdot
		\Xi^{*}_{\Lambda}(K;v) \cdot
		\cals{T}^{+}_{\Lambda}\big(	 \gq^{-3/2}v, K			\big)^{-1},
		\label{4.62}
	\end{align}
	where we explicitly determine the proportional factor $K^{-|\Lambda|}$ in appendix.

	\subsection{MacMahon KZ equation}
	\label{section5}
	
	All the commutation relations between the MacMahon intertwiners and the shift operators 
	are obtained by using equation \eqref{CBH}. It turns out that we can write the commutation relations
	compactly by introducing 
	\begin{equation}
		\tilde{R}^{K_1,K_2}_{\Pi,\Lambda}\big(
		\frac{u}{v}
		\big)
		:= \exp\bigg[
		- \sum_{r=1}^{\infty}\frac{\kappa_r}{r}\Big(	\frac{u}{v}			\Big)^{-r}
		\Big(	\sum_{	(i,j,k) \in \Lambda				}x^r_{ijk} + \frac{1 - K^r_2}{\kappa_r}	 \Big)
		\Big(	\sum_{	(i,j,k) \in \Pi	}x^{-r}_{ijk}  - \frac{1-K^{-r}_{1}}{\kappa_r} 	\Big)
		\bigg],
		\label{5.13}
	\end{equation}
	The commutation relations are
	\begin{align}
		\cals{T}^+_{\Lambda}(v,K)\Xi_{\Lambda^{\prime}}(K^{\prime},v^{\prime} )  
		&= \tilde{R}^{K,K^{\prime}}_{\Lambda,\Lambda^{\prime}}\big(
		\frac{v}{\gq^{-1/2}v^{\prime}}
		\big) \cdot \Xi_{\Lambda^{\prime}}(K^{\prime},v^{\prime})\cals{T}^+_{\Lambda}(v,K), \\
		\cals{T}^+_{\Lambda}(v,K)\Xi^*_{\Lambda^{\prime}}(K^{\prime},v^{\prime})
		&= \tilde{R}^{K,K^{\prime}}_{\Lambda,\Lambda^{\prime}}\big(
		\frac{v}{\gq^{1/2}v^{\prime}}
		\big)^{-1} \cdot \Xi^*_{\Lambda^{\prime}}(K^{\prime},v^{\prime})\cals{T}^+_{\Lambda}(v,K), \\
		\cals{T}^-_{\Lambda}(v,K)\Xi_{\Lambda^{\prime}}(K^{\prime},v^{\prime})
		&=  \tilde{R}^{K^{\prime},K}_{\Lambda^{\prime},\Lambda}\big(
		\frac{v^{\prime}}{\gq^{-1/2}v}
		\big) \cdot \Xi_{\Lambda^{\prime}}(K^{\prime},v^{\prime})\cals{T}^-_{\Lambda}(v,K), \\
		\cals{T}^-_{\Lambda}(v,K)\Xi^*_{\Lambda^{\prime}}(K^{\prime},v^{\prime})
		&= \tilde{R}^{K^{\prime},K}_{\Lambda^{\prime},\Lambda}\big(
		\frac{v^{\prime}}{\gq^{1/2}v}
		\big)^{-1} \cdot \Xi^*_{\Lambda^{\prime}}(K^{\prime},v^{\prime})\cals{T}^-_{\Lambda}(v,K).
	\end{align}
	They should be compared with the corresponding relations \eqref{prrrrro5.20}--\eqref{prrrrrro5.23} for the intertwiners 
	of the vector representation. 

	We can determine the relation between $ \tilde{R}^{K_1,K_2}_{\Pi,\Lambda}\big(
	\frac{u}{v} \big) $ 
	and the MacMahon R-matrix $ R^{K_1,K_2}_{\Pi,\Lambda}\big(	 \frac{u}{v}		\big) $ 
	which appears in the paper \cite{MacMahon}. From \eqref{5.13} we get that 
	\begin{equation}
		\tilde{R}^{K_1,K_2}_{\Pi,\Lambda}\big(
		\frac{u}{v}
		\big)
		= K^{\frac{- |\Lambda|			}{2}}_{1}K^{	- \frac{|\Pi|}{2}		}_{2}R^{K_1,K_2}_{\Pi,\Lambda}\big(	 \frac{u}{v}		\big)
		\exp\bigg[
		\sum_{r=1}^{\infty}\frac{1}{r}\Big(	\frac{u}{v}			\Big)^{-r}
		\frac{	(1-K^r_2)(1-K^{-r}_{1})		}{\kappa_r}
		\bigg].
		\label{8.86}
	\end{equation}
	Note that the factor 
	\begin{equation}
		\exp\bigg[
		\sum_{r=1}^{\infty}\frac{1}{r}\Big(	\frac{u}{v}			\Big)^{-r}
		\frac{	(1-K^r_2)(1-K^{-r}_{1})		}{\kappa_r}
		\bigg]
	\end{equation} 
	is the vacuum contribution to $ \tilde{R}^{K_1,K_2}_{\Pi,\Lambda}$ or the normalization factor to make $ R^{K_1,K_2}_{\varnothing,\varnothing}=1$.

	Now we are ready to derive the generalized KZ equation for MacMahon intertwiner. We define the correlation function by
	\begin{gather}
		G^{\Omega^1\cdots\Omega^m}_{\Lambda^1\cdots\Lambda^n}
		\Big(^{	K^{\prime}_{1}, \dots, K^{\prime}_{m}, K_1 \dots, K_n		}_{w_1,\dots,w_m,v_1,\dots,v_n}				\Big)
		:= \langle \emptyset | \Xi^*_{\Omega^1}(K^{\prime}_{1};w_1)
		\cdots
		\Xi^*_{\Omega^m}(K^{\prime}_{m};w_m)
		\Xi_{\Lambda^1}(K_1;v_1)\cdots\Xi_{\Lambda^n}(K_n;v_n)|\emptyset\rangle.
	\end{gather}
	Let us first consider the shift of the spectral  parameter $v_i$;
	\begin{gather}
		(\gq^{-2})^{v_i\frac{\partial}{\partial v_i}}G^{\Omega^1\cdots\Omega^m}_{\Lambda^1\cdots\Lambda^n}
		\Big(^{	K^{\prime}_{1}, \dots, K^{\prime}_{m}, K_1 \dots, K_n		}_{w_1,\dots,w_m,v_1,\dots,v_n}				\Big)
		\notag \\
		= \langle \emptyset | \Xi^*_{\Omega^1}(K^{\prime}_{1};w_1)\cdots
		\Xi^*_{\Omega^m}(K^{\prime}_{m};w_m)
		\Xi_{\Lambda^1}(K_1;v_1)\cdots\Xi_{\Lambda^{i-1}}(K_{i-1};v_{i-1})\Xi_{\Lambda^i}(K_i;\gq^{-2}v_i)
		\notag \\ 
		\Xi_{\Lambda^{i+1}}(K_{i+1};v_{i+1})
		\cdots\Xi_{\Lambda^n}(K_n;v_n)|\emptyset\rangle.
	\end{gather}
	From \eqref{4.35}, the above equation becomes
	\begin{gather}
		{K_i}^{|\Lambda_i|}
		\langle \emptyset | \Xi^*_{\Omega^1}(K^{\prime}_{1};w_1)\cdots\Xi^*_{\Omega^m}(K^{\prime}_{m};w_m)
		\Xi_{\Lambda^1}(K_1;v_1)\cdots\Xi_{\Lambda^{i-1}}(K_{i-1};v_{i-1})
		\notag \\
		\cdot \cals{T}^-_{\Lambda^i	}\big(	\gq^{-3/2}v_i, K_i	\big) \cdot \Xi_{\Lambda^i}(K_i;v_i) \cdot
		\cals{T}^+_{\Lambda^i		}\big(	\gq^{-1/2}v_i,K_i	\big)\Xi_{\Lambda^{i+1}}(K_{i+1};v_{i+1})\cdots
		\Xi_{\Lambda^{n}}(K_n;v_{n})|\emptyset\rangle~.
	\end{gather}
	Next we move the operator $ \cals{T}^+_{\Lambda^i}\big(	v_i,K_i			\big)^{-1} $ 
	to the right by using the result in the section \ref{section5}, which implies
	\begin{gather}
		\cals{T}^+_{\Lambda^i}\big(	\gq^{-1/2}v_i,K_i	\big) \cdot \Xi_{\Lambda^{i+1}}(K_{i+1};v_{i+1})\cdots
		\Xi_{\Lambda^{n}}(K_n;v_{n})|\emptyset\rangle
		\notag \\ 
		= \bigg(
		\prod_{l = i+1}^{n}\tilde{R}^{K_i,K_l}_{\Lambda^i,\Lambda^l}\big(
		\frac{v_i}{v_l}
		\big)
		\bigg)
		\Xi_{\Lambda^{i+1}}(K_{i+1};v_{i+1})\cdots
		\Xi_{\Lambda^{n}}(K_n;v_{n})|\emptyset\rangle.
	\end{gather}
	Thus, at this step we see that
	\begin{gather}
		(\gq^{-2})^{v_i\frac{\partial}{\partial v_i}}G^{\Omega^1\cdots\Omega^m}_{\Lambda^1\cdots\Lambda^n}
		\Big(^{	K^{\prime}_{1}, \dots, K^{\prime}_{m}, K_1 \dots, K_n		}_{w_1,\dots,w_m,v_1,\dots,v_n}				\Big)
		= {K_i}^{|\Lambda_i|}
		\bigg(
		\prod_{l = i+1}^{n}\tilde{R}^{K_i,K_l}_{\Lambda^i,\Lambda^l}\big(
		\frac{v_i}{v_l}
		\big)
		\bigg)
		\notag \\ 
		\langle \emptyset | \Xi^*_{\Omega^1}(K^{\prime}_{1};w_1)\cdots
		\Xi^*_{\Omega^m}(K^{\prime}_{m};w_m)
		\Xi_{\Lambda^1}(K_1;v_1)\cdots\Xi_{\Lambda^{i-1}}(K_{i-1};v_{i-1})
		\notag \\
		\cdot \cals{T}^-_{\Lambda^i}\big(	\gq^{-3/2}v_i, K_i			\big)\Xi_{\Lambda^i}(K_i;v_i)
		\Xi_{\Lambda^{i+1}}(K_{i+1};v_{i+1})\cdots
		\Xi_{\Lambda^{n}}(K_n;v_{n})|\emptyset\rangle.
	\end{gather}
	We then move the operator $ \cals{T}^-_{\Lambda^i}\big(	\gq^{-3/2}v_i, K_i			\big)$ to the left.
	In the same manner as above, we obtain
	\begin{gather}
		(\gq^{-2})^{v_i\frac{\partial}{\partial v_i}}G^{\Omega^1\cdots\Omega^m}_{\Lambda^1\cdots\Lambda^n}
		\Big(^{	K^{\prime}_{1}, \dots, K^{\prime}_{m}, K_1 \dots, K_n		}_{w_1,\dots,w_m,v_1,\dots,v_n}				\Big)
		= \cals{A}_i \cdot G^{\Omega^1\cdots\Omega^m}_{\Lambda^1\cdots\Lambda^n}
		\Big(^{	K^{\prime}_{1}, \dots, K^{\prime}_{m}, K_1 \dots, K_n		}_{w_1,\dots,w_m,v_1,\dots,v_n}				\Big),
		\notag \\ 
		\cals{A}_i := K_i^{|\Lambda^i|}
		\bigg(
		\prod_{l = i+1}^{n}\tilde{R}^{K_i,K_l}_{\Lambda^i,\Lambda^l}\big(
		\frac{v_i}{v_l}
		\big)
		\bigg)
		\bigg(	\prod_{k=1}^{i-1} \tilde{R}^{K_k,K_i}_{\Lambda^k,\Lambda^i}
		\big(	\gq^{2}\frac{v_k}{v_i}				\big)^{-1}
		\bigg)
		\bigg(	\prod_{s=1}^{m}	
		\tilde{R}^{K^{\prime}_s,K_i}_{\Omega^s,\Lambda^i}
		\big(		\gq\frac{w_s}{v_i}					\big)
		\bigg).
		\label{MacKZ1}
	\end{gather}

	Next let us consider the shift of the spectral parameter $w_i$;
	\begin{gather}
		(\gq^{-2})^{w_i\frac{\partial}{\partial w_i}}G^{\Omega^1\cdots\Omega^m}_{\Lambda^1\cdots\Lambda^n}
		\Big(^{	K^{\prime}_{1}, \dots, K^{\prime}_{m}, K_1 \dots, K_n		}_{w_1,\dots,w_m,v_1,\dots,v_n}				\Big)
		\notag \\
		= \langle \emptyset |
		\Xi^*_{\Omega^1}(K^{\prime}_{1};w_1)\cdots\Xi^*_{\Omega^{i-1}}(K^{\prime}_{i-1};w_{i-1})\Xi^*_{\Omega^i}(K^{\prime}_{i};\gq^{-2}w_i)
		\Xi^*_{\Omega^{i+1}}(K^{\prime}_{i+1};w_{i+1})\cdots
		\Xi^*_{\Omega^m}(K^{\prime}_{m};w_m)
		\notag \\ 
		\cdot
		\Xi_{\Lambda^1}(K_1;v_1)\cdots\Xi_{\Lambda^n}(K_n;v_n)|\emptyset\rangle.
		\label{6.7}
	\end{gather}
	From \eqref{4.62} we get that \eqref{6.7} is equal to  
	\begin{gather}
		(K^{\prime}_{i})^{-|\Omega^i|}
		\langle \emptyset | \Xi^*_{\Omega^1}(K^{\prime}_{1};w_1)\cdots\Xi^*_{\Omega^{i-1}}(K^{\prime}_{i-1};w_{i-1})
		\cals{T}^-_{\Omega^i}(\gq^{-1/2}w_i,K^{\prime}_{i})^{-1} \cdot \Xi^*_{\Omega^i}(K^{\prime}_{i};w_i)\cals{T}^+_{\Omega^i}(\gq^{-3/2}w_i,K^{\prime}_i)^{-1}
		\notag \\
		\cdot\Xi^*_{\Omega^{i+1}}(K^{\prime}_{i+1};w_{i+1})\cdots
		\Xi^*_{\Omega^m}(K^{\prime}_{m};w_m)
		\Xi_{\Lambda^1}(K_1;v_1)\cdots\Xi_{\Lambda^n}(K_n;v_n)|\emptyset\rangle
		\label{6.8}
	\end{gather}
	As above we can move $ \cals{T}^+_{\Omega^i}$ to the right
	and $ \cals{T}^-_{\Omega^i}$ to the left by using the result in section \ref{section5}.
	We see that \eqref{6.8} is equal to 
	\begin{gather}
		(K^{\prime}_{i})^{-|\Omega^i|}
		\Big(	\prod_{k=i+1}^{m}\tilde{R}^{K^{\prime}_{i},K^{\prime}_{k}}_{\Omega^i,\Omega^k}\big( \gq^{-2}\frac{w_i}{w_k}  \big)	\Big)
		\Big(	\prod_{l=1}^{n} 
		\tilde{R}^{K^{\prime}_{i},K_l}_{\Omega^i,\Lambda^l}\big( \gq^{-1}\frac{w_i}{v_l}			\big)^{-1}		\Big)
		\Big(	\prod_{s=1}^{i-1}  \tilde{R}^{K^{\prime}_{s},K^{\prime}_{i}}_{\Omega^s,\Omega^i}
		\big(		\frac{w_s}{w_i}				\big)^{-1}
		\Big)
		\notag \\ 
		\langle \emptyset | \Xi^*_{\Omega^1}(K^{\prime}_{1};w_1)\cdots\Xi^*_{\Omega^{i-1}}(K^{\prime}_{i-1};w_{i-1})
		\Xi^*_{\Omega^i}(K^{\prime}_{i};w_i)\cdots\Xi^*_{\Omega^m}(K^{\prime}_{m};w_m)
		\Xi_{\Lambda^1}(K_1;v_1)\cdots\Xi_{\Lambda^n}(K_n;v_n)|\emptyset\rangle.
	\end{gather}
	Hence, we obtain
	\begin{gather}
		(\gq^{-2})^{w_i\frac{\partial}{\partial w_i}}G^{\Omega^1\cdots\Omega^m}_{\Lambda^1\cdots\Lambda^n}
		\Big(^{	K^{\prime}_{1}, \dots, K^{\prime}_{m}, K_1 \dots, K_n		}_{w_1,\dots,w_m,v_1,\dots,v_n}				\Big)
		= \cals{A}_i^{*} \cdot
		G^{\Omega^1\cdots\Omega^m}_{\Lambda^1\cdots\Lambda^n}
		\Big(^{	K^{\prime}_{1}, \dots, K^{\prime}_{m}, K_1 \dots, K_n		}_{w_1,\dots,w_m,v_1,\dots,v_n}				\Big),
		\notag \\
		\cals{A}_i^{*} := (K^{\prime}_{i})^{-|\Omega^i|}
		\Big(	\prod_{k=i+1}^{m}\tilde{R}^{K^{\prime}_{i},K^{\prime}_{k}}_{\Omega^i,\Omega^k}\big( \gq^{-2}\frac{w_i}{w_k}  \big)	\Big)
		\Big(	\prod_{l=1}^{n} 
		\tilde{R}^{K^{\prime}_{i},K_l}_{\Omega^i,\Lambda^l}\big( \gq^{-1}\frac{w_i}{v_l}			\big)^{-1}		\Big)
		\Big(	\prod_{s=1}^{i-1}  \tilde{R}^{K^{\prime}_{s},K^{\prime}_{i}}_{\Omega^s,\Omega^i}
		\big(		\frac{w_s}{w_i}				\big)^{-1}
		\Big).
		\label{MacKZ2}
	\end{gather}
	The system of Eqs.\eqref{MacKZ1} and \eqref{MacKZ2} is the generalized KZ equation for MacMahon intertwiners. 
	
	
	\section{Solutions to the generalized KZ equation}
	\label{section9}
	
	In this section we discuss solutions to the generalized KZ equation derived in the last section.
	Since the generalized KZ equation is a system of linear difference equations for the vertical spectral parameters of
	the intertwiners, solutions are given up to a multiplicative constant or a $\gq^2=q_3$ periodic function. 
	In particular we will drop the normalization factors which are independent of the spectral parameters\footnote{
		This does not mean that the normalization factor is unimportant. For example, the normalization factor of the Fock
		intertwiner is closely related to the theory of Macdonald function.}.
	Due to the abelian nature of the $R$ matrices (recall they are diagonal), the solutions factorize into 
	the product or more precisely the ration of the fundamental solutions, namely two point functions,
	which reminds us of the Wick theorem for free field correlators.
	
	\subsection{Correlation function of the vector intertwiners}
	
	Since both the intertwiner and the dual intertwiner keep the level of the horizontal Fock representation,
	we may assume that it is zero without loss of generality. Under this assumption 
	$z_n(v)$ and $z_n^*(v)$ are independent of the spectral parameter (see \eqref{levelN}) and 
	hence for simplicity we will drop these factors in the following formulas;
	As mentioned above we expect that the general correlation function $ \mathcal{G}_{\lambda_1,\cdots,\lambda_n}^{\mu_1,\cdots,\mu_m}
	\big(	u_n \big|~{w_1,\cdots,w_m, z_1,\cdots,z_m} 		\big) $ can be expressed in terms of the two-point function.
	Thus, let us begin with the two-point correlator 
	\begin{equation}
		A_{nm}(v,w) := \langle 0 \big| \bb{I}_n(v)\bb{I}_m(w)\big| 0 \rangle 
		= \exp\bigg[
		\sum_{r=1}^{\infty}\Big(	\frac{q^n_1v}{q^m_1w}				\Big)^{-r}
		\frac{q_1^r}{r}\frac{1-q^r_2}{1-q^r_1}
		\bigg].
	\end{equation}
	Thus, we can confirm the generalized KZ equation for two-point function;
	\begin{gather}\label{vector2point}
		\frac{A_{nm}(\gq^{-2} v, w)}{A_{nm}(v,w)} = R_{nm}\Big(	v/w				\Big).
	\end{gather}
	By using \eqref{vector2point} one can show that 
	\begin{gather}\label{vectorKZsol}
		\mathcal{G}_{\lambda_1,\cdots,\lambda_n}^{\mu_1,\cdots,\mu_m}
		\big(	u_n \big|~{w_1,\cdots,w_m, z_1,\cdots,z_m} 		\big)
		\frac{\prod\limits_{	 \substack{ i, j = 1 \\ i < j  }				}^{n}A_{\lambda_i\lambda_j}
			\Big(	\frac{z_i}{z_j}			\Big)\prod\limits_{	 \substack{ i, l = 1 \\  l < i}				}^{m}
			A_{\mu_l\mu_i}\Big(	\frac{w_l}{\gq^2w_i}					\Big)}{
			\prod_{j=1}^m\prod_{l=1}^nA_{\mu_j\lambda_l}\Big(
			\frac{\gq w_j}{\gq^2 z_l}
			\Big)}
	\end{gather}
	solves the generalized KZ equation. We should emphasize again that the formula is valid up to a multiplicative constant.

	The AGT correspondence motivates us to compare $ A_{nm}(v,w) $ with the Nekrasov factor \eqref{1.2}.
	Under the normalization
	\begin{gather}
		\tilde{A}_{nm}(v,w) := \frac{A_{nm}(v,w)}{A_{00}(v,w)}
		= \exp\bigg[
		\sum_{r=1}^{\infty}(	q^{(m-n)r}_{1} - 1		)\Big(	\frac{v}{w}				\Big)^{-r}
		\frac{q^r_1}{r}\frac{1-q^r_2}{1-q^r_1}
		\bigg].
	\end{gather}
	by the relation \eqref{q-factorial}, we see
	\begin{align}
		\tilde{A}_{nm}(v,w) = \frac{\Big(	 \frac{wq_1}{v} ; q_1	\Big)_{\infty}\Big(		\frac{wq^{m-n+1}_1q_2}{v} ; q_1	\Big)_{\infty}}{\Big(	\frac{wq_1q_2}{v} ; q_1		\Big)_{\infty}\Big(		\frac{wq^{m-n+1}_1}{v} ; q_1		\Big)_{\infty}}
		=  
		\begin{cases}
			\frac{	N_{\emptyset,m-n}(	w/q_2, v	)		}{N_{\emptyset,m-n}(w,v)}
			%
			;~ m \geq n 
			\\
			\frac{	N_{n-m,\emptyset}(q_2w,v)				}{N_{n-m,\emptyset}(w,v)}
			%
			; ~ n \geq m 
		\end{cases}
	\end{align}
	where the non-negative integer in the subscript stands for the Young diagram with a single row. 
	The insertion of the screening operator is required to obtain more general solution. 
	In particular with the contour integral associated with the spectral parameter of the screening operators 
	the correlation function \eqref{vectorKZsol} provides building block of the partition function of
	three dimensional quiver gauge theories, or the $K$ theoretic lift of the vortex counting function
	\cite{YS} \cite{AHKS} \cite{AHS} \cite{AS} \cite{BKK} \cite{AFO} \cite{NPZ} \cite{APZ} \cite{P} \cite{LNP}. 
	Quite recently an elliptic lift of the correlation function of the
	vector intertwiners is discussed in \cite{GKKZ} based on the \say{Higgsed} network calculus and the elliptic DIM algebra
	\cite{Saito} \cite{IKY} \cite{Ni}.

	\subsection{Correlation function of the MacMahon intertwiners}
	
	We also expect that the general correlation function
	$ G^{\Omega^1\cdots\Omega^m}_{\Lambda^1\cdots\Lambda^n}
	\Big(^{	K^{\prime}_{1}, \dots, K^{\prime}_{m}, K_1 \dots, K_n		}_{w_1,\dots,w_m,v_1,\dots,v_n}				\Big) $ 
	can be expressed in terms of the two-point correlators. We first provide an explicit expression of
	$ \langle \emptyset | \Xi_{\Lambda}(K_1 ; v)\Xi_{\Lambda^{\prime}}(K_2 ; w)|\emptyset\rangle $.
	Up to the normalization factor such as $\cals{G}^{-1}_{\Lambda}$ which is independent of the spectral parameter,
	it is 
	\begin{align}
		A^{K_1K_2}_{\Lambda\Lambda^{\prime}}(v,w) :&= \langle \emptyset | \Xi_{\Lambda}(K_1 ; v)\Xi_{\Lambda^{\prime}}(K_2 ; w)|\emptyset\rangle 
		\notag \\ 
		&= z_{\Lambda}(K_1 ; v)z_{\Lambda^{\prime}}(K_2 ; w) E^{K_1K_2}_{\Lambda\Lambda^{\prime}}(v,w),
		\label{Mac2pt}
	\end{align}
	where
	\begin{align}
		E^{K_1K_2}_{\Lambda\Lambda^{\prime}}&(v,w)
		\notag \\ 
		= \exp&\bigg[
		-\sum_{r=1}^{\infty}\Big(	\frac{w}{v}			\Big)^r
		\bigg(	 \sum_{(i,j,k) \in \Lambda^{\prime}}x^r_{ijk} + \frac{1 - K^r_2}{\kappa_r}			\bigg)
		\bigg(	 \sum_{	(i,j,k) \in \Lambda				}x^{-r}_{ijk} - \frac{1-K^{-r}_{1}}{\kappa_r}			\bigg)\frac{ (1-q^r_1)(1-q^r_2)			}{r}
		\bigg].
	\end{align}
	But from \eqref{5.13} and $ \kappa_r = (q^r_1 - 1)(q^r_2 - 1)(q^r_3 - 1) $ we obtain 
	\begin{align}
		&\tilde{R}^{K_1,K_2}_{\Pi,\Lambda}\big(
		\frac{v}{u}
		\big)
		\notag \\ 
		&= \frac{
			\exp\bigg(
			\displaystyle{\sum_{r=1}^{\infty}} \frac{(1-q^r_1)(1-q^r_2)}{r}
			\big(	\frac{u}{v}		\big)^r\bigg(
			\sum_{	(i,j,k) \in \Lambda				}x^r_{ijk} + \frac{1-K^r_2}{\kappa_r}
			\bigg)
			\bigg(
			\sum_{	(i,j,k) \in \Pi				}x^{-r}_{ijk} - \frac{1-K^{-r}_{1}}{\kappa_r}
			\bigg)
			\bigg)
		}{\exp\bigg(
			\displaystyle{\sum_{r=1}^{\infty}} \frac{(1-q^r_1)(1-q^r_2)}{r}
			\big(	\frac{q_3u}{v}		\big)^r\bigg(
			\sum_{	(i,j,k) \in \Lambda				}x^r_{ijk} + \frac{1-K^r_2}{\kappa_r}
			\bigg)
			\bigg(
			\sum_{	(i,j,k) \in \Pi				}x^{-r}_{ijk} - \frac{1-K^{-r}_{1}}{\kappa_r}
			\bigg)
			\bigg)}
		\notag \\ 
		&= \frac{ E^{K_1K_2}_{\Pi\Lambda}(\gq^{-2} v,u)				}{E^{K_1K_2}_{\Pi\Lambda}(v,u)},
		\label{9.2}
	\end{align}
	which is nothing but the generalized KZ equation for two point function up to a factor coming from the zero modes. 
	Hence, we see that $ G^{\Omega^1\cdots\Omega^m}_{\Lambda^1\cdots\Lambda^n}
	\Big(^{	K^{\prime}_{1}, \dots, K^{\prime}_{m}, K_1 \dots, K_n		}_{w_1,\dots,w_m,v_1,\dots,v_n}				\Big) $ has
	to be proportional to 
	\begin{gather}
		\frac{
			\prod\limits_{	 \substack{ j, l = 1 \\ j < l }				}^{n}
			E^{K_j,K_l}_{\Lambda^j,\Lambda^l}(v_j,v_l)
			\prod\limits_{	 \substack{ l, j = 1 \\ j < l }				}^{m}
			E^{K^{\prime}_{j},K^{\prime}_{l}}_{\Omega^j,\Omega^l}(w_j, \gq^2w_l)
		}{
			\prod_{l=1}^{n}\prod_{s=1}^{m}E^{	K^{\prime}_{s}, K_l		}_{\Omega^s,\Lambda^l}(w_s,\gq v_l)
		}. 
	\end{gather}
	More precisely, we get that up to a multiplicative constant
	\begin{align}
		&G^{\Omega^1\cdots\Omega^m}_{\Lambda^1\cdots\Lambda^n}
		\Big(^{	K^{\prime}_{1}, \dots, K^{\prime}_{m}, K_1 \dots, K_n		}_{w_1,\dots,w_m,v_1,\dots,v_n}				\Big) 
		\nn \\
		= 
		&\bigg(\prod_{i=1}^{n}z_{\Lambda^i}(	K_i ; v_i			)\bigg)
		\bigg(\prod_{i=1}^{m}z^*_{\Omega^i}(K^{\prime}_{i};w_i)\bigg)
		\frac{
			\prod\limits_{	 \substack{ j, l = 1 \\ j < l }				}^{n}
			E^{K_j,K_l}_{\Lambda^j,\Lambda^l}(v_j,v_l)
			\prod\limits_{	 \substack{ l, j = 1 \\ j < l }				}^{m}
			E^{K^{\prime}_{j},K^{\prime}_{l}}_{\Omega^j,\Omega^l}(w_j, \gq^2w_l)
		}{
			\prod_{l=1}^{n}\prod_{s=1}^{m}E^{	K^{\prime}_{s}, K_l		}_{\Omega^s,\Lambda^l}(w_s,\gq v_l)
		},
		\notag \\ 
		&= 
		\frac{	\prod_{l=1}^{n} z_{\Lambda^l}(	K_l ; v_l			)z_{\Lambda^l}(K_l ; \gq v_l) \prod_{s=1}^{m}z^*_{\Omega^s}(K^{\prime}_{s};w_s)
			z_{\Omega^s}(K^{\prime}_s ; w_s)						}{
			\prod\limits_{	 \substack{ j, l = 1 \\ j < l }				}^nz_{\Lambda^j}(K_j ; v_j)z_{\Lambda^l}(K_l ; v_l)
			\cdot
			\prod\limits_{	 \substack{ l, j = 1 \\ j < l }}^mz_{\Omega^j}(K^{\prime}_{j} ; w_j)z_{\Omega^l}(K^{\prime}_l ; \gq^2w_l)		
		}
		\nn \\
		&\qquad\cdot 
		\frac{
			\prod\limits_{	 \substack{ j, l = 1 \\ j < l }				}^{n}
			A^{K_j,K_l}_{\Lambda^j,\Lambda^l}(v_j,v_l)
			\prod\limits_{	 \substack{ l, j = 1 \\ j < l }				}^{m}
			A^{K^{\prime}_{j},K^{\prime}_{l}}_{\Omega^j,\Omega^l}(w_j, \gq^2w_l)
		}{
			\prod_{l=1}^{n}\prod_{s=1}^{m}A^{	K^{\prime}_{s}, K_l		}_{\Omega^s,\Lambda^l}(w_s,\gq v_l)
		} .
	\end{align}

	In view of the AGT correspondence, it is quite curious to see if the solutions to the MacMahon KZ-equation give
	some generalization of the Nekrasov function. 
	Let us analyze the two-points correlator $A^{K_1K_2}_{\Lambda\Lambda^{\prime}}(v,w)$ (see \eqref{Mac2pt}).
	Defining the normalized two-point function by
	\begin{equation}
		\tilde{A}^{K_1K_2}_{\Lambda\Lambda^{\prime}}(v,w)
		:= \frac{A^{K_1K_2}_{\Lambda\Lambda^{\prime}}(v,w) }{A^{K_1K_2}_{\emptyset\emptyset}(v,w) },
	\end{equation}
	we see that 
	\begin{align}
		\log&\tilde{A}^{K_1K_2}_{\Lambda\Lambda^{\prime}}(v,w)
		= \log\bigg(	z_{\Lambda}(K_1 ; v)z_{\Lambda^{\prime}}(K_2 ; w)		\bigg)
		+ \mathrm{(I)} +  \mathrm{(II)} +  \mathrm{(III)},
		\label{normalized2pt}
	\end{align}
	where 
	\begin{align}
		\mathrm{(I)}  &= - \sum_{r=1}^{\infty}\Big( \frac{w}{v}
		\Big)^r\frac{(1-q^r_1)(1-q^r_2)}{r}
		\bigg(	\sum_{(i^{\prime},j^{\prime},k^{\prime}) \in \Lambda^{\prime}}x^r_{i^{\prime}j^{\prime}k^{\prime}}
		\sum_{	(i,j,k) \in \Lambda				}x^{-r}_{ijk}				\bigg)
		\notag \\ 
		&= \sum_{(i^{\prime},j^{\prime},k^{\prime}) \in \Lambda^{\prime}}\sum_{	(i,j,k) \in \Lambda				}
		\notag \\ 
		&\bigg[
		\log\big(	1 - \frac{wx_{i^{\prime},j^{\prime},k^{\prime}}}{vx_{ijk}}					\big)
		- \log\big(		1 - \frac{wq_2x_{i^{\prime},j^{\prime},k^{\prime}}	}{vx_{ijk}}		\big)
		-\log\big(		1 - \frac{wq_1x_{i^{\prime},j^{\prime},k^{\prime}}	}{vx_{ijk}}								\big)
		+\log\big( 1 - 
		\frac{wq_1q_2x_{i^{\prime},j^{\prime},k^{\prime}}	}{vx_{ijk}}			
		\big)
		\bigg].
	\end{align}
	Since $ \kappa_r = (q^r_1 - 1)(q^r_2 - 1)(q^r_3 - 1) $, the second term is 
	\begin{gather}
		\mathrm{(II)}  = \sum_{r=1}^{\infty}\Big( \frac{w}{v}
		\Big)^r\frac{1}{r}\cdot\frac{1-K^{-r}_{1}}{q^r_3 - 1}\sum_{(i,j,k) \in \Lambda^{\prime}}x^r_{ijk}
		= \sum_{r=1}^{\infty}\Big( \frac{w}{vK_1}
		\Big)^r\frac{1}{r}\cdot\frac{K^r_1 - 1}{q^r_3 - 1}\sum_{(i,j,k) \in \Lambda^{\prime}}x^r_{ijk}.
		\label{sdfusoieruwiroweuoeiwufsxjslkjsdofiusdofudsfoisdu}
	\end{gather}
	Using the identity \eqref{q-factorial},
	we see that Eq. \eqref{sdfusoieruwiroweuoeiwufsxjslkjsdofiusdofudsfoisdu} becomes 
	\begin{align}
		\sum_{r=1}^{\infty}&\Big( \frac{w}{vK_1}
		\Big)^r\frac{1}{r}\bigg[
		\frac{K^r_1}{q^r_3 - 1} - \frac{1}{q^r_3 - 1}
		\bigg]\sum_{(i,j,k) \in \Lambda^{\prime}}x^r_{ijk}
		\notag \\ 
		&= - \sum_{(i,j,k) \in \Lambda^{\prime}}\sum_{r=1}^{\infty}\frac{1}{r(1-q^r_3)}
		\Big(\frac{w}{v}x_{ijk}\Big)^r
		+ \sum_{(i,j,k) \in \Lambda^{\prime}}\sum_{r=1}^{\infty}\frac{1}{r(1-q^r_3)}
		\Big( \frac{w}{vK_1}x_{ijk}			\Big)^r
		\notag \\
		&= \sum_{(i,j,k) \in \Lambda^{\prime}}\log\Big(	\frac{w}{v}x_{ijk}	;	q_3	\Big)_{\infty}
		- \sum_{(i,j,k) \in \Lambda^{\prime}}\log\Big(	\frac{w}{vK_1}x_{ijk}	;	q_3	\Big)_{\infty}.
	\end{align}
	Similarly the third term becomes
	\begin{align}
		\mathrm{(III)} = \sum_{r=1}^{\infty}&\Big( \frac{w}{v}
		\Big)^r\frac{1}{r}\frac{1 - K^r_2}{1 - q^r_3 }\sum_{	(i,j,k) \in \Lambda				}x^{-r}_{ijk}
		\notag \\ 
		&= \sum_{	(i,j,k) \in \Lambda				}\sum_{r=1}^{\infty}
		\frac{1}{r(1-q^r_3)}\Big(			\frac{w}{vx_{ijk}}		\Big)^r
		- \sum_{	(i,j,k) \in \Lambda				}\sum_{r=1}^{\infty}
		\frac{1}{r(1-q^r_3)}\Big(			\frac{wK_2}{vx_{ijk}}		\Big)^r
		\notag \\
		&= -\sum_{	(i,j,k) \in \Lambda				}\log\big(	\frac{w}{vx_{ijk}}		; q_3	\big)_{\infty}
		+ \sum_{	(i,j,k) \in \Lambda				}\log\big( \frac{wK_2}{vx_{ijk}}	
		; q_3
		\big)_{\infty}.
	\end{align}

	Now merging all the three terms, we obtain
	\begin{align}
		\tilde{A}^{K_1K_2}_{\Lambda\Lambda^{\prime}}(v,w) = &z_{\Lambda}(K_1 ; v)z_{\Lambda^{\prime}}(K_2 ; w)
		\notag \\
		&\cdot 
		\bigg[
		\prod_{(i^{\prime},j^{\prime},k^{\prime}) \in \Lambda^{\prime}}\prod_{	(i,j,k) \in \Lambda				}
		\frac{ \big(	1 - \frac{wx_{i^{\prime}j^{\prime}k^{\prime}}}{vx_{ijk}}		\big) \big( 1 - 
			\frac{wq_1q_2x_{i^{\prime},j^{\prime},k^{\prime}}	}{vx_{ijk}}			
			\big)					}{\big(1 - \frac{wq_2x_{i^{\prime},j^{\prime},k^{\prime}}	}{vx_{ijk}}\big) \big(		1 - \frac{wq_1x_{i^{\prime},j^{\prime},k^{\prime}}	}{vx_{ijk}}								\big)	}
		\bigg]
		\cdot 
		\bigg[
		\prod_{	(i,j,k) \in \Lambda^{\prime}				}\frac{
			\Big(	\frac{w}{v}x_{ijk}	;	q_3	\Big)_{\infty}
		}{\Big(	\frac{w}{vK_1}x_{ijk}	;	q_3	\Big)_{\infty}}
		\bigg]
		\notag \\ 
		&\cdot 
		\bigg[
		\prod_{	(i,j,k) \in \Lambda				}
		\frac{\big( \frac{wK_2}{vx_{ijk}}	
			; q_3
			\big)_{\infty}}{\big(	\frac{w}{vx_{ijk}}		; q_3	\big)_{\infty}}
		\bigg].\label{MKZsol}
	\end{align}
	To show the relation of $ \tilde{A}^{K_1K_2}_{\Lambda\Lambda^{\prime}}(v,w) $ and the Nekrasov function, 
	we show that under a certain appropriate limit the inverse of the normalized two-point correlators 
	$ \tilde{A}^{K_1K_2}_{\Lambda\Lambda^{\prime}}(v,w) $ reduces to the Nekrasov function 
	up to multiplicative constant\footnote{Note that in the ordinary case, the Nekrasov functions appear in the denominator of the two-point correlators.}. 
	Using $ x_{ijk} = q^{i-1}_{1}q^{j-1}_{2}q^{k-1}_{3} = q^i_1q^j_2q^k_3 $, we see that under the conditions $ K_1 = K_2 = q_3 $ and $ k = 1 $, 
	\begin{align}
		\tilde{A}^{q_3,q_3}_{\lambda\lambda^{\prime}}(v,w)^{-1}
		= &z_{\lambda}(q_3 ; v)^{-1}z_{\lambda^{\prime}}(q_3 ; w)^{-1}
		\notag \\
		&\cdot 
		\bigg[
		\prod_{	(i^{\prime},j^{\prime}) \in \lambda^{\prime}			}\prod_{(i,j) \in \lambda}
		\frac{	\big(1 - \frac{w}{v}q_2q^{i^{\prime} - i}_{1}q^{j^{\prime} - j}_{2}\big)		\big(	1 - \frac{w}{v}q_1q^{i^{\prime} - i}_{1}q^{j^{\prime} - j}_{2}			\big)							}{	\big(	1 - \frac{w}{v}q^{i^{\prime} - i}_{1}q^{j^{\prime} - j}_{2}		\big)		\big(		1 - \frac{w}{v}q_1q_2q^{i^{\prime} - i}_{1}q^{j^{\prime} - j}_{2}			\big)					}
		\bigg]
		\cdot 
		\bigg[
		\prod_{(	i^{\prime}, j^{\prime}) \in \lambda^{\prime}			}
		\big(
		1 - \frac{w}{v}q^{i^{\prime}}_{1}q^{j^{\prime}}_{2}
		\big)
		\bigg]
		\notag \\
		&\cdot 
		\bigg[
		\prod_{(i,j) \in \lambda}\big(
		1 - q_1q_2\frac{w}{vq^{i}_{1}q^{j}_{2}}
		\big)
		\bigg].
	\end{align}
	Comparing with the Nekrasov function \eqref{1.2}, we see that 
	\begin{gather}
		\tilde{A}^{q_3,q_3}_{\lambda\lambda^{\prime}}(v,w)^{-1}
		= z_{\lambda}(q_3 ; v)^{-1}z_{\lambda^{\prime}}(q_3 ; w)^{-1}
		\cdot N_{\lambda \lambda^{\prime}}(w,v).
	\end{gather}
	This justifies that $ \tilde{A}^{K_1K_2}_{\Lambda\Lambda^{\prime}}(v,w)^{-1} $ can be regarded as a generalized Nekrasov function.

	It is an interesting challenge to understand the solution \eqref{MKZsol} as some kind of the partition function of supersymmetric gauge
	theory and/or generalization of Macdonald function \cite{Z17} \cite{Morozov}. 
	For example, it is tempting to relate it to six dimensional gauge theory.
	However the direction $q_3$ appears as \say{preferred} in \eqref{MKZsol} and it seems to suggest the existence of 
	an additional defect in the theory\footnote{If we adopt a suggestion in \cite{Z20} that the MacMahon representation may be associated with
		a system of $D7$ and anti-$D7$ branes, the defect would be five branes attached to it.}.
	Though the MacMahon representation is symmetric in $(q_1, q_2, q_3)$,
	the horizontal Fock representation breaks it by the commutation relation with $\gq=q_3^{1/2}$. 
	The appearance of the infinite product is another issue against an interpretation in terms of the gauge theory.

	{\vskip 20mm}
	
	\begin{center}
		{\bf Acknowledgements}
	\end{center}
	We would like to thank H.~Awata, A.~Mironov, A.~Morozov and Y.~Zenkevich for useful discussions. 
	Our work is supported in part by Grants-in-Aid for Scientific Research (\# 18K03274) (H.K.).
	The work of P.C. is supported by the MEXT Scholarship. 
	
	\bigskip
	
	\appendix

	\section{Determine the proportional factor}
	\label{subsection4.1.4}
	
	Here we determine the proportional factor of \eqref{4.35} explicitly. From
	\eqref{macmahonansatttz} we get that 
	\begin{align}
		\Xi_{\Lambda}(K ; pv) =
		\Big[
		\frac{z_{\lambda}(K;pv)}{z_{\lambda}(K;v)}\Big]z_{\lambda}(K;v)\cals{M}^{[n]}(K)\tilde{\Phi}^{[n]}_{\Lambda}(pv)\Gamma_n(K;pv).
		\label{4.24ver2}
	\end{align}
	Now we can write
	\begin{equation}
		\tilde{\Phi}^{[n]}_{\Lambda}(pv)\Gamma_n(K;pv)
		= A_{\Lambda}^{(-)}(v;p) \tilde{\Phi}^{[n]}_{\Lambda}(v) A_{\Lambda}^{(+)}(v;p)
		\cdot B_n^{(-)}(v,p) \Gamma_n(K;v) B_n^{(+)}(v,p),
	\end{equation}
	where
	\begin{gather}
		A_{\Lambda}^{(\pm)}(v;p) = \exp\bigg[ \mp
		\sum_{r=1}^{\infty}\frac{(p^{\mp r}-1) H_{\pm r}}{\gq^r - \gq^{-r}}
		\sum_{k=1}^{n}\gq^{-r/2}(q^{k-1}_3 v				)^{\mp r} 
		\chi_{\mp r}^{(k)}(q_1,q_2)
		\bigg],
		\notag \\ 
	\end{gather}
	and 
	\begin{gather}
		B_n^{(\pm)}(v,p) = \exp\bigg(
		\sum_{r=1}^{\infty} \frac{ (p^{\mp r }- 1)H_{ \pm r}}{\gq^r - \gq^{-r}}\frac{q^{\mp nr}_3 - K^{\mp r}}{\kappa_r}
		\gq^{-r/2} v^{\mp r}
		\bigg),
	\end{gather}
	where we have used the notation \eqref{chi(k)}.
	By using \eqref{CBH}, we get that 
	\begin{gather*}
		A_{\Lambda}^{(+)}(v;p) B_n^{(-)}(v,p) 
		= \exp\big(	\clubsuit		\big)
		B_n^{(-)}(v,p) A_{\Lambda}^{(+)}(v;p),
	\end{gather*}
	where
	\begin{gather*}
		\clubsuit =
		-\sum_{r=1}^{\infty}
		(p^{-r} - 1)(p^r - 1)\bigg(
		\frac{q^{nr}_{3} - K^r}{\kappa_r}
		\bigg)
		\frac{(1-q^r_1)(1-q^r_2)}{r}
		\bigg(
		\sum_{(i,j,k) \in \Lambda}x^{-r}_{ijk} - \frac{1-q^{-rn}_{3}}{\kappa_r}
		\bigg).
	\end{gather*}
	
	Next we would like to switch between $A_{\Lambda}^{(+)}(v;p)$
	and $ \Gamma_n(K;v) $. We know the explicit expression of $ \Gamma_n(K;v) $ from Eq. \eqref{4.3}. 
	So by using \eqref{CBH} again, we can show that 
	\begin{equation}
		A_{\Lambda}^{(+)}(v;p) \Gamma_n(K;v)
		\notag \\ 
		= \exp\big(	\spadesuit		\big)
		\Gamma_n(K;v) A_{\Lambda}^{(+)}(v;p) 
	\end{equation}  
	with
	\begin{gather}
		\spadesuit = -\sum_{r=1}^{\infty}(p^{-r} - 1)\bigg(
		\sum_{	(i,j,k) \in \Lambda			}x^{-r}_{ijk} - \frac{1-q^{-rn}_{3}}{\kappa_r}
		\bigg)\frac{q^{nr}_{3} - K^r }{\kappa_r}\frac{(1-q^r_1)(1-q^r_2)}{r}. 
	\end{gather}

	In the last step we would like to switch between $ \tilde{\Phi}^{[n]}_{\Lambda}(v) $  and
	$B_n^{(-)}(v,p)$.
	Again since we know the explicit expression of $ \tilde{\Phi}^{[n]}_{\Lambda}(v) $ from Lemma \ref{lemma4.1}, 
	and \eqref{CBH} implies
	\begin{equation}
		\tilde{\Phi}^{[n]}_{\Lambda}(v) B_n^{(-)}(v,p)
		= \exp \big( \heartsuit  \big) 
		B_n^{(-)}(v,p) \tilde{\Phi}^{[n]}_{\Lambda}(v),
	\end{equation}
	where
	\begin{gather*}
		\heartsuit =
		-\sum_{r=1}^{\infty}\frac{q^{nr}_{3} - K^r}{\kappa_r}(p^r - 1)
		\frac{(1-q^r_1)(1-q^r_2)}{r}
		\bigg(
		\sum_{	(i,j,k) \in \Lambda			}x^{-r}_{ijk} - \frac{1 - q^{-rn}_{3}}{\kappa_r}
		\bigg). 
	\end{gather*}
	Consequently, 
	\begin{gather}
		\tilde{\Phi}^{[n]}_{\Lambda}(pv)\Gamma_n(K;pv)
		= \exp\big(	\clubsuit		\big)\exp\big(	\spadesuit		\big)\exp\big(	\heartsuit		\big)
		\notag \\ 
		\cdot 
		A_{\Lambda}^{(-)}(v;p)  B_n^{(-)}(v,p) \cdot 
		\tilde{\Phi}^{[n]}_{\Lambda}(v)
		\Gamma_n(K;v)
		\cdot 
		A_{\Lambda}^{(+)}(v;p) B_n^{(+)}(v,p).
		\label{4.31}
	\end{gather}
	From the explicit expressions of $ \clubsuit, \spadesuit, $ and $ \heartsuit $ we get that 
	$\clubsuit + \spadesuit + \heartsuit = 0$.
	In conclusion,
	\begin{align}
		&\Xi_{\Lambda}(K;pv)
		= \frac{		z_{\Lambda}(K;pv)		}{z_{\Lambda}(K;v)}
		\nn \\
		& \cdot\exp\bigg(
		\sum_{r=1}^{\infty} \frac{(p^r-1)H_{-r}}{(\gq^r - \gq^{-r})}
		\big\{
		(\gq^{-1/2}v)^r\big(
		\frac{q^{nr}_3 - K^r}{\kappa_r}
		\big)
		+ \sum_{k=1}^n (	\gq^{-1/2}q^{k-1}_3v		)^r \chi_{r}^{(k)} (q_1, q_2)
		\big\}
		\bigg)\cdot \, \Xi_{\Lambda}(K;v)
		\notag \\
		&
		\cdot \exp\bigg(
		\sum_{r=1}^{\infty}\frac{(		p^{-r}-1	)H_r}{(\gq^r - \gq^{-r})}
		\big\{
		(\gq^{1/2}v)^{-r}\frac{	q^{-nr}_3 - K^{-r}			}{\kappa_r}
		- \sum_{k=1}^{n}(	\gq^{1/2}q^{k-1}_3v				)^{-r} \chi_{-r}^{(k)} (q_1, q_2)
		\big\}
		\bigg).
	\end{align}
	When $p=\gq^{-2}=q_3^{-1}$, by \eqref{p-shift} and \eqref{6.13} we can show
	$ \displaystyle \frac{	z_{\Lambda}(K;\gq^{-2}v)			}{	z_{\Lambda}(K;v)			} = K^{|\Lambda|} $,
	and hence we arrive at \eqref{4.35}.
	%
	
	Next  we would like to we determine the proportional factor in \eqref{4.62} explicitly 
	From Eq. \eqref{dualMacint} we know that 
	\begin{align}
		\Xi_{\Lambda}^{*}(K;pv) =
		\frac{z_{\Lambda}^*(K;pv)}{z_{\Lambda}^*(K;v)} z_{\Lambda}^*(K;v)\cals{M}^{[n]*}(K)\tilde{\Phi}^{[n]*}_{\Lambda}(pv)\Gamma^{*}_{n}(K;pv).
		\label{4.50}
	\end{align}
	We can write
	\begin{equation}
		\tilde{\Phi}^{[n]*}_{\Lambda}(pv)\Gamma^{*}_{n}(K;pv)
		= C_{\Lambda}^{(-)}(v;p) \tilde{\Phi}^{[n]*}_{\Lambda}(v)C_{\Lambda}^{(+)}(v;p) 
		\cdot
		D_{n}^{(-)}(v;p)\Gamma^*_n(K;v) D_{n}^{(+)}(v;p),
	\end{equation}
	where
	\begin{gather}
		C_{\Lambda}^{(\pm)}(v;p) =
		\exp\bigg[ \pm
		\sum_{m=1}^{n}\sum_{r=1}^{\infty}
		\frac{(p^{\mp r} - 1)H_{\pm r}}{\gq^r - \gq^{-r}}\gq^{r/2}(q^{m-1}_{3}v		)^{\mp r}
		\chi_{\mp r}^{(m)} (q_1, q_2)
		\bigg],
		\notag 
	\end{gather}
	and 
	\begin{gather}
		D_{n}^{(\pm)} (v;p) = \exp\bigg(
		-\sum_{r=1}^{\infty} \frac{(p^{\mp r} - 1) H_{\pm r}}{\gq^r - \gq^{-r}} \frac{q^{\mp nr}_3 - K^{\mp r}}{\kappa_r}
		\gq^{r/2}v^{\mp r}\bigg).
	\end{gather}
	By using \eqref{CBH}, we get that
	\begin{equation}
		C_{\Lambda}^{(+)}(v;p) D_{n}^{(-)} (v;p) 
		= \exp\big(	\bigstar	\big)
		D_{n}^{(-)} (v;p) C_{\Lambda}^{(+)}(v;p)
	\end{equation}
	where
	\begin{gather*}
		\bigstar = -\sum_{r=1}^{\infty}(p^{-r}-1)(p^{r}-1)\frac{q^{nr}_{3} - K^r}{\kappa_r}
		\frac{q^r_3(q^r_1 - 1)(q^r_2 - 1)}{r}
		\cdot\bigg[
		\sum_{	(i,j,k) \in \Lambda				}x^{-r}_{ijk} - \frac{1-q^{-rn}_{3}}{\kappa_r}
		\bigg].
	\end{gather*}
	
	Next we would like to switch between $C_{\Lambda}^{(+)}(v;p)$
	and $ \Gamma^*_n(K;v) $. We know the explicit expression of $ \Gamma^*_n(K;v) $
	from \eqref{4.36}. So by using \eqref{CBH} again, we can show that 
	\begin{gather}
		C_{\Lambda}^{(+)}(v;p) \Gamma^*_n(K;v)
		= \exp\big(		\mho	\big) \Gamma^*_n(K;v) C_{\Lambda}^{(+)}(v;p),
	\end{gather}
	where 
	\begin{gather}
		\mho =
		-\sum_{r=1}^{\infty}(p^{-r} - 1)\frac{q^{nr}_{3} - K^r}{\kappa_r}
		\frac{	q^r_3(q^r_1 - 1)(q^r_2 - 1)			}{r}
		\bigg(
		\sum_{	(i,j,k) \in \Lambda			}x^{-r}_{ijk} - \frac{1-q^{-nr}_{3}}{\kappa_r}
		\bigg). 
	\end{gather}

	Finally we would like to switch between $ \tilde{\Phi}^{[n]*}_{\Lambda}(v) $ and $D_{n}^{(-)} (v;p)$.
	Again since we know the explicit expression of $ \tilde{\Phi}^{[n]*}_{\Lambda}(v) $ from Lemma \ref{lemma4.3}, 
	we can use \eqref{CBH} to show that
	\begin{equation}
		\tilde{\Phi}^{[n]*}_{\Lambda}(v) D_{n}^{(-)} (v;p)
		= \exp\big(		\natural		\big)
		D_{n}^{(-)} (v;p) \tilde{\Phi}^{[n]*}_{\Lambda}(v),
	\end{equation}
	where 
	\begin{gather}
		\natural =
		-\sum_{r=1}^{\infty}(p^r-1)\frac{q^{nr}_{3} - K^r}{k_r}\frac{q^r_3(q^r_1 - 1)(q^r_2 - 1)}{r}\bigg(
		\sum_{	(i,j,k) \in \Lambda			}x^{-r}_{ijk} - \frac{1 - q^{-rn}_{3}}{\kappa_r}
		\bigg).
	\end{gather}
	Consequently, 
	\begin{gather}
		\tilde{\Phi}^{[n]*}_{\Lambda}(pv)\Gamma^{*}_{n}(K;pv)
		= \exp\big(	 \bigstar	\big)\exp\big(	\mho		\big)\exp\big(	\natural		\big)
		\notag \\ 
		C_{\Lambda}^{(-)}(v;p) D_{n}^{(-)} (v;p) \cdot \tilde{\Phi}^{[n]*}_{\Lambda}(v) 
		\Gamma^*_n(K;v) \cdot C_{\Lambda}^{(+)}(v;p) D_{n}^{(+)} (v;p).
		\label{4.59}
	\end{gather}
	Substituting \eqref{4.59} into \eqref{4.50}, we see that the proportional factor of \eqref{4.62} is 
	\begin{align*}
		\frac{z^*_{\Lambda}(K;pv)}{z^*_{\Lambda}(K;v)}
		\exp\big(	\bigstar + \mho + \natural				\big). 
	\end{align*}
	However, it is easy to show that $ \bigstar + \mho + \natural	= 0 $.
	Furthermore, for $p=\gq^{-2}=q_3^{-1}$ \eqref{p-shift} and \eqref{6.33} tell us that
	$ \displaystyle \frac{	z^*_{\Omega}(K^{\prime};\gq^{-2}w)		}{z^*_{\Omega}(K^{\prime};w)	} 
	= \frac{1}{K^{\prime\,|\Omega|}}  $.
	Hence, we finally obtain  \eqref{4.62}.


\end{document}